\documentclass[11pt]{article}

\usepackage[T1]{fontenc}
\usepackage[margin=1in]{geometry}

\usepackage{mathtools}
\usepackage{amsmath}
\usepackage{amsfonts}
\usepackage{amssymb}
\usepackage{amsthm}
\usepackage{bm}
\usepackage{graphicx}
\usepackage{algorithm}
\usepackage{algpseudocode}
\usepackage{tikz}
\usepackage{xspace}
\usepackage{thmtools,thm-restate}
\usepackage{natbib}
\usepackage{endnotes}
\usepackage{hyperref}
\usepackage{cleveref}
\usepackage{thm-restate}

\newtheorem{lemma}{Lemma}
\newtheorem{corollary}{Corollary}
\newtheorem{assumption}{Assumption}

\theoremstyle{definition}
\newtheorem{definition}{Definition}

\theoremstyle{remark}
\newtheorem{remark}{Remark}

\newcommand{\abspe}{\mbox{\textsf{Abs-PE}}\xspace}
\newcommand{\posfppe}{\textsf{Pos-FPPE}\xspace} 
\newcommand{\fppe}{\textsf{FPPE}\xspace}

\title{Pacing Equilibria in Abstract Mechanisms}

\author{
\begin{tabular}{ccc}
Salam Afiouni & Christian Kroer  \\[2pt]
Columbia University & Columbia University\\
\texttt{sa4316@columbia.edu} & \texttt{ christian.kroer@columbia.edu} 
\end{tabular}
}

\date{}

\begin{document}

\maketitle

\begin{abstract}
Digital platforms increasingly rely on automated budget-management systems to regulate participation across allocation opportunities. A common tool is multiplicative pacing, which scales buyers' bids so that campaign-level budgets are spent gradually. Although pacing is operationally simple, its aggregate behavior depends on the underlying mechanism. In first-price single-item auction markets, pacing enjoys strong structural and computational properties that often fail in second-price auctions. It is unclear whether these properties extend to richer platform mechanisms.

We develop a unified theory of pacing equilibria across a hierarchy of mechanisms. We first show that, even in first-price position auctions, the market-equilibrium interpretation of the single-item benchmark can fail. We then identify bid-maximizing pay-your-bid mechanisms as a broad class in which pacing equilibria exist, are unique, and admit an Eisenberg-Gale-type convex-program characterization. This implies efficient computation, Pareto-efficiency, and liquid-welfare guarantees, enabling pacing integration into large-scale optimization-based allocation routines.
Finally, for abstract mechanisms satisfying a payment monotonicity condition, we prove equilibrium existence by smoothing discontinuities. With additional conditions and appropriate tie-breaking, we obtain uniqueness, revenue maximality among budget-feasible pacing vectors, and shill-proof implementability, and provide convergent budget-adjustment dynamics for approximate equilibria. Overall, many favorable properties of first-price pacing extend beyond single-item auctions under appropriate mechanism-level conditions.
\end{abstract}

\section{Introduction}\label{sec:Intro}
Digital platforms increasingly manage economic activity through automated allocation systems. Advertising platforms allocate impressions and sponsored positions, recommendation systems allocate user attention, marketplaces allocate visibility, inventory, or service capacity, and matching platforms allocate opportunities among heterogeneous participants. These decisions are made repeatedly, at high frequency, and at large scale. A single advertiser, seller, or service provider may participate in thousands or millions of allocation opportunities over the course of a campaign, while facing a hard budget that limits total expenditure over the relevant horizon.

In such environments, participants typically do not manage each allocation decision directly. An advertiser, for example, may specify a target audience, a value-per-click for relevant opportunities, and an overall campaign budget. The platform's campaign-management system then determines how those inputs are used across many allocation events. In effect, the platform operates a proxy bidder on behalf of the participant: it decides when the participant should enter the market, how aggressively she should bid, and how her budget should be distributed over time. This delegation shifts an important part of market design from the rules of a single auction to the aggregate behavior of the platform's automated budget-management system.

This creates a fundamental operational problem. If the platform allows a budget-constrained buyer to bid her full value whenever an opportunity arrives, the buyer may spend too quickly, exhaust her campaign prematurely, and be absent from later opportunities that may be equally or more valuable. Conversely, if the platform suppresses participation too aggressively, the buyer may fail to spend her budget, reducing both the buyer’s realized value and the platform’s revenue. Budget management is therefore a central component of platform design: the platform must determine how participants are represented across many allocation events, and this representation affects spending, competition, revenue, and welfare.

A widely used approach to this problem is \emph{pacing}. Rather than submitting a buyer’s full bid in every interaction, the platform applies a \emph{pacing multiplier} between zero and one that uniformly scales the buyer’s values or bids across the mechanisms in which she participates. A buyer whose budget would otherwise be exceeded is paced down, while a buyer whose budget is not binding remains unpaced. 

Pacing is used by major advertising platforms~\citep{meta_pacing,linkedinhelp_talentinsights_a422101}, and its managerial appeal is clear: a small number of campaign-level controls can regulate expenditure over a large number of allocation events, while preserving participation in valuable opportunities.
Although pacing is operationally simple, its aggregate effects are not. When many buyers are paced simultaneously, a buyer's spending depends not only on her own multiplier, but also on the paced bids of the other buyers competing for the same opportunities. Budget management is therefore not a collection of independent single-buyer controls. It is a market-level equilibrium problem generated by the interaction between budgets, pacing multipliers, allocation rules, and payment rules. Understanding this equilibrium is essential for predicting whether a pacing system will be stable, whether it will spend budgets effectively, how it will affect platform revenue and market efficiency, and whether it can be implemented at platform scale.

The existing theory provides a sharp benchmark. \citet{conitzer2022pacing, conitzer2022multiplicative} introduce and study \emph{pacing equilibria} in first-price and second-price auction markets. In their model, advertisers participate in a sequence of independent single-item auctions for ad opportunities and are assigned pacing multipliers to manage their budgets over the course of a campaign. Formally, a pacing equilibrium is a vector of pacing multipliers such that, given the induced bids, each buyer satisfies her budget constraint and is paced strictly below one only when necessary to meet that constraint.

The interpretation and properties of pacing equilibria depend critically on the underlying mechanism. In second-price settings, uniform multiplicative pacing has a direct strategic interpretation: under standard models of budget enforcement and deviation, a buyer's best response can be represented by a pacing multiplier, so pacing equilibria coincide with Nash equilibria ~\citep{balseiro2015repeated, balseiro2019learning, babaioff2020non}. In first-price settings, by contrast, optimal bidding generally involves strategic bid shading that is not captured by uniform scaling. Pacing equilibrium is therefore not a best-response concept. Instead, it has a market-equilibrium interpretation: in first-price single-item auctions, pacing equilibria coincide with Fisher-market equilibria and can be computed through an Eisenberg-Gale convex program~\citep{conitzer2022pacing}. 

Pacing equilibria exhibit very different structural and computational properties across mechanisms. When individual opportunities are sold through first-price single-item auctions, pacing equilibria are well behaved: they exist, are essentially unique, admit a convex-program characterization, and satisfy attractive economic properties such as Pareto optimality and shill-proofness \citep{conitzer2022pacing}. In second-price settings, however, pacing equilibria can be non-unique, can violate natural monotonicity and optimality properties, and can be computationally hard to find \citep{conitzer2022multiplicative, chen2023complexity}. This contrast has made first-price pacing an appealing benchmark for automated budget management in large-scale platform markets.

Yet the single-item auction model abstracts from how modern platforms actually allocate opportunities. A platform rarely allocates only one indivisible object at a time. A search page or feed visit may contain a ranked list of sponsored positions with different exposure levels. A cloud, bandwidth, or inventory platform may allocate shared capacity across several buyers. A recommendation or matching platform may allocate opportunities according to scoring rules, feasibility constraints, or combinatorial assignment objectives. Platforms may also use payment rules that are not simply first-price. In all of these settings, the platform may want to use the same pacing primitive: a simple multiplier that regulates each buyer’s participation across many allocation events. But it is no longer clear whether pacing retains the stability, uniqueness, efficiency, revenue, and computational properties that make it attractive in the first-price single-item benchmark.
This raises the central question of the paper:
\begin{quote}
\textit{Do the appealing structural and economic properties that make pacing a reliable budget-management tool in first-price single-item auctions extend to more general mechanisms?}
\end{quote}
We answer this question by developing a theory of pacing equilibria for increasingly general classes of mechanisms.
The paper's central message is that the desirable properties of first-price pacing are not merely artifacts of single-item auctions, but neither do they extend automatically to richer allocation environments. Some properties fail under seemingly modest generalizations, showing that the single-item theory cannot be applied mechanically. Yet, many of the core guarantees can be recovered under appropriate mechanism-level conditions. These conditions identify which features of the platform’s allocation and payment rules allow simple multiplicative pacing to produce predictable, computable, and economically desirable outcomes.

Our analysis is organized around three levels of generality. We first study first-price position auctions, a natural extension of first-price single-item auctions, in which each opportunity contains multiple ranked slots. We then study bid-maximizing pay-your-bid mechanisms, in which the platform selects a feasible allocation that maximizes the sum of paced bids and charges each buyer her paced bid for the allocation she receives. This class includes both first-price single-item and position auctions, and it also captures richer allocation systems in which the platform optimizes over a feasible set of outcomes. Finally, we study fully abstract mechanisms with arbitrary allocation rules, satisfying a monotonicity condition on induced payments, allowing for platform environments that need not be standard auctions.
Taken together, these three classes encompass increasingly general platform allocation systems, from multi-slot advertising auctions to optimization-based allocation mechanisms and, ultimately, to mechanisms incorporating platform-specific payment rules and business constraints.

\subsection{Contributions}
We develop a mechanism-level theory of pacing equilibria for platform markets that use simple budget pacing controls in richer allocation environments. Our results identify when the favorable properties of first-price pacing in single-item auctions -- including existence, uniqueness, tractability, efficiency, revenue guarantees, and implementability -- continue to hold beyond the canonical model. The analysis proceeds through three increasingly general mechanism classes and then studies adaptive budget dynamics for computing pacing equilibria.

\subsubsection{A limitation of the single-item benchmark.}
We begin by analyzing \emph{first-price position auctions}, one of the simplest generalizations of first-price single-item auctions. These mechanisms retain first-price pay-your-bid pricing, but allow for a richer allocation structure: multiple ranked slots are sold in every auction. We show that the market equilibrium interpretation of pacing equilibrium already breaks down in this setting, even under utility models that account for the position-auction structure.

\subsubsection{A convex-program theory for bid-maximizing pay-your-bid mechanisms.}

We identify a broad class of mechanisms, \emph{bid-maximizing pay-your-bid mechanisms}, in which pacing equilibria admit a convex-program characterization and retain the main structural properties of first-price pacing. This class subsumes first-price single-item and position auctions, while allowing the platform to impose richer feasibility constraints on allocations. In \emph{bid-maximizing pay-your-bid mechanisms}, the platform selects a feasible outcome maximizing total paced bid value, and each buyer pays her paced bid for the allocation she receives.

For this class, we establish a tractable characterization of pacing equilibrium. Under mild convexity assumptions on the feasible allocation set and valuations, we prove existence of a pacing equilibrium, characterize the equilibrium pacing vector through a convex program, and give an Eisenberg-Gale-type program that recovers an equilibrium allocation. These results extend the computational tractability of first-price pacing beyond single-item auctions. They also imply
structural and economic guarantees: the equilibrium pacing vector is unique, and the corresponding equilibrium outcomes are Pareto efficient and satisfy additive liquid-welfare guarantees, including an equilibrium-independent bound depending only on budgets and a global bound on valuations.
Beyond these convexity assumptions, we show that any pacing equilibrium of a bid-maximizing pay-your-bid mechanism achieves at least one half of the optimal liquid welfare, and that this guarantee is tight for the class.
Finally, leveraging the convex-program formulation, we show how the PACE dynamics of~\citet{gao2021online} can be applied to compute approximate equilibria using only a bid-maximization oracle, making the approach
compatible with platforms that already implement allocation through large-scale optimization routines.

\subsubsection{Existence and structure in abstract mechanisms.}
We define pacing equilibrium for fully abstract mechanisms with arbitrary allocation and payment rules. Our analysis relies on a global monotonicity property of induced payments: holding a buyer’s bid fixed, lowering her opponents’ bids cannot decrease her payment (Assumption~\ref{ass:pymt-opp-inverse-mnt}). This condition is satisfied by a range of mechanisms, including single-item, multi-unit, and position first-price auctions; pay-your-bid mechanisms with monotone allocation rules; all-pay and entry-fee mechanisms with bid-dependent payments; and capacity-constrained resource-allocation mechanisms such as proportional sharing or priority-based allocation.

A central technical challenge in this abstract setting is that discontinuities in allocation and payment rules can obstruct both equilibrium existence and convergence of budget dynamics. To address this, we introduce a bid-perturbation (smoothing) framework that regularizes induced payments. 
Under this framework, budget feasibility is closed under coordinate-wise maxima, yielding a join-semilattice structure over budget-feasible pacing vectors. This property allows us to characterize equilibrium through a maximal budget-feasible pacing vector, extending a central structural feature of first-price pacing beyond the single-item setting. 
We prove existence of a pacing equilibrium for the original (unsmoothed) mechanism under an appropriate tie-breaking rule, adopting the tie-breaking framework of~\citet{babaioff2020non}. In the smoothed mechanism, we further characterize the equilibrium as the Pareto-dominant (coordinate-wise greatest) budget-feasible pacing vector. Finally, under additional mild conditions on induced payments and appropriate tie-breaking, we obtain stronger guarantees including uniqueness, revenue maximization among budget-feasible pacing vectors, and shill-proofness.

\subsubsection{Spending-based equilibrium computation.}

We provide decentralized budget-adjustment dynamics for computing approximate pacing equilibria in the smoothed model. First, we generalize the dynamics of~\citet{borgs2007dynamics} beyond first-price single-item auctions, establishing convergence of analogous spending-based updates for mechanisms whose payments respond sufficiently to a common scaling of all bids. Second, we analyze no-regret multiplicative-weights updates under a weaker regularity condition on payments, and show that time-averaged iterates satisfy the defining properties of an approximate pacing equilibrium. In both cases, buyers update their multipliers using only realized expenditure, rather than full knowledge of the market or the mechanism.

These dynamics show that the equilibrium guarantees can be supported by decentralized adjustment based on realized spending. Under the structural conditions identified in the paper, pacing can be operated as a scalable feedback system: multipliers respond to spending deviations, and the resulting process moves toward approximate equilibrium behavior. This supports disciplined budget management in large, dynamic~markets.

\subsection{Implications for Platform Design}
Together, our results develop a mechanism-level theory of platform pacing. They characterize when the desirable properties of first-price pacing extend beyond the single-item auction setting. The main conclusion is that pacing inherits its equilibrium properties from the mechanism through which paced bids are implemented. Although pacing is a simple operational tool, its effect depends on how scaled bids are translated into outcomes and payments. Thus, the same pacing rule can behave differently across allocation environments, and the first-price single-item model need not provide a reliable benchmark for richer platform mechanisms.

Our results identify mechanism structures under which pacing equilibria retain desirable structural and economic properties and remain compatible with stable budget management. Bid-maximizing pay-your-bid mechanisms provide one broad class in which pacing remains tractable and compatible with optimization-based allocation, which is particularly relevant for platforms that allocate through large-scale ranking, matching, or constrained optimization routines. More generally, the abstract-mechanism results show how payment monotonicity, smoothing, and appropriate tie-breaking support equilibrium existence in settings with features such as reserves, eligibility thresholds, business rules, or discrete tie-breaking. 

In this sense, the paper treats pacing not only as an auction-theoretic equilibrium concept,  but also a component of platform-level budget-management design. By identifying mechanism-level conditions that support stable, efficient, and implementable pacing equilibria, our results provide guidance for automated budget management in complex platform environments.

\subsection{Additional Related Literature}

Our work builds most directly on the literature on pacing equilibria in auction markets. \citet{conitzer2022pacing,conitzer2022multiplicative}, study pacing for additive valuations over divisible items. They formalize a static game in which budget-constrained bidders choose pacing multipliers and define the notion of \emph{pacing equilibrium} of that game.
For second-price single-item auctions, \citet{conitzer2022multiplicative} introduce \emph{second-price pacing equilibria} (\textsf{SPPE}), showing that they can be non-unique with poor welfare/revenue properties, and that optimizing or computing them is computationally hard~\citep{chen2023complexity}. In contrast, in the first-price setting, \citet{conitzer2022pacing} show that \emph{first-price pacing equilibria} (\textsf{FPPE}) coincide with Fisher-market equilibria and admit efficient computation via the Eisenberg-Gale convex program. We move beyond the single-item setting, investigating which first-price structural properties generalize to position auctions, bid-maximizing pay-your-bid mechanisms, and fully abstract mechanisms.

Our work is also closely related to the broader literature on auctions with budget-constrained bidders and automated budget management. Much of this literature focuses on repeated second-price auction environments~\citep{balseiro2015repeated,balseiro2019learning,gummadi2012repeated,gummadi2013optimal,conitzer2022multiplicative}, where multiplicative pacing is optimal under standard assumptions. While this optimality need not extend to first-price auctions, the simplicity of pacing as a budget-management primitive motivates its study in that setting as well~\citep{conitzer2022pacing, chen2024budget}. \citet{balseiro2023contextual}
show that value shading can yield Bayes-Nash equilibria satisfying budgets in expectation, and \citet{borgs2007dynamics} prove convergence of pacing dynamics in a perturbed first-price model. 
We extend their convergence analysis beyond auction formats to mechanisms whose payments respond sufficiently to common bid scaling, and further establish average convergence of multiplicative-weights dynamics under a weaker regularity condition on payments. 
Multiplicative-scaling equilibria have also been studied beyond standard auction models. \citet{aggarwal2019autobidding} study truthful multi-slot auctions with affine constraints, and \citet{babaioff2020non} consider mechanisms with non-quasi-linear buyers. We adopt a similar level of generality, characterizing pacing equilibria in abstract mechanisms under arbitrary allocation and payment~rules.

Our convex-program characterization also connects to the literature on Fisher markets and the Eisenberg-Gale convex program. Linear Fisher markets provide a canonical model of market clearing with budgets, and the Eisenberg-Gale (EG) convex program gives a tractable convex characterization of equilibrium~\citep{eisenberg1959consensus,eisenberg1961aggregation}. The EG program has motivated a rich algorithmic literature: \citet{devanur2008market} gave the first polynomial-time algorithm via a primal-dual scheme, later improved to strongly polynomial time by~\citet{orlin2010improved}, while~\citet{gao2020first} developed first-order methods for broader utility classes. The EG framework has also been extended to infinite-dimensional settings~\citep{gao2023infinite}. We build on the EG characterization of \fppe from~\citet{conitzer2022pacing}, extending it, under mild convexity assumptions, beyond single-item auctions to any bid-maximizing pay-your-bid mechanism.

A related stream studies welfare guarantees under budget constraints. Liquid welfare has become a standard benchmark in such settings, capping each buyer's contribution by her budget~\citep{dobzinski2014efficiency,azar2017liquid}. Prior work establishes liquid-welfare guarantees for truthful mechanisms, including a 2-approximation for pacing-based pure Nash equilibria~\citep{aggarwal2019autobidding,babaioff2020non}. \citet{balseiro2023contextual} extend such guarantees to Bayes-Nash equilibria in contextual standard auctions. \citet{gaitonde2023budget} obtain matching guarantees for pacing dynamics without convergence, \citet{fikioris2023liquid} extend them under a weaker behavioral assumption, and \citet{baldeschi2026optimal} derive tight type-dependent bounds with heterogeneous autobidding agents. 
Our guarantees hold directly at pacing equilibrium. For
bid-maximizing pay-your-bid mechanisms, we show that every pacing equilibrium achieves at least one half of the optimal liquid welfare. Under the additional convexity assumptions, we also obtain a complementary additive approximation guarantee.

Finally, following \citet{borgs2007dynamics} who use random perturbations to smooth tie-surface discontinuities in first-price auctions, we adopt a multiplicative bid-perturbation (smoothing) framework that induces continuous expected allocations and payments, enabling our structural and convergence analyses. \citet{conitzer2022multiplicative} show that tie-breaking affects equilibrium existence in unperturbed mechanisms, motivating equilibrium notions with randomized allocation, which can be viewed as the limit of expected allocations of a perturbed model as the perturbation vanishes. \citet{babaioff2020non} extend this perspective to general mechanisms. We follow their framework to transfer our existence and structural results from the smoothed mechanism to the original mechanism under appropriate tie-breaking.

\section{Model}\label{sec:model}

Consider a set $N=\{1,\dots,n\}$ of buyers, and a space $X$ of potential outcomes. We assume that $X$ is a nonempty compact subset of $\mathbb R^d$ (for some finite $d$). Buyer $i$ has a hard budget $B_i>0$ and a valuation function $v_i:X\to \mathbb R_{\ge 0}, \,v_i\in V$, where $V$ denotes the space of allowable valuation functions and is assumed to be closed under positive scalar multiplication. We assume valuations are bounded: there exists $0<\bar V<\infty$ such that $0 \le v_i(x) \le \bar V$ for all $i\in N$ and all $x\in X$.
We write $B=(B_1,\dots,B_n)$ and $v=(v_1,\dots,v_n)$ for the profiles of budgets and valuations, respectively. 
For notational convenience, let $x_i$ denote the component of outcome $x\in X$ that is payoff-relevant to buyer $i$, and assume that $v_i$ depends only on $x_i$, so that $v_i(x)=v_i(x_i)$ for all $i$ and all $x\in X$.

A mechanism $\mathcal M=(x,p)$ consists of an allocation rule $x: V^n\to X$ and a payment rule $p: V^n\to \mathbb{R}^n_{\ge 0}$ that map a profile of reports to an outcome and a profile of payments. We interpret a report $b_i\in V$ as buyer $i$'s \emph{bid (reported value) function}, and write $b=(b_1,\dots,b_n)\in V^n$ for the bid profile. For $b_i,b_i'\in V$, write $b_i\le b_i'$ if $b_i(z)\le b_i'(z)$ for all $z$ in the domain of $b_i$, and we extend this order coordinate-wise to $V^n$.

To control budget expenditure, each buyer $i\in N$ is associated with a \emph{pacing multiplier} $\alpha_i\in[0,1]$ that scales her valuation uniformly; by closure of $V$ under positive scalar multiplication, the \emph{paced bid} $b_i:=\alpha_i v_i$ also lies in $V$. Let $b=\alpha\!\cdot \!v$ denote the resulting bid profile, where the product is componentwise. 
Given $b$, the mechanism $\mathcal M=(x,p)$ selects an outcome $x(b)\in X$ and payments $p(b)\in \mathbb R_{\ge 0}^n$, such that buyer $i$ receives an outcome $x_i(b)$ and pays $p_i(b)$. 

We assume that the mechanism $\mathcal M$ is individually rational, that is $p_i(b)\le b_i\bigl(x_i(b)\bigr)$ for all $i\in N$ and all $b\in V^n$, and that it is \emph{invariant to zero bids}: for any subset of buyers $Z\subseteq N,\,|Z|<n$, if buyers in $Z$ submit zero bids, then the outcome and payments for buyers in $N\setminus Z$ coincide between running $\mathcal M$ on $N$ and running $\mathcal M$ on the reduced buyer set $N\setminus Z$ with bid profile $b_{-Z}$\endnote{When restricting to a buyer subset, we view $\mathcal M$ as \emph{instantiated on} that subset with the corresponding bid profile.}. We further assume that a buyer's payment varies continuously as others' bids vanish\endnote{This rules out mechanisms that treat exact zero bids as special case, e.g., allocating an item to buyer $i$ if and only if buyer $j$'s bid is exactly zero.}: for any subset $Z\subseteq N$, any $i\notin Z$, and any sequence $(b_{-Z}^t, b_Z^t) \to (b_{-Z}, 0^Z)$ in the bounded bid region, $p_i(b_{-Z}^t, b_Z^t) \to p_i(b_{-Z}, 0^Z)$. 

With the model in hand, we now define the equilibrium concept which will be the main object of study in~this~work.
\begin{definition}[Budget-feasible pacing vector]\label{def:budget-feasible}
    A vector of pacing multipliers $\alpha\in[0,1]^n$ is budget-feasible for $\mathcal M$ if, under the induced bid profile $b=\alpha\cdot v$, each buyer satisfies her budget constraint, that is, $p_i(b)\le B_i$ for all $i\in N$.
\end{definition}

We are interested in solutions where no buyer is \emph{unnecessarily paced}. This captures that the platform does not want to pace a buyer that is not spending the full budget: if a buyer would not exhaust her budget when bidding her true value, then her pacing multiplier should be exactly one.

\begin{definition}[Pacing equilibrium for abstract mechanism]
    A vector of pacing multipliers $\alpha~\in~[0,1]^n$ constitutes a Pacing Equilibrium of the Abstract mechanism $\mathcal M$ (\textsf{Abs-PE}) if it is budget-feasible for $\mathcal M$  (\Cref{def:budget-feasible}) and satisfies the additional condition that: 
    if $p_i(b) <B_i$, then $\alpha_i =1$ for each buyer $i \in N.$
    This condition is referred to as the ``no-unnecessary-pacing'' condition.
\end{definition}

The model is broad enough to capture several platform allocation environments. We highlight a few examples.
\begin{itemize}
    \item \emph{Auctions with additive valuations over divisible goods.} Consider $m$ goods. An outcome specifies for each buyer $i$ a vector $x_i = (x_{ij})_{j\in[m]}$ with $x_{ij}\in[0,1]$ denoting the fraction of good $j$ allocated to $i$. Valuations are additive: $v_i(x_i) = \sum_{j=1}^mv_{ij}\, x_{ij},$ where $v_{ij}\ge 0$ is $i$'s value per unit of good $j$.  
    This class includes first-price single-item and position auctions as special cases. 

     \item \emph{Resource allocation under shared capacity.} Consider $m$ shared resources, each with capacity $C_j>0$. An outcome $x\in X$ specifies for each buyer $i$ a nonnegative consumption vector $x_i=(x_{ij})_{j\in[m]}$, subject to the capacity constraints $\sum_{i=1}^n x_{ij} \le C_j \, \forall j\in[m].$ Each buyer has a valuation function $v_i(x_i)$ over their allocated resources. This captures settings such as cloud computing, bandwidth allocation, or capacity-constrained inventory allocation. 

     \item \emph{Bundled allocation.}
    The mechanism selects for each buyer $i$ a bundle $x_i$ from a feasible family of bundles. The buyer's valuation $v_i(x_i)$ is over the bundle and need not decompose across its elements.  
    This captures settings where the value of receiving a collection of resources depends on the combination~as~a~whole.
\end{itemize}
\section{Position Auctions and the Market-Equilibrium Interpretation}
\label{sec:posfppe} 
Modern search and display advertising platforms routinely allocate multiple sponsored positions in each interaction, offering buyers multiple heterogeneous slots such as ranked sponsored listings or positions in a social-media feed.
To capture this setting, we study pacing equilibria in \emph{first-price position auction markets}, where each auction allocates a ranked sequence of (ad) slots to buyers based on their bids. Pricing remains pay-your-bid, as in the first-price single-item model, but the allocation rule now assigns multiple heterogeneous positions. We show that this richer allocation structure is enough to break the market-equilibrium interpretation of pacing established in the single-item setting.

\paragraph{Model.}
Formally, let $N=\{1,\dots,n\}$ be the set of buyers (advertisers) and $M=\{1,\dots,m\}$ the set of independent auctions (e.g., user queries). Each auction $j\in M$ offers a ranked list $S=\{1,\dots,s\}$ of ad slots with slot qualities $q_1^j \ge q_2^j \ge \cdots \ge q_s^j \ge 0$, representing advertiser-independent click probabilities (we assume that the slots are shown in ranked order). For auction $j\in M$, buyer $i$'s value from receiving slot $k$ is $v_{ijk} = v_{ij}\, q_k^j$, where $v_{ij}$ is her value per click in that auction\endnote{The value $v_{ij}$ may already incorporate advertiser-specific click-through rate; so in particular, it could be $v_{ijk} = v_{ij}'.CTR_{ij}.q_k^j$ is their actual value per click, and $CTR_{ij}$ is the likelihood that the user clicks on ad $i$ in auction $j$}. Values are additive across slots and auctions, and payments are first-price.
We assume that the goods  are sold through independent position auctions. 
We also assume without loss of generality that the number of slots per auction is at least the number of bidders, i.e. $s\geq n$\endnote{Otherwise we can reduce to $s=n$ by adding $n-s$ dummy slots with quality score zero if $s<n$, or restricting to the top $n$ slots if $s>n$.}. 
An allocation is given by $x\in[0,1]^{n\times m\times s}$, where $x_{ijk}$ is the fraction of slot $k$ assigned to buyer $i$ in auction $j$. Given paced bids $(\alpha_i v_{ij})$, an allocation in auction $j$ is determined by sorting buyers by their bids and assigning slots greedily in decreasing order; ties are split fractionally among the tied buyers across the corresponding slots such that each buyer receives one slot per auction. 

This position-auction model is a special case of the abstract mechanism in Section~\ref{sec:model}. We now state the corresponding pacing equilibrium concept explicitly.

\begin{definition}[First-price pacing equilibrium for position auctions (\textsf{Pos-FPPE})]
A vector of pacing multipliers $\alpha \in [0,1]^n$ forms a 
pacing equilibrium for a first-price position auction mechanism if, under the induced bids $b_{ij} = \alpha_i v_{ij}$ and the resulting allocation and prices $(x,p)$ generated by the mechanism, the following conditions hold:
\begin{itemize}
        \item \emph{(Prices).} For every auction $j$, unit price for slot $k$: $p_{jk} = q_k^j \alpha_h v_{hj}$ if buyer $h$ is (possibly tied for) the $k^\text{th}$ highest bid in auction $j$.
        \item \emph{(Allocation to highest bids).} If $x_{ijk}>0$ then $\alpha_iv_{ij}$ is the $k^\text{th}$ highest paced bid in auction $j$.
        \item \emph{(Slots are fully allocated).} If $p_{jk}>0$ then for each slot $k$ in auction $j$, $\sum _{i\in N} x_{ijk}=1$.
        \item \emph{(Each buyer receives exactly one slot per auction).} $\sum_{k \in S}x_{ijk}=1$ for any buyer $i$ in auction~$j$.
        \item \emph{(Budget-feasibility).} For each buyer $i \in N$: $\sum_{j\in M} \sum_{k\in S} p_{jk}x_{ijk}\leq B_i$.
        \item \emph{(No-unnecessary-pacing).} If $\sum_{j\in M} \sum_{k\in S} p_{jk} x_{ijk} <B_i$, then $\alpha_i =1$ for each buyer $i \in N$.
    \end{itemize}
\end{definition}

\subsection{Relationship to Market Equilibrium}
\citet{conitzer2022pacing} show that in first-price single-item auctions, pacing equilibria admit a clean market-equilibrium interpretation. We investigate whether this correspondence extends to first-price position auctions, where the outcome space exhibits additional structure.

\begin{definition}
    A pair $(x,p)$ of feasible allocations and prices forms a market equilibrium if the following properties hold:
    \begin{enumerate}
        \item \emph{(Buyer optimality).} Each buyer maximizes her utility under prevailing prices, subject to her budget constraint, that is, for all $i \in N$,
        \(
        x_i \in \arg\max_{x_i\ge 0}\;\bigl\{ u_i(x_i,p): \sum_{j,k} p_{jk} x_{ijk}\le B_i\bigr\},
        \)
        where $u_i:~\mathbb R^{m\times s}\times~\mathbb R^{m\times s}\rightarrow~\mathbb R~\cup~\{-\infty\}$ denotes buyer $i$’s utility function under prices $p$ and outcome $x_i$. This means that every bidder is exactly allocated her demand set.
 
        \item \emph{(Feasibility and market clearing).}  The total allocation does not exceed supply, and every slot with positive price is fully allocated, that is, for all $k\in S, j\in M$, $\sum_i x_{ijk}\le 1$ and $p_{jk}>0 \Rightarrow \sum_i x_{ijk} = 1$.
    \end{enumerate}
\end{definition}

As a first step, we consider standard quasilinear utilities for the buyers: \(u_i(x_i,p)=\sum_{j,k}(v_{ijk}-p_{jk})\,x_{ijk}\). \Cref{thm:no-me-1} shows that a \posfppe outcome may fail to form a market equilibrium: the multi-slot structure and rank-by-bid pricing can yield allocations outside buyers' demand sets.

\begin{restatable}{theorem}{MEnonEquiv}\label{thm:no-me-1}
    Under quasilinear utilities,  a \posfppe outcome $(x,p)$ may fail to be a market equilibrium. In particular, this occurs whenever some buyer receives some fraction of a positive-quality slot while a lower positive-quality slot is priced at a strictly smaller nonzero paced bid. 
\end{restatable}

The condition in \Cref{thm:no-me-1} is mild and is generally met in any instance with at least two positive-quality slots and two distinct positive paced bids. Thus, except for degenerate cases (e.g., all positive paced bids are tied, or there is only one positive-quality slot), the \posfppe outcome cannot coincide with a market equilibrium under standard quasilinear utilities.

\citet{conitzer2022multiplicative} show that \textsf{SPPE} correspond to a refinement of \emph{supply-aware} market equilibria, where buyers are aware of the supplies of each item and choose their demand set accordingly.
This motivates a \emph{position-aware} analogue, in which buyers maximize quasilinear utility over budget-feasible allocations satisfying the position-auction constraint of at most one slot per auction. Intuitively, position-aware buyers understand the market structure and restrict attention to bundles containing (at most) one slot from every auction.
We show that even under this refinement, \posfppe outcomes need not constitute market equilibria.

We encode the one-slot-per-auction feasibility constraint from the \posfppe model into each buyer's quasilinear demand problem. 
Equivalently, we define the \emph{position-aware quasilinear utility}~as
\[
u_i(x_i,p)=
\begin{cases}
\sum_{j,k}\bigl(v_{ijk}-p_{jk}\bigr)\,x_{ijk}, 
& \text{if } \sum_{k} x_{ijk}\le 1 \forall j,\\
-\infty, 
& \text{otherwise}.
\end{cases}
\]

\begin{restatable}{proposition}{posMEnonEqui}\label{prop:no-me-pos}
    There exist instances in which the allocation and prices induced by a \posfppe do not constitute a position-aware market equilibrium.
\end{restatable}

The results in~\Cref{thm:no-me-1} and \Cref{prop:no-me-pos} show that moving from a single advertising slot to multiple ranked positions is already sufficient to break the market-equilibrium interpretation of first-price pacing. Thus, the desirable properties of the single-item benchmark do not automatically extend to richer advertising environments, motivating the search for more general mechanism-level conditions.
The remainder of the paper identifies mechanism-level conditions that restore structural and computational guarantees for broader classes of mechanisms that include first-price single-item and position auctions.

\section{Bid-Maximizing Pay-Your-Bid Mechanisms}\label{sec:bid-max-pyb}

Many digital platforms allocate opportunities by solving large-scale constrained optimization problems rather than through simple auctions. Examples include matching markets, cloud resource allocation, capacity-constrained assignment, and other optimization-based allocation systems. To capture this common operational structure, we study pacing equilibria in \emph{bid-maximizing pay-your-bid mechanisms}, which select an outcome maximizing total paced bid value and charge each buyer her paced bid for the allocation she receives. This class includes first-price single-item auctions and first-price position auctions as special cases, and the corresponding pacing equilibria coincide with \fppe and \textsf{Pos-FPPE}, respectively.

We study bid-maximizing pay-your-bid mechanisms under the general mechanism formulation introduced in~\Cref{sec:model}. Let $\mathcal M=(x,p)$ define a mechanism with outcome rule $x$ and payment rule $p$, which map each bid profile to outcomes and payments. 
For this class of mechanisms, we show that pacing equilibria are not only guaranteed to exist, but also admit strong structural and algorithmic properties. Specifically, the equilibrium pacing vector is unique, can be computed via a convex program, and the resulting allocation is Pareto efficient and satisfies liquid welfare guarantees.

\subsection{\abspe via Convex Programming}\label{sec:convex-prog}
We first characterize \abspe in bid-maximizing pay-your-bid mechanisms via convex optimization under mild assumptions on the allocation and valuation spaces. 

\begin{restatable}[\abspe via convex program]{theorem}{PYBConvexProgram}\label{thm:convex-program}
    Assume $X$ is nonempty, compact, and convex, and each $v_i$ is affine in $x_i$.
    Consider a mechanism $\mathcal M=(x,p)$ which, for bids $b=\alpha\!\cdot\!v$, selects an element of $\arg\max_{x\in X}\sum_i \alpha_i v_i(x_i)$. 
    Let $\alpha^*$ solve the convex program
    \begin{align*}
    \min_{\alpha\le 1}\ \max_{x\in X}\sum_i \alpha_i v_i(x_i)-\sum_i B_i\log\alpha_i.\tag{CP1}
    \end{align*}
    Then there exists $x^*\in \arg\max_{x\in X}\sum_i\alpha_i v_i(x_i)$ such that, with $p_i^*=\alpha_i^* v_i(x_i^*)$, $\alpha^*$ is an \abspe of $\mathcal M$ with outcome $x^*$ and payments $p^*$.
\end{restatable}

\begin{remark}
    If $X$ is non-convex but the mechanism is allowed to randomize over bid-maximizing allocations, we can instead work in the space of randomized allocations. Applying~\Cref{thm:convex-program} to this randomized allocation space yields an \abspe whose outcome is a distribution over allocations in $X$ supported on bid-maximizing allocations in the original space $X$.
\end{remark}

\Cref{thm:convex-program} shows that the Eisenberg-Gale convex-program characterization of first-price pacing extends well beyond the single-item auction setting. In particular, the same convex optimization framework applies to bid-maximizing pay-your-bid mechanisms whose allocation rule is determined by solving a matching, capacity-constrained assignment, or other optimization problem. As a result, pacing equilibria remain efficiently computable both through the convex-program characterization and, operationally, through the PACE-style dynamics developed below.

The convex program (CP1) in \Cref{thm:convex-program} not only computes an \textsf{Abs-PE}, but also characterizes the equilibrium pacing vector. \Cref{lem:abspetoCP} shows that every \abspe corresponds to an optimal solution of (CP1). Since (CP1) admits a unique minimizer, the equilibrium pacing vector is unique.

\begin{restatable}{lemma}{PYBAbspeToCP}\label{lem:abspetoCP}
    Suppose the assumptions of~\Cref{thm:convex-program} hold. Let $\alpha$ be an \abspe of $\mathcal M$ and $x(\alpha)\in\arg\max_{x\in X}\sum_i \alpha_i v_i(x_i)$ be the induced outcome. Then $\alpha$ is an optimal solution of \emph{(CP1)}.
\end{restatable}

\begin{corollary}[Uniqueness of \textsf{Abs-PE}]\label{cor:unique-absp}
    Since $\max_{x\in X}\sum_i \alpha_i v_i(x_i)$ is convex and the regularizer $-\sum_i B_i\log\alpha_i$ is strictly convex for $\alpha>0$, the objective in \emph{(CP1)} is strictly convex on $(0,1]^n$ and thus has a unique minimizer $\alpha^*$. Hence bid-maximizing pay-your-bid mechanisms admit a unique~\textsf{Abs-PE}.
\end{corollary}

From a platform-design perspective, uniqueness ensures that budget management is predictable: the equilibrium pacing multipliers $\{\alpha_i^*\}_{i\in N}$ are uniquely determined. The allocation itself, however, need not be unique and may depend on tie-breaking among bid-maximizing outcomes. Nevertheless, all such outcomes generate the same total revenue since, for any maximizing outcome $x(\alpha^*)$, $\sum_i p_i= \sum_i \alpha_i^* v_i\bigl(x_i(\alpha^*)\bigr) = \max_{x\in X}\sum_i \alpha_i^* v_i(x_i)$.
Moreover, we show that, by convex duality, a maximizing allocation at $\alpha^*$ is computable via an Eisenberg-Gale-type convex program over $X$, generalizing the formulation of~\citet{cole2017convex} for Fisher markets with quasilinear utilities. 

\begin{restatable}[EG-type solution]{proposition}{PYBdualEG}\label{prop:eg-dual}
    Suppose $\mathcal M=(x,p)$ is a bid-maximizing pay-your-bid mechanism that satisfies the assumptions of~\Cref{thm:convex-program}. Then an optimal allocation $x^* = x(\alpha^*\!\cdot\!v)$ at an \abspe $\alpha^*$ of $\mathcal M$ can be obtained by solving the following convex program:
    \begin{align*}
      \max_{\delta_i\in[0,B_i],u,x}\quad \tag{CP2}
        & \sum_i \bigl( B_i \log u_i - \delta_i \bigr) \\ 
        & x\in X,\\
        & u_i \le v_i(x_i) + \delta_i \quad \forall i.
    \end{align*}
    Here $\delta_i$ can be interpreted as leftover budget for buyer $i$: when $\delta_i=0$ buyer $i$ spends her budget entirely and the objective is equivalent to $\sum_i B_i\log v_i(x_i)$; when $\delta_i>0$, buyer $i$ is below budget and the objective is linear in $v_i(x_i)$.
\end{restatable}

\begin{remark}
    Algorithmically, the tractability of pacing equilibrium in the bid-maximizing pay-your-bid mechanisms setting boils down to managing $X$: if we have efficient routines to optimize over~$X$, then the convex program is tractable. 
\end{remark}

The convex-program characterization also provides a natural framework for establishing Pareto efficiency.
Let $U$ denote the vector of buyer utilities $U_i:=v_i(x_i)+\delta_i$, and let $R$ denote seller revenue $R:=\sum_i(B_i-\delta_i)$. Since $\log(\cdot)$ is strictly increasing (and $U_i>0$ at feasible solutions we consider), any Pareto improvement in $(U,R)$ would strictly improve the objective of (CP2), contradicting optimality. Hence, every optimal solution of (CP2) is Pareto optimal in $(U,R)$.

\begin{corollary}[Pareto efficiency]\label{cor:pareto-eff-abspe}
Let $(x^*,\delta^*,u^*)$ be an optimal solution of \emph{(CP2)}. Then the outcome $(x^*,\delta^*)$ is Pareto efficient with respect to buyer utilities $(U_i)_{i\in N}$ and seller utility $R$.
\end{corollary}

\paragraph{PACE-style dynamics.}
Our EG-type characterization suggests a direct generalization of the PACE updates of~\citet{gao2021online} for computing \abspe. 
Let $F(\alpha)=\max_{x\in X}\sum_i \alpha_i v_i(x_i)-\sum_i B_i\log\alpha_i$ over $\alpha\in(0,1]^n$ be the inner problem in~\Cref{thm:convex-program}. 
Given an iterate $\alpha^t$, let $x^t\in \arg\max_{x\in X} \sum_i \alpha_i^t v_i(x_i)$, and set $g^t=(v_i(x_i^t))_{i\in N}.$ Since $\phi(\alpha):=\max_{x\in X}\sum_i \alpha_i v_i(x_i)$ is a pointwise maximum of linear functions in $\alpha$, then $g^t\in \partial \phi(\alpha^t)$ is a valid subgradient.
Moreover, under bounded valuations $0\le v_i(x_i)\le \bar V$, optimal multipliers satisfy $\alpha_i^*\in[\ell_i,1]$ with $\ell_i:=\min\{1,B_i/\bar V\}$ (as shown in the proof of~\Cref{prop:eg-dual}).
Thus, one can apply standard dual-averaging/mirror-descent methods to minimize $F$ over $\prod_i[\ell_i,1]$ using only a bid-maximization oracle, yielding closed-form PACE-like updates of the form $\alpha_i^{t+1}=\prod_{[\ell_i,1]}\!\left(\frac{B_i}{\bar g_i^t}\right)$, where $\bar g_i^t$ denotes the dual average (time-averaged utilities) for buyer $i$. Convergence rates then follow from standard convex optimization~analysis.

\subsection{Liquid Welfare Guarantees} 
The previous subsection establishes that bid-maximizing pay-your-bid mechanisms admit a well-structured equilibrium that can be computed efficiently and yields Pareto-efficient outcomes. Beyond computational tractability and structural guarantees, a natural question is whether the resulting allocations are economically efficient. For platforms operating with budget-constrained participants, liquid welfare provides a standard benchmark for evaluating allocation quality. We establish liquid welfare guarantees for pacing equilibria in this class.
Recall that the liquid welfare of a buyer is the minimum of the value she receives and her budget. More precisely, for any buyer $i$, the liquid welfare she achieves from allocation $x_i$ is $\mathrm {LW}(x_i) = \min\{v_i(x_i), B_i\}$. Let $\mathrm {LW}(x) = \sum_i \mathrm{LW}(x_i)$ denote the total liquid welfare.
We establish complementary multiplicative and additive guarantees for \abspe allocations in bid-maximizing pay-your-bid mechanisms. In particular, equilibrium revenue coincides with liquid welfare and every \abspe achieves at least half of the optimal liquid welfare. Under the convexity assumptions of \Cref{thm:convex-program}, we further obtain an additive guarantee in terms of the equilibrium pacing multipliers, together with an equilibrium-independent version depending only on budgets and the valuation bound $\bar V$. 

\begin{restatable}[LW guarantees for \abspe allocations]{theorem}{abspeLW}  \label{thm:abspe-lw}
Consider a bid-maximizing pay-your-bid mechanism with \abspe $\alpha^*$, and let $x^*$ be the corresponding allocation.  Then the following guarantees hold:
\begin{enumerate}
    \item Equilibrium revenue equals liquid welfare, and $\mathrm{LW}(x^*) \ge \frac12\mathrm{OPT}_{\mathrm{LW}}.$

    \item If, in addition, the assumptions of \Cref{thm:convex-program} hold, then $\mathrm{LW}(x^*) \ge \mathrm{OPT}_{\mathrm{LW}} - \sum_{i\in N} B_i\log\frac1{\alpha_i^*}.$
    Combining this with part~(1),
    \(
    \mathrm{LW}(x^*) \ge
    \mathrm{OPT}_{\mathrm{LW}} -
    \min\left\{
        \frac12\mathrm{OPT}_{\mathrm{LW}},
        \sum_{i\in N} B_i\log\frac1{\alpha_i^*}
    \right\}.
    \)
\end{enumerate}
\end{restatable}

The additive guarantee in \Cref{thm:abspe-lw} depends on the equilibrium pacing multipliers. Using the lower bound $\alpha_i^*\ge\min\left\{1,\frac{B_i}{\bar V}\right\}\; \forall i\in N$ established in the proof of \Cref{prop:eg-dual} gives the following equilibrium-independent additive bound.

\begin{corollary}[Equilibrium-independent additive LW guarantee]\label{corr:lw-abspe}
Under the assumptions of \Cref{thm:convex-program}, $\mathrm{LW}(x^*)\ge \mathrm{OPT}_{\mathrm{LW}} -  \sum_{i\in N} B_i \left[\log\frac{\bar V}{B_i}\right]_+.$
\end{corollary}

Moreover, the multiplicative guarantee in \Cref{thm:abspe-lw} is tight for the class of bid-maximizing pay-your-bid mechanisms. Indeed, first-price single-item auctions are a special case of this class, and the following construction shows that the liquid welfare ratio can be made arbitrarily close to~$1/2$.

\begin{restatable}{lemma}{PYPfppeBound}\label{lem:lw-tight-exp}
Fix an integer $m\ge 2$ and let $\varepsilon:=1/m$. There is an instance of first-price single-item auction markets for which an \fppe allocation $x^*$ satisfies
\(
\frac{\mathrm{LW}(x^*)}{\max_{y\in X}\mathrm{LW}(y)} = \frac{1}{2-\varepsilon}.
\)
\end{restatable}

\section{General Abstract Mechanisms}\label{sec:abstract-mech}

In this section, we analyze pacing in general abstract mechanisms, specified only by an allocation rule and a payment rule defined over bid profiles. This abstraction encompasses platform allocation systems that need not fit standard auction models, including mechanisms with reserves, eligibility constraints, and other platform-specific business rules. Working at this level of generality, we identify monotonicity conditions on payments under which an \abspe exists and admits strong structural and economic properties. We then study equilibrium computation via adaptive bidding dynamics.

\paragraph{Assumptions.}
Our abstract analysis relies on monotonicity conditions on the payment rule. We distinguish a \emph{single standing assumption}, which underlies all existence and structural results, from \emph{additional assumptions} that are invoked only for stronger guarantees or for computational convergence.

\textbf{Standing assumption.} Throughout Sections \ref{subsec:existence}--\ref{subsec:computation}, we maintain Assumption~\ref{ass:pymt-opp-inverse-mnt} (\emph{opponent inverse monotonicity}): holding buyer $i$’s bid fixed, lowering the other buyers’ bids weakly increases $i$’s payment. Intuitively, decreasing bids can only make a buyer pay weakly more, usually because that causes her to win more goods.

\textbf{Additional structural assumptions.} To obtain sharper economic guarantees, we impose two further conditions. Assumption~\ref{ass:collec-pymt-mnt} (\emph{collective strict monotonicity}) is used only to establish uniqueness of the pacing equilibrium. If, in addition, Assumption~\ref{ass:monotone-rev} (\emph{monotone revenue}) holds, then the unique equilibrium pacing vector maximizes revenue among all budget-feasible pacing vectors. 

\textbf{Computational regularity.} Finally, in the computational section (\Cref{subsec:computation}), we impose two algorithm-specific regularity conditions. Assumption~\ref{ass:homog} (\emph{common-scale responsiveness}) is required for the budget dynamics of~\citet{borgs2007dynamics} (\Cref{alg:borgs}), while Assumption~\ref{ass:mult-OBM} (\emph{own-bid responsiveness}) is used for the multiplicative-weights updates (\Cref{alg:MW}). These assumptions play no role in the existence or structural~results.

Unless stated otherwise, all assumptions are imposed on the payment rule of the \emph{original} mechanism; for tie-breaking results, we assume the stated properties for the selected payment rule at tie-breaking.

\paragraph{Examples.}
Several natural mechanism families satisfy the assumptions used for our equilibrium guarantees. In particular, Assumptions~\ref{ass:pymt-opp-inverse-mnt}--\ref{ass:monotone-rev} hold for mechanisms with separable strictly increasing payment rules $p_i(b)=c_i(b_i)$, including all-pay formats (e.g., all-pay auctions with or without reserves, weighted all-pay contests and Tullock-style contests) and bid-dependent entry-fee formats in which agents pay regardless of allocation.
They also hold for standard rank-based first-price mechanisms, including single-item, multi-unit, and position auctions (as well as variants with reserve prices), and for standard bid-priority capacity-constrained mechanisms with first-price payments, such as priority-based resource allocation under shared capacity.
More generally, Assumption~\ref{ass:pymt-opp-inverse-mnt} holds for mechanisms with payments \(p_i(b)=b_i^\gamma x_i(b), \, \gamma>0\), whenever \(x_i(\cdot)\) is monotone in the natural sense (i.e., weakly decreasing in opponents' bids and weakly increasing in own bid). This class includes pay-your-bid mechanisms as the special case \(\gamma=1\), such as (weighted) proportional-sharing mechanisms.

The additional conditions used only for convergence of the learning dynamics are also satisfied by standard payment formats. For example, Assumption~\ref{ass:homog} holds for mechanisms with positively homogeneous payment rules, including standard first-price auctions, and Assumption~\ref{ass:mult-OBM} holds additionally for mechanisms whose payments have positive elasticity in own bid, such as power payments $p_i(b)=b_i^\mu,\,\mu>0$, or payment rules \(p_i(b)=b_i^\gamma x_i(b),\, \gamma>0\) with \(x_i(\cdot)\) weakly increasing in own bid.

\subsection{Existence of \abspe}
\label{subsec:existence}
Our main result establishes that any abstract mechanism with an arbitrary allocation rule and a payment rule satisfying Assumption~\ref{ass:pymt-opp-inverse-mnt} admits a pacing equilibrium under appropriate tie-breaking.
We establish existence in three steps, following a lattice-based approach inspired by the structure of \textsf{FPPE}. First, we show that under Assumption~\ref{ass:pymt-opp-inverse-mnt}, the set of budget-feasible multipliers is closed under coordinate-wise maxima, and therefore forms a join-semilattice. 
Second, since discontinuities in general mechanisms can make existence sensitive to tie-breaking, we introduce a smoothing framework that ensures expected payments and allocations are continuous in each multiplier. This yields a coordinate-wise greatest budget-feasible multiplier, which we show satisfies the no-unnecessary-pacing condition and is therefore an \abspe of the smoothed mechanism. Finally, we transfer this existence result back to the original mechanism via a limit-based tie-breaking construction  as the perturbation vanishes.
The following subsections develop this argument in detail and establish the~result.

\subsubsection{Join-semilattice property}
\label{sec:join-semilattice}
Consider the set of budget-feasible pacing multipliers $\mathcal A~:=~\{\alpha~\in~[0,1]^n: p_i(\alpha\cdot v)\le B_i \ \forall i\}$, ordered coordinate-wise.  
Under the following standing monotonicity condition on payments, we show that $\mathcal A$ is closed under the join operation $\vee$ given by component-wise maximum; equivalently, $(\mathcal A,\vee)$ forms a join-semilattice.

\begin{assumption}[Opponent inverse monotonicity]
\label{ass:pymt-opp-inverse-mnt}
     Consider two profiles $b$ and $b'$ where $b_i=b_i'$ and $b_{-i}'\leq~b_{-i}$ coordinate‐wise, then $p_i(b_i,b_{-i}') \ge p_i(b_i,b_{-i})$.
\end{assumption}

Assumption~\ref{ass:pymt-opp-inverse-mnt} states that, when holding buyer $i$’s bid fixed, decreasing opponents’ bids weakly increases $i$’s payment. Intuitively, facing weaker competition can only make buyer $i$’s payment weakly higher.
This assumption underlies our structural and algorithmic results and is maintained throughout Sections~\ref{subsec:existence}-\ref{subsec:computation}.
As we show next, this condition implies that budget feasibility is preserved under component-wise maxima.

\begin{lemma}[Join-closure] \label{lem:join-closure}
    Given any two budget-feasible pacing vectors $\alpha^1$ and $\alpha^2$, their componentwise maximum $\alpha^* =\max(\alpha^1, \alpha^2)$ is budget-feasible. 
\end{lemma}

The proof follows directly from observing that in either $\alpha^1$ or $\alpha^2$, buyer $i$ had the exact same multiplier -- say it was in $\alpha^d$ for $d \in \{1,2\}$ (breaking ties arbitrarily toward 1). Since $\alpha_i^*=\alpha_i^d$ and $\alpha_{-i}^* \ge \alpha_{-i}^d$, then Assumption~\ref{ass:pymt-opp-inverse-mnt} ensures that $p_i(\alpha_i^*v_i,\alpha_{-i}^*v_{-i}) \leq p_i(\alpha_i^dv_i,\alpha_{-i}^dv_{-i})$. Since $\alpha^d$ is budget-feasible, then so is $\alpha^*$.

\subsubsection{Mechanism perturbation} \label{section:mechanism-perturb}
The join-semilattice structure suggests constructing an equilibrium pacing vector as a coordinate-wise greatest budget feasible vector, mirroring the maximum-element characterization of \fppe in the first-price single-item setting. However, join-closure alone is not sufficient: allocations and payments in $\mathcal M$ can be discontinuous in bids (e.g., due to tie-breaking among equal bids, changes in the set of binding constraints, integrality requirements), so $\mathcal A$ need not be closed and coordinate-wise suprema may fail to be attained. Consequently, the maximum-element argument can break down, and existence may depend on arbitrary resolution choices such as tie-breaking.

To obtain a robust existence statement, we internalize the resolution of discontinuities via small, atomless perturbations of bids. This smoothing yields continuity of the expected allocation and payment functions in each buyer's pacing multiplier, restoring the regularity needed for the maximum-element argument.

\paragraph{Smoothing.}
We construct a perturbed version of the abstract mechanism $\mathcal M$, in which each buyer’s bid is slightly randomized prior to computing outcomes and payments.
Formally, for $\delta\in(0,1)$, we define the perturbed mechanism $\tilde{\mathcal M}_\delta$ by independently drawing $\epsilon_i\sim \mathrm{Unif}[1-\delta,1]$ for each buyer $i$ and randomizing the buyer’s bid via a multiplicative perturbation of their valuation function: we replace buyer $i$'s original value $v_i(x)$ with a perturbed value $\tilde v_i(x)\coloneqq\epsilon_iv_i(x)$ for every outcome $x\in X$. Note that $\tilde v_i\in V$ by closure of $V$ under positive scaling. 

In the following, we show that the randomized perturbation ensures that both the expected allocation and payment functions become continuous functions in each buyer's own pacing multiplier.

\begin{restatable}[Continuous outcomes under smoothing]{proposition}{ContinuousPayments}\label{prop:cont-alloc-pymt}
    Fix $\delta\in (0,1)$, and let $\epsilon=(\epsilon_1,\dots,\epsilon_n)$ be a vector of i.i.d.\ $\mathrm{Unif}[1-\delta,1]$ random variables.
    For any fixed $\alpha_{-i}\in[0,1]^{n-1}$, the functions $\tilde x_i(\alpha v)\coloneqq \mathbb{E}_{\epsilon}\bigl[x_i(\alpha_i \epsilon_i v_i,\alpha_{-i}\epsilon_{-i} v_{-i})\bigr]$ and $\tilde p_i(\alpha v)\coloneqq \mathbb{E}_{\epsilon}\bigl[p_i(\alpha_i\epsilon_i v_i,\alpha_{-i}\epsilon_{-i} v_{-i})\bigr]$, denoting buyer $i$'s expected allocation and payment respectively, are continuous in $\alpha_i$.
\end{restatable}

\Cref{prop:cont-alloc-pymt} shows that the perturbations smooth out the mechanism's allocation and payment functions. 
We therefore refer to the perturbed mechanism $\tilde{\mathcal M}_\delta$ as the \emph{smoothed} mechanism, in which each buyer maximizes expected utility subject to satisfying her budget constraint in expectation over the perturbations (with $-\infty$ utility if the expected budget constraint is violated). This restores continuity while preserving the strategic structure of the \emph{original} mechanism $\mathcal M$. We therefore define an \abspe with respect to the smoothed mechanism.

\begin{definition}[$\delta$-smooth \textsf{Abs-PE}]
    A $\delta$-smooth \abspe of $\mathcal M$ is an \abspe of the smoothed mechanism $\tilde{\mathcal M}_\delta$. That is, it is a pacing vector $\alpha\in[0,1]^n$ such that $\tilde p_i(\alpha v)\le B_i$ for all $i$, and the no-unnecessary-pacing condition holds (with respect to $\tilde p$).
\end{definition}

\subsubsection{Existence of $\delta$-smooth \textsf{Abs-PE}}
We now work with the smoothed mechanism $\tilde{\mathcal M}_\delta$ and define the set of pacing multipliers that are budget-feasible with respect to expected payments, $\mathcal A_\delta := \{\alpha\in[0,1]^n:\ \tilde p_i(\alpha v)\le B_i\ \ \forall i\}.$ 
First, we remark that Assumption~\ref{ass:pymt-opp-inverse-mnt} is preserved under smoothing: fix a buyer $i$ and bid vectors $b,b'$ with $b_i=b_i'$ and $b'_{-i}\le b_{-i}$, then for every $\epsilon\ge 0$, $p_i(\epsilon\cdot b')\ge p_i(\epsilon\cdot b)$, and so $\tilde p_i(b')\ge \tilde p_i(b)$ after taking expectations.
Hence, the join-closure property from \Cref{lem:join-closure} carries over to $\mathcal A_\delta$. 
Moreover, smoothing ensures the expected payment function $\tilde p$ is Lipschitz in the log-domain over $(0,1]^n$, which implies compactness of $\mathcal A_\delta$. 
Combined with join-closure, compactness implies that $\mathcal A_\delta$ admits a coordinate-wise greatest element $\alpha_\delta^*\in\mathcal A_\delta$, which will serve as our candidate pacing equilibrium~vector.

\begin{restatable}[Log-Lipschitz payments under smoothing]{proposition}{logLip}\label{prop:log-lip}
    For every buyer $i$, the expected payment $\tilde p_i(\alpha v) \coloneqq \mathbb{E}_{\epsilon}\bigl[p_i(\alpha_i \epsilon_i v_i,\alpha_{-i} \epsilon_{-i} v_{-i})\bigr]$ is Lipschitz continuous on $(0,1]^n$ in the log domain. In particular, for any $\alpha,\alpha'\in(0,1]^n$,
    \( |\tilde p_i(\alpha'v)-\tilde p_i(\alpha v)|
          \le \frac{n\bar V}{\delta}\,\|\log \alpha' - \log \alpha\|_\infty,
    \)
    where $\log$ is applied coordinate-wise.
\end{restatable}

Proposition~\ref{prop:log-lip}, together with the boundary regularity conditions on the mechanism, yields compactness of the budget-feasible set.

\begin{restatable}{lemma}{compactness}\label{lem:compact-space}
    The set $\mathcal A_{\delta}:=\{\alpha\in[0,1]^n: \tilde p_i(\alpha v)\le B_i\ \forall i\}$ is compact.
\end{restatable}

Since $\mathcal A_\delta$ is nonempty, compact, and closed under the component-wise maximum, it admits a coordinate-wise greatest element.

\begin{restatable}[Pareto-dominant budget-feasible multiplier]{proposition}{ParetoDominantBFM}\label{prop:pareto-dominant-BFM}
    The set $\mathcal A_\delta$ of budget-feasible multipliers in the smoothed mechanism admits a Pareto-dominant element $\alpha^*$.
\end{restatable}

It remains only to rule out unnecessary pacing: if a paced buyer had slack budget, her multiplier could be increased while preserving feasibility, contradicting maximality.

\begin{restatable}[Guaranteed existence of $\delta$-smooth \textsf{Abs-PE}]{proposition}{ExistenceDelta}\label{prop:existence-delta}
The Pareto-dominant multiplier vector~$\alpha^*$ has no unnecessarily paced buyers, and hence forms a $\delta$-smooth \abspe of $\mathcal M$.
\end{restatable}

\Cref{prop:existence-delta} generalizes the defining structural property \textsf{FPPE}: the coordinate-wise greatest budget-feasible vector is an equilibrium vector. The same characterization extends to any smoothed abstract mechanism, under the opponent inverse monotonicity condition on payments.

\subsubsection{Existence of \textsf{Abs-PE}}\label{subsubsec:tie-breaking}
Having established the existence of an \abspe in the smoothed mechanism $\tilde{\mathcal M}_{\delta}$, we transfer this result back to the original (possibly discontinuous) mechanism $\mathcal M$ by adopting the limit-based tie-breaking/selection approach of~\cite{babaioff2020non}: for each bid profile $b$, and in particular at points where the outcome rule $(x,p)$ is discontinuous, the mechanism can select any outcome in the convex hull of all limit points of $(x(b^n),p(b^n))$ as $b^n \to b$. The key step is to show that any limit point of smoothed equilibrium outcomes lies in the tie-breaking correspondence $\hat X(b)$ defined below.

\begin{definition}[Tie‐breaking correspondence]\label{def:tie-break-corresp}
Consider the mechanism $\mathcal M=(x,p)$.  For each bid profile $b=\alpha\!\cdot\!v$, define
\(
  \hat x(b)=\Bigl\{(x^*,p^*) :\exists\,b^n\to b\text{ with }x(b^n)\to x^*,\,p(b^n)\to p^*\Bigr\},
\)
and let $\hat X(b)=\mathrm{conv}\,\hat x(b).$
We define a tie‐breaking rule for $\mathcal M$ as any pair $(x',p')$ satisfying $(x'(b),p'(b))\in \hat X(b) \, \forall b$.  We say $\mathcal M'~=~(x',p')$ is equivalent to $\mathcal M=(x,p)$ up to tie-breaking.
\end{definition}

\begin{remark}[Extended tie-breaking]
    Standard tie-breaking selects among outcomes that are tied at the same bid profile. Here we use an extended, limit-based notion: $\hat x(b)$ collects all outcome-payment pairs that can arise as limits of $(x(b^n),p(b^n))$ along sequences $b^n\to b$. 
    The convexification $\hat X(b)=\mathrm{conv}\,\hat x(b)$ can be interpreted as randomized tie-breaking among such limit outcomes (equivalently, selecting expected outcomes).
    In \fppe and \textsf{Pos-FPPE}, discontinuities arise only from exact ties, and this construction recovers the usual (randomized) tie-breaking; for abstract mechanisms, discontinuities can be more general, so we work with the correspondence $\hat X(\cdot)$ to obtain a well-defined selection at every bid profile.
\end{remark}

At bid profiles where $(x,p)$ is discontinuous, a tie-breaking rule specifies the realized outcome by selecting a point in $\hat X(b)$. Accordingly, we say that statements hold \emph{under appropriate tie-breaking} if they hold for some such selection rule.
With this terminology in place, we use the \abspe of the smoothed mechanism $\tilde{\mathcal M}_\delta$ from \Cref{prop:existence-delta} to construct an \abspe for the original mechanism $\mathcal M$ under an appropriate tie-breaking rule by taking limits as the perturbation parameter $\delta\to 0$.

\begin{restatable}[Guaranteed existence of \textsf{Abs-PE}]{theorem}{existence}\label{thm:existence}
    Consider any mechanism $\mathcal M$ whose payments satisfy Assumption~\ref{ass:pymt-opp-inverse-mnt}. Given valuations and budget profiles, there exists a tie-breaking rule for $\mathcal M$ for which an \abspe exists.
\end{restatable}

Notably, \Cref{thm:existence} establishes existence of an \abspe for some tie-breaking rule. This qualification is substantive: different resolutions of discontinuities can change the implemented outcome and payments thereby changing utilities and the set of budget-feasible vectors, and an \abspe need not exist for every tie-breaking rule. The theorem guarantees that at least one choice of tie-breaking admits an \textsf{Abs-PE}\endnote{See~\cite{babaioff2020non} for a similar qualification regarding equilibrium existence and a brief discussion of the role of tie-breaking.}.

This result shows that pacing equilibria retain a lot of structure beyond the bid-maximizing pay-your-bid mechanisms studied in \Cref{sec:bid-max-pyb}. 
In the smoothed mechanism, the greatest budget-feasible vector is the least restrictive pacing vector consistent with all buyers’ budgets, giving the platform a natural benchmark for setting pacing multipliers. 
It can also be reached by a simple monotone procedure that repeatedly increases the multipliers of unnecessarily paced buyers. 
More broadly, the join-semilattice structure of the budget-feasible set supports our existence result, showing that a simple monotonicity condition on payments is enough to extend pacing equilibrium to a significantly broader class of mechanisms.

\subsection{Properties of \abspe}
\label{subsec:properties}
In this section, we establish structural and revenue properties of \abspe for abstract mechanisms under additional monotonicity conditions on payments. 
Since $\mathcal M=(x,p)$ may be discontinuous in bids, equilibrium statements for the unperturbed model are understood relative to a tie-breaking rule under which an \abspe exists (such a tie-breaking rule is guaranteed to exist by~\Cref{thm:existence}). Accordingly, whenever we invoke an assumption on payments, it is meant to hold for the payment rule induced by the relevant tie-breaking. For readability, we state the assumptions using a generic payment rule $p$. We continue to assume Assumption~\ref{ass:pymt-opp-inverse-mnt} throughout this section.

\subsubsection{Uniqueness and revenue-maximization}
Fix a tie-breaking rule $(x',p')$ such that the induced mechanism $\mathcal M'=(x',p')$ admits an \textsf{Abs-PE}. We give conditions on the induced payment rule $p'$ under which the equilibrium pacing vector is unique and revenue-maximizing among budget-feasible pacing vectors.
The key requirement is strict responsiveness of group payments to coordinated increases in pacing multipliers whenever the group pays a positive amount.

\begin{assumption}[Collective strict monotonicity]
\label{ass:collec-pymt-mnt}
    For any nonempty subset of buyers $G\subseteq N$, and any two pacing vectors $\alpha,\alpha'\in[0,1]^n$ such that $\alpha_{-G}=\alpha'_{-G}$ and $\alpha_i>\alpha_i'$ for all $i\in G$, we have $\sum_{i\in G}p_i(\alpha v)~>~\sum_{i\in G}p_i(\alpha' v)$ whenever $\sum_{i\in G}p_i(\alpha' v)>0.$
\end{assumption}

Assumption~\ref{ass:collec-pymt-mnt} rules out the possibility of two distinct equilibrium pacing vectors. Intuitively, since coordinate-wise maxima preserve budget feasibility (by join-closure), collective strict monotonicity implies that an equilibrium multiplier cannot be strictly increased on any nonempty set of buyers that is paying a positive amount. Hence multiple equilibria cannot coexist.

\begin{restatable}[Uniqueness under tie-breaking]{proposition}{uniquness}\label{prop:uniqueness}
If $p'$ satisfies Assumption~\ref{ass:collec-pymt-mnt}, then $\mathcal M'$ admits a unique \abspe~$\alpha^*$.
\end{restatable}

\Cref{prop:uniqueness} shows that the predictability of pacing is not limited to the optimization-based mechanisms of Section~\ref{sec:bid-max-pyb}. Under the payment monotonicity Assumptions \ref{ass:pymt-opp-inverse-mnt} and \ref{ass:collec-pymt-mnt}, the platform obtains a unique pacing policy despite the more general mechanism class.

Beyond uniqueness, a natural design question is whether the resulting pacing policy is also the most favorable from the platform's perspective. In particular, one may ask whether it maximizes revenue among budget-feasible pacing vectors. The following lemma shows that, under the additional assumption that total payments are monotone in the bid vector, the unique pacing equilibrium indeed achieves the maximum revenue among all pacing vectors satisfying budget feasibility. 

\begin{assumption}[Monotone revenue]\label{ass:monotone-rev}
    Let $\mathrm{Rev}_p(b):=\sum_{i\in N} p_i(b)$ denote the revenue induced by~$p$. We assume $\mathrm{Rev}_p$ is monotone in bids: if $\bar b\ge b$ coordinate-wise, then $\mathrm{Rev}_p(\bar b)\ge \mathrm{Rev}_p(b)$.
\end{assumption}

Revenue monotonicity yields revenue optimality from the maximum-element characterization: let $\alpha^*$ denote the unique \abspe pacing vector from~\Cref{prop:uniqueness}. If some other budget-feasible pacing vector exceeded $\alpha^*$ in any coordinate, then by join-closure their coordinate-wise join would remain budget-feasible and would strictly dominate $\alpha^*$, contradicting maximality. Hence $\alpha^*$ dominates all feasible vectors, and by revenue monotonicity it maximizes revenue.

\begin{restatable}[Revenue maximization under tie-breaking]{lemma}{revMax}\label{lem:rev-max-tiebreak}
    Suppose $p'$ satisfies Assumptions~\ref{ass:collec-pymt-mnt} and~\ref{ass:monotone-rev}. Then the (unique) \abspe of $\mathcal M'$ maximizes revenue among all budget-feasible pacing vectors.
\end{restatable}

\subsubsection{Shill-proofness} 
A natural question is whether the platform can increase revenue by introducing fake bids into the mechanism. We show that this form of manipulation provides no benefit: the revenue-maximizing pacing equilibrium is shill-proof. More precisely, for every set of fake buyers $F$ with valuations $v^F$ and budgets $B^F$, the seller's revenue at the revenue-maximizing \abspe of the shilled instance $(N\cup~F,(v^N,v^F),(B^N,B^F))$ does not exceed the revenue at the revenue-maximizing \abspe of the base instance $(N,v^N,B^N)$.

\begin{restatable}[Shill-proofness up to tie-breaking]{proposition}{shillProof}\label{prop:shill-proof}
    There exists a tie-breaking rule $(x',p')$ under which an \abspe exists and for which the revenue-maximizing \abspe is shill-proof.
\end{restatable}

\begin{remark}
In bid-maximizing pay-your-bid mechanisms, the \abspe pacing vector $\alpha^*$ is unique. Under Assumption~\ref{ass:monotone-rev}, $\alpha^*$ is revenue-maximizing among all budget-feasible pacing vectors. Moreover, since the equilibrium revenue is invariant to tie-breaking in this class of mechanisms, the shill-proofness guarantee of \Cref{prop:shill-proof} applies to this unique equilibrium as well.
\end{remark}

\subsection{Computation of \abspe via Adaptive Bidding Algorithms} \label{subsec:computation}
We study decentralized pacing dynamics in which buyers adjust their pacing multipliers based on spending in the smoothed mechanism \(\tilde{\mathcal M}_\delta\). Under our standing opponent inverse monotonicity assumption, we analyze two update rules: the budget-adjustment dynamics of~\citet{borgs2007dynamics}, which converge under common-scale payment responsiveness (Assumption~\ref{ass:homog}), and a multiplicative-weights (MW) update, which extends convergence to a broader class of mechanisms whose payments satisfy the responsiveness condition only with respect to each buyer's own bid (Assumption~\ref{ass:mult-OBM}). We show that the former converges to an approximate \abspe, while the latter satisfies the defining equilibrium conditions in time average. We first formalize the notion of an approximate equilibrium.

\begin{definition}[Approximate \textsf{Abs-PE}]\label{def:approx-abs-pe}
    Given an abstract mechanism $\mathcal M$, a valuation profile $v$ and a budget profile $B$, a vector of pacing multipliers $\alpha\in [0,1]^n$ is called a $\gamma$-approximate \abspe of $\mathcal M$ if, for every buyer $i$, it satisfies:
    \begin{itemize}
        \item Approximate budget feasibility: $p_i(\alpha_iv_i,\alpha_{-i}v_{-i})\le B_i(1+\gamma)$,
        \item Approximate no-unnecessary-pacing: if $p_i(\alpha_iv_i,\alpha_{-i}v_{-i}) \le (1-\gamma)B_i$, then $\alpha_i \ge 1-\gamma$.
    \end{itemize}
    Moreover, given a smoothing parameter $\delta\in (0,1)$, we say that a pacing vector is a \emph{$(\delta,\gamma)$-approximate \abspe} of $\mathcal{M}$ if it is a $\gamma$-approximate \abspe of the smoothed mechanism $\tilde{\mathcal{M}_{\delta}}$. 
\end{definition}

In the remainder of this section, we analyze the two dynamics on \(\tilde M_\delta\). Both convergence analyses use the log-domain Lipschitz continuity of expected payments established in~\Cref{prop:log-lip}. Accordingly, we initialize \(\alpha^0\in(0,1]^n\), and both update rules preserve positivity.

\subsubsection{Budget dynamics of~\citet{borgs2007dynamics}}
\citet{borgs2007dynamics} study a simple pacing heuristic for budget-constrained buyers who repeatedly participate in a first-price \(m\)-item auction mechanism. Each buyer adjusts her pacing multiplier based on spending: the multiplier is decreased when expenditure exceeds the budget target and increased otherwise. To avoid instability and induce fractional allocations in expectation, they inject small random perturbations into values. In the first-price single-item setting, they show that these updates converge efficiently to an approximate equilibrium. We extend their convergence analysis to general mechanisms whose payments respond sufficiently to a common scaling of all bids. Under Assumption~\ref{ass:homog}, Algorithm~\ref{alg:borgs}, applied to \(\tilde M_\delta\), converges efficiently to a \(\gamma\)-approximate \abspe. Hence, convergence of these spending-based feedback dynamics is not specific to first-price auction markets. In particular, standard first-price payments satisfy \Cref{ass:homog} with equality for \(\kappa_i=1\) for all \(i\).

\begin{assumption}[Common-scale responsiveness]
\label{ass:homog}
For every buyer $i$, there exists $\kappa_i>0$ such that, for every
bid profile $b$ and every $\lambda\ge 1$, $p_i(\lambda b) \ge \lambda^{\kappa_i}p_i(b).$
\end{assumption}

\begin{algorithm}[h]
\caption{Budget dynamics of~\cite{borgs2007dynamics}}
\label{alg:borgs}
\begin{algorithmic}[1]
\Require Initial pacing multipliers $\alpha^0\in(0,1]^n$, step size $\eta>0$
\For{$t=0,1,\ldots$}
    \State Run $\tilde{\mathcal M}_{\delta}$ with bids $b^t=\alpha^t\!\cdot\!\tilde v$
    \For{each buyer $i$}
        \If{$\tilde p_i(b^t)>B_i$}
            \State $\alpha_i^{t+1}\gets \alpha_i^t e^{-\eta}$
        \Else
            \State $\alpha_i^{t+1}\gets \min\{\alpha_i^t e^{\eta},1\}$
        \EndIf
    \EndFor
\EndFor
\end{algorithmic}
\end{algorithm}

\begin{restatable}{theorem}{borgs}
\label{thm:borgs-conv}
    Consider a mechanism $\mathcal M$ whose payments satisfy Assumption~\ref{ass:homog}. Fix $\delta\in(0,1)$ and run~\Cref{alg:borgs} on the smoothed mechanism $\tilde{\mathcal M}_\delta$ with budgets $B$ and approximation parameter $\gamma\in(0,1)$. Let $\underline{\kappa}:=\min_{i\in N}\kappa_i.$
    Then there exist $\eta>0$ and $T_\gamma<\infty$ such that, for all $t\ge T_\gamma$, the algorithm returns a $(\delta,\gamma)$-approximate \abspe of $\mathcal M$, provided that $\eta \le \gamma/C$, where $C=\frac{n\bar V}{\delta \min_i B_i}$. 
    More precisely, approximate budget feasibility holds for all $t\ge T_\gamma^{\mathrm{BF}} := \frac{1}{\underline{\kappa}\eta} \ln\!\left( \frac{\bar V/\min_i B_i}{1+\gamma} \right),$ and one may take $T_\gamma \ge 2T_\gamma^{\mathrm{BF}} -\frac{1}{\eta}\ln\!\left(\min_i\alpha_i^0\right).$
\end{restatable}

\subsubsection{Multiplicative-weights algorithm}
While \Cref{ass:homog} is satisfied by standard first-price auction mechanisms, its restriction on how payments respond to a common scaling of the entire bid profile may fail in more general mechanisms. For the multiplicative-weights dynamics, convergence can be established under the weaker condition that a buyer's payment responds sufficiently to increases in her own bid, holding opponents' bids~fixed.

\begin{assumption}[Own-bid responsiveness]\label{ass:mult-OBM}
For each buyer $i$, there exists $\mu_i>0$ such that for any fixed opponents' bids $b_{-i}$, and all $\alpha_i,\alpha_i'\in(0,1]$ with $\alpha_i>\alpha_i'$,
\(
  p_i(\alpha_i v_i,\, b_{-i}) \ge
  p_i(\alpha_i' v_i,\, b_{-i})\cdot\Bigl(\tfrac{\alpha_i}{\alpha_i'}\Bigr)^{\mu_i}.
\)
\end{assumption}

Assumption~\ref{ass:mult-OBM} captures a minimum responsiveness of payments to scaling buyer $i$’s bid: holding opponents’ bids fixed, raising $i$’s bid from $\alpha_i' v_i$ to $\alpha_i v_i$ must increase her payment by at least a factor of $(\alpha_i/\alpha_i')^{\mu_i}$, i.e., at least polynomially in the scale factor. Intuitively, it requires payments to exhibit positive elasticity in a buyer’s own bid. 
Under opponent inverse monotonicity, this condition is strictly weaker than Assumption~\ref{ass:homog}.
    \begin{algorithm}[h]
    \caption{Multiplicative-weights updates}
    \label{alg:MW}
    \begin{algorithmic}[1]
    \Require Initial pacing multipliers $\alpha^0\in(0,1]^n$, step size $\eta>0$
    \For{$t=0,1,\ldots,T$}
        \State Run $\tilde{\mathcal M}_{\delta}$ with bids $b^t=\alpha^t\!\cdot\!\tilde v$
        \For{each buyer $i=1,\ldots,n$}
            \State Observe excess payment $g_{t,i}=\tilde p_i(b^t)-B_i$
            \State Update $\alpha_i^{t+1}\gets \min\{\alpha_i^t e^{-\eta g_{t,i}},1\}$
        \EndFor
    \EndFor
    \end{algorithmic}
    \end{algorithm}

Fix a smoothing parameter $\delta\in(0,1)$ and consider the smoothed mechanism $\tilde{\mathcal M}_\delta$. Since smoothing preserves own-bid responsiveness of expected payments, Assumption~\ref{ass:mult-OBM} holds for $\tilde p$ as well.
We analyze the multiplicative-weights update in~\Cref{alg:MW} and establish two complementary guarantees: budget feasibility holds in time average, while violations of no-unnecessary-pacing occur on only a vanishing fraction of rounds.
Together, these imply convergence on average to a $(\delta,\gamma)$-approximate \textsf{Abs-PE} of $\mathcal M$. \Cref{thm:MW-avg-conv} formalizes the result.

\begin{restatable}[Average convergence of MW]{theorem}{MWconv}
\label{thm:MW-avg-conv}
Consider a mechanism $\mathcal M$ whose payments satisfy \Cref{ass:mult-OBM}. Fix $\delta\in(0,1)$, $\gamma\in(0,1)$, and define $R_B:=\max_{i\in N}\frac{\sum_{j\neq i}B_j}{B_i}.$
Let $\{\alpha^t\}_{t\ge 0}$ be the iterates produced by~\Cref{alg:MW} on $\tilde{\mathcal M}_\delta$ with step size $\eta$ satisfying $\frac{1}{(\gamma-\zeta)\min_{i\in N}\{\mu_iB_i\}} \ln\left(\frac{1-\gamma}{1-\gamma-\zeta}\right)
\le \eta \le \min\left\{\frac{\delta}{\bar V},\frac{\zeta\delta}{(1-\gamma)R_B\bar V} \right\},$ where $0<\zeta<\min\{\gamma,1-\gamma\}$ is arbitrary and the interval is assumed nonempty. Then:
\begin{enumerate}
\item \emph{Average budget feasibility.} For every buyer $i$ and every $T\ge 1$,
\(
\frac{1}{T}\sum_{t=0}^{T-1}\tilde p_i(\alpha^t v)
\le B_i + \frac{\ln(1/\bar\alpha_i)}{\eta\,T},
\)
where \(\bar\alpha_i>0\) is a uniform lower bound on buyer \(i\)'s iterates.
In particular, if $T \ge \frac{\ln(1/\bar\alpha_i)}{\eta\,\gamma B_i}$, then
\(
\frac{1}{T}\sum_{t=0}^{T-1}\tilde p_i(\alpha^t v)~\le~ (1+\gamma)B_i.  
\)

\item \emph{Vanishing unnecessary pacing.} Under Assumption~\ref{ass:mult-OBM}, for every buyer $i$,
\(
\frac{1}{T}\,
\#\Bigl\{t\in\{0,\dots,{T-1}\}:\ 
\tilde p_i(\alpha^t v)<(1-\gamma)B_i
\ \text{and}\ 
\alpha_i^t < 1-\gamma\Bigr\}
\;\xrightarrow[T\to\infty]{}\; 0.
\)
\end{enumerate}
Consequently, the MW iterates are approximately budget-feasible in time average, and the no-unnecessary-pacing condition is violated on only a vanishing fraction of rounds; we refer to this as convergence on average to a $(\delta,\gamma)$-approximate \textsf{Abs-PE}.
\end{restatable}

\begin{remark}[Coordinate-specific step-size bounds]
The upper bound on the step size in Theorem~\ref{thm:MW-avg-conv} uses generic bounds on how expected payments change with the pacing multipliers. These bounds can be improved when the mechanism admits tighter coordinate-specific estimates. In particular, suppose there exist constants $L_{ij}\ge 0$ such that, for any $\alpha,\alpha'\in(0,1]^n$ that differ only in coordinate $j$, with $\alpha_j'\ge\alpha_j$,
\(
\left|\tilde p_i(\alpha'v)-\tilde p_i(\alpha v)\right| \le L_{ij}|\log\alpha_j'-\log\alpha_j|,
\)
where, for $j\neq i$, this condition is required only when $\alpha_i=\alpha_i'~<~1-\gamma$. 
Letting $M_B:=\min_i\mu_iB_i$, the convergence guarantee continues to hold for any step size satisfying $\frac{1}{(\gamma-\zeta)M_B} \ln\!\left(\frac{1-\gamma}{1-\gamma-\zeta}\right) \le \eta \le \min\left\{\min_i\frac{1}{L_{ii}},\; \min_i\frac{\zeta B_i}{\sum_{j\neq i}L_{ij}B_j}\right\}.$
The generic upper bound in Theorem~\ref{thm:MW-avg-conv} is recovered from the smoothing estimates $L_{ii}=\frac{\bar V}{\delta},$ and $L_{ij}=\frac{(1-\gamma)\bar V}{\delta},\, j\neq i.$
Hence, mechanisms with tighter coordinate-specific payment bounds may admit a larger range of learning rates.
\end{remark}

Taken together, the budget dynamics in Algorithms~\ref{alg:borgs} and~\ref{alg:MW} show that approximate pacing equilibria can be approached through decentralized adjustment based on realized spending. Under the structural conditions identified above, pacing operates as a scalable feedback system: multipliers respond to spending deviations from budget targets, with the dynamics of~\cite{borgs2007dynamics} converging to an approximate \abspe and the multiplicative-weights dynamics satisfying the equilibrium conditions in time average. This provides a dynamic foundation for budget management in general allocation environments.

\section{Experiments}
The previous sections establish structural, economic, and computational properties of pacing equilibria across increasingly general allocation mechanisms. We complement our theoretical analysis with numerical experiments that explore aspects of pacing equilibria not captured by our theoretical results.
First, Section~\ref{sec:posfppe} shows that a \posfppe need not constitute a market equilibrium because buyers may prefer alternative feasible allocations at the induced prices. We ask how frequently such profitable deviations arise and how large the resulting utility gains are. 
Second, Section~\ref{sec:bid-max-pyb} provides worst-case liquid welfare guarantees.
We examine how close \abspe allocations are to the liquid welfare optimum across two allocation environments within the bid-maximizing pay-your-bid class.
Finally, we study the performance of the proposed pacing dynamics in Algorithms~\ref{alg:borgs} and \ref{alg:MW}, comparing their observed convergence behavior with the equilibrium conditions and theoretical guarantees established in Section~\ref{sec:abstract-mech}.

Across the experiments, we generate synthetic instances representing several canonical platform-allocation environments including first-price position auctions as well as shared-capacity allocation, matching, and proportional sharing mechanisms. 
When generating budgets, we use variants of the rule $B_i=\lambda_i \frac{1}{n}\sum_j v_{ij},$ where $\frac{1}{n}\sum_j v_{ij}$ is buyer $i$’s expected value under a uniform random assignment. 
This follows the calibration used by~\cite{conitzer2022pacing}, who set $B_i=\frac{1}{n}\sum_j v_{ij}$ to generate a ``good mixture'' of paced and unpaced buyers. In experiments that vary market structure, we rescale the distribution of $\lambda_i$ so that the fraction of paced buyers remains approximately constant across treatments. This lets us isolate the effect of the market characteristic being varied, while keeping the overall budget tightness roughly fixed.

\subsection{Position Auctions and Market Equilibrium}\label{sec:exps-position}

Section~\ref{sec:posfppe} shows that buyers may obtain higher utility by reoptimizing at \posfppe prices than from their \posfppe allocations. While \Cref{thm:no-me-1} shows that such gains can arise under mild conditions, it does not indicate their economic significance in finite position-auction markets. We compare each buyer's \posfppe utility with the maximum utility she could obtain by reoptimizing at the induced prices.

Consider a \posfppe allocation  $x^*$ and its equilibrium prices $p^*$. Let $U_i^D(p^*)$ denote buyer $i$'s maximum quasilinear utility at $p^*$, subject to her budget and the feasibility constraint of at most one slot per auction. 
For $U_i^D(p^*)>0$, we define buyer $i$'s \emph{relative demand gap} as $r_i = \frac{U_i^D(p^*) - U_i(x_i^*,p^*)}{U_i^D(p^*)}$, which measures the fraction of the optimal utility at $p^*$ that she misses under her assigned allocation $x^*$. Buyers with $U_i^D(p^*)=0$ are excluded. We classify a buyer as having a \emph{demand violation} when $r_i>0.05$, corresponding to a loss of more than $5\%$ of her attainable utility.
We generate markets with \(m=10\) auctions, values drawn independently from \(U(0,1)\), and geometric slot qualities \(q_k=\rho^{k-1}\). The baseline has \(n=20\), \(s=5\), and \(\rho=0.6\). We vary slot similarity \(\rho\in\{0.20,0.35,0.50,0.65,0.80,0.95\}\), the number of slots \(s\in\{2,3,4,5,6,8\}\) holding \(n=20\), and market thickness \(n/s\in\{1,2,4,8,16\}\) by varying \(n\) and holding \(s=5\). Budgets are chosen so that approximately \(70\%\) of buyers are paced, and each point averages 50 independently generated markets.

\begin{figure}[h]
    \centering
    \includegraphics[width=\textwidth]{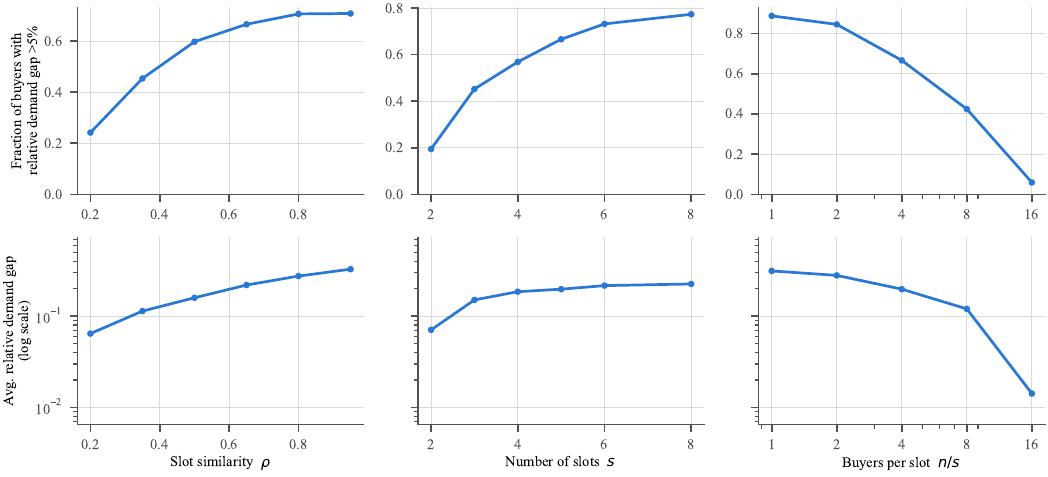}
    \caption{Demand violations at \posfppe prices. The top row shows the fraction of buyers with a demand violation, while the bottom row reports the average relative demand gap across buyers. The columns vary slot similarity, the number of slots per auction, and the number of buyers per~slot.}
    \label{fig:sec3-demand-violation}
\end{figure}

Figure~\ref{fig:sec3-demand-violation} shows that when slots are similar and competition is limited, the average fraction of attainable utility lost is larger, and such losses arise for a larger fraction of buyers. As \(\rho\) increases, both the incidence and magnitude of demand violations increase. Holding \(\rho=0.6\) fixed, adding more slots produces a similar pattern. In contrast, as the number of buyers per slot increases, both the frequency of violations and the average fraction of attainable utility lost decrease.
We further examine how demand gaps vary with market thickness, equilibrium pacing multipliers, and budget tightness in Appendix~\ref{app:sec-exps-position}.

\subsection{Liquid Welfare Performance of \abspe}\label{sec:exps-LW}

We next examine how close the liquid welfare of \abspe allocations is to the optimum in finite bid-maximizing pay-your-bid markets. We further compare the realized welfare losses with the worst-case additive and multiplicative guarantees from \Cref{thm:abspe-lw} in Appendix~\ref{sec:app-LW-tightness}.

We consider two bid-maximizing pay-your-bid environments: (1) a one-to-one matching mechanism in which the allocation space is $X=\left\{x\geq0:\sum_i x_{ij}\leq1\forall j,\;\sum_jx_{ij}\leq1\forall i\right\}$, and (2) a shared-capacity mechanism, with $X=\left\{x\geq 0:\sum_i x_{ij}\leq c_j\forall j,\;x_{ij}\leq \bar x_{ij}\forall i.j\right\}$; here $\bar x_{ij}\sim U(0.4,1)$ is buyer $i$'s usage limit for resource $j$, and $c_j=\frac{1}{2}\sum_i \bar x_{ij}$ is the total capacity.
In both environments, the number of items or resources equals the number of buyers, and values are drawn independently from $U(0,1)$. We vary $n\in\{10,20,40,80\}$ and generate 100 markets for each environment and market size. Budgets are calibrated separately so that approximately half of the buyers are paced.

\begin{figure}[h]
    \centering
    \includegraphics[width=0.9\textwidth]{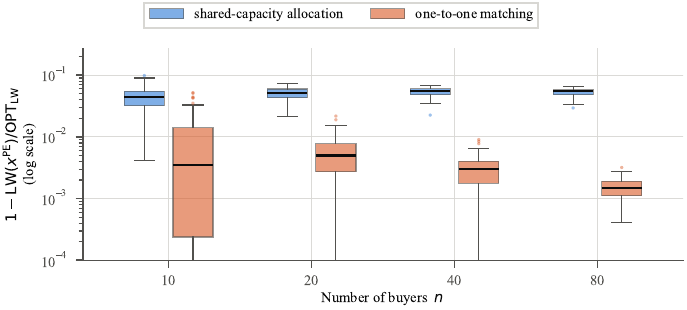}
    \caption{Liquid welfare loss at \abspe. Each box summarizes the relative LW loss $1-LW(x^{PE})/OPT_{LW}$ across 100 independently generated markets. The vertical axis is logarithmic and truncated at \(10^{-4}\).}
    \label{fig:sec4-lw-loss}
\end{figure}

Figure~\ref{fig:sec4-lw-loss} reports, for each simulated market, the relative liquid welfare loss $1-\frac{LW(x^*)}{OPT_{LW}}$, where $x^*$ is the \abspe allocation and $OPT_{LW}=\max_{x\in X}LW(x)$ denotes the optimal liquid welfare.
The results show that the loss remains small as market size grows. In the shared-capacity environment, the median loss stays close to $5\%$ as $n$ increases from 10 to 80. In the matching environment, median losses remain below $0.5\%$ and decrease further with market size.
Thus, relative liquid welfare loss does not increase with market size over the range considered. However, its magnitude varies across allocation environments, with smaller losses in matching than in shared-capacity allocation.

\subsection{Computing Approximate \abspe in Smoothed Mechanisms}
Section~\ref{sec:bid-max-pyb} extends the PACE framework of~\cite{gao2021online} to computing pacing equilibria in bid-maximizing pay-your-bid mechanisms. Appendix~\ref{app:exps-pace} compares the resulting pacing multipliers with the corresponding solutions of (CP1) from~\Cref{thm:convex-program}. In this subsection, we focus on the pacing dynamics for general mechanisms from Section~\ref{sec:abstract-mech}.

\paragraph{Last-iterate budget dynamics.}

We first evaluate the budget dynamics of \cite{borgs2007dynamics}, who prove convergence for first-price single-slot auctions and note the extension to multi-slot auctions. We generalize this guarantee to abstract mechanisms whose payments increase sufficiently when all bids are scaled up (\Cref{thm:borgs-conv}). 
We test Algorithm~\ref{alg:borgs} on two pay-your-bid mechanisms: (1) first-price position auctions, as a benchmark within the multi-slot auction setting, and (2) proportional sharing, as an example outside the auction environment.  
We vary $n\in\{10,20,30\}$ and generate 25 markets for each mechanism and market size. Position-auction markets have 20 auctions with $s=n$ slots each. In the proportional sharing mechanism, there are 20 items, and the allocation rule is given by $x_{ij}(b)=\frac{b_{ij}}{\sum_h b_{hj}}$. 
Payments are smoothed with $\delta=0.2$, following the perturbation framework in~\Cref{section:mechanism-perturb}. Budgets are calibrated so that the mean equilibrium pacing multiplier is approximately 0.5 in both environments.
We initialize the pacing multipliers independently from \(\operatorname{Unif}(0,1]\), and use the maximum step size allowed by \Cref{thm:borgs-conv} for $\gamma=0.1$.

\begin{figure}[h]
    \centering
    \includegraphics[width=\textwidth]{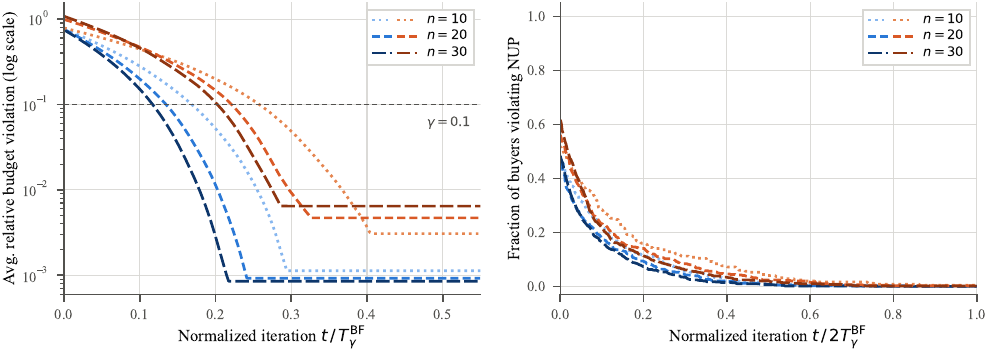}
    \caption{Convergence of the budget-adjustment dynamics of~\citet{borgs2007dynamics} in first-price proportional-sharing (blue) and position-auction (orange) mechanisms.}
    \label{fig:sec5-borgs}
\end{figure}

We examine how quickly the dynamics satisfy the approximate equilibrium conditions relative to the convergence bounds in \Cref{thm:borgs-conv}.
Figure~\ref{fig:sec5-borgs} reports the average relative budget violation, defined as the budget excess normalized by the buyer’s budget, and the fraction of buyers violating the approximate no-unnecessary-pacing (NUP) condition. Time is normalized by the corresponding bounds in \Cref{thm:borgs-conv}. 
The average budget violation falls below $\gamma=0.1$ after about $11\%$ to $25\%$ of the budget-feasibility horizon $T_\gamma^{BF}$ across the specifications considered. The fraction of buyers violating approximate NUP reaches zero within the corresponding guaranteed horizon in all cases. Hence, these simulations converge significantly faster than the worst-case iteration bounds suggest. Appendix~\ref{app:exps-worst-buyer} reports analogous results for maximum violations across buyers.

\paragraph{Multiplicative-weights dynamics.}

\Cref{thm:MW-avg-conv} replaces common-scale responsiveness with the weaker own-bid responsiveness condition (Assumption~\ref{ass:mult-OBM}) and establishes convergence of the approximate equilibrium conditions in time average. We test these dynamics in nonlinear-payment mechanisms. We consider the same environments as in the last-iterate budget dynamics experiment, but replace pay-your-bid payments with power payments: in proportional sharing, buyer $i$'s payment is $p_i(b)=b_i^2 x_i(b)$, and in the position mechanism a buyer assigned rank $r$ pays $b_i^2 q_r$, where $q_r=1-0.001(r-1)$. 

\begin{figure}[h]
    \centering
    \includegraphics[width=\textwidth]{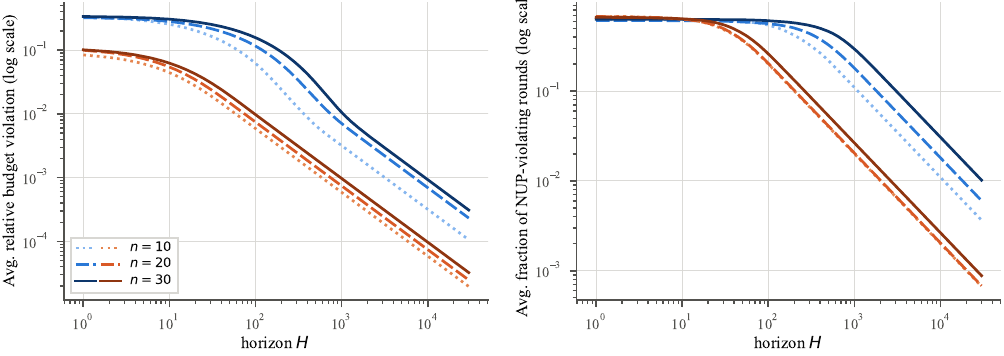}
    \caption{Time-averaged convergence of the multiplicative-weights dynamics under power payments for proportional sharing (blue) and first-price position auctions (orange).}
    \label{fig:sec5-mw}
\end{figure}

We run the dynamics for $H=3\times 10^4$ iterations. Figure~\ref{fig:sec5-mw} reports the average relative budget violation and the fraction of rounds violating approximate NUP. Both are computed for the time-averaged iterates. The average relative budget violation declines approximately inversely with the averaging horizon, consistent with the $O(1/H)$ dependence in \Cref{thm:MW-avg-conv}. The fraction of NUP-violating rounds also declines toward zero. By $H=3\times 10^4$, both time-averaged equilibrium violations are small across all specifications. Appendix~\ref{app:exps-worst-buyer} reports the corresponding worst-buyer measures and shows that the same qualitative pattern holds for the maximum violations across buyers.
\section{Conclusion}

In digital platforms, the allocation mechanism is only one component of a larger and more complex decision system. Participants often do not directly choose how to compete in each individual allocation opportunity. Instead, they submit campaign-level inputs -- such as bids, values, targeting criteria, and budgets -- and the platform's budget-management system determines how those inputs are translated into participation across many allocation events. This makes the design problem different from the design of a standalone auction. The relevant question is not only whether an allocation rule has desirable properties in isolation, but whether the aggregate system induced by allocation rules, payment rules, budgets, and pacing multipliers behaves~predictably.

This paper studies that question for pacing-based budget management in general mechanisms. The first-price single-item auction model provides an important benchmark, where pacing equilibria are known to have strong structural and computational properties. Our results show that these properties are not simply consequences of first-price payments, nor do they extend automatically to richer environments. Even in first-price position auctions, the market-equilibrium interpretation of the single-item model can fail. Thus, the single-item theory should not be applied mechanically to platforms with multi-slot allocation, capacity constraints, matching constraints, or other complex feasibility restrictions.

At the same time, the paper shows that many of the favorable properties of first-price pacing can be recovered once the right mechanism-level conditions are identified. In bid-maximizing pay-your-bid mechanisms, pacing equilibria remain well structured under mild convexity assumptions on the allocation and valuation spaces: they are uniquely determined, admit a convex-program characterization, and support efficiency and computational guarantees. In more general abstract mechanisms, payment monotonicity, smoothing, and appropriate tie-breaking provide conditions under which pacing equilibria exist and can be approached through expenditure-based adjustment dynamics. These results suggest that the robustness of pacing depends less on the label of the auction format and more on how the mechanism's allocation and payment rules respond to paced~bids. 

Taken together, our results support a more integrated view of platform design: the choice of a pacing rule and the choice of an allocation mechanism are not independent design decisions. A pacing multiplier is a simple operational control, but its effect depends on how the mechanism maps scaled bids into outcomes and payments. If this mapping preserves the relevant monotonicity and continuity properties, pacing can serve as a tool for regulating expenditure while remaining compatible with the platform's allocation and computational~objectives.  

From a broader perspective, our results do not suggest that there is a single best mechanism or that one budget management architecture is best for all platforms. Platforms differ in their objectives, information environments, allocation constraints, and implementation requirements, and may place different weights on considerations such as revenue, stability, bidder control, or computational simplicity. Our contribution is instead to develop a unified framework for evaluating when multiplicative pacing is compatible with a platform's allocation and payment environments. By analyzing increasingly general classes of mechanisms -- from single-item auctions, to position auctions, to bid-maximizing pay-your-bid mechanisms, and finally to general abstract mechanisms -- we identify the structural properties that support predictable, computable, and robust pacing outcomes. This moves the study of pacing from a single-item auction benchmark toward a general theory of automated budget management in complex platform markets.

\bibliographystyle{plainnat} 
\bibliography{refs}

\appendix
\section{Missing Proofs}
\subsection{Missing Proofs from \Cref{sec:posfppe}}

\subsubsection{Proof of \Cref{thm:no-me-1}}
\MEnonEquiv*
\begin{proof}
    Given market prices $p_{jk}$, bidder $i$'s demand problem is
    \begin{align*}
    \max_{x_i \ge 0}\ &\sum_{j,k} (v_{ij}q^j_k-p_{jk})x_{ijk}
    \\ &\text{s.t.} \quad
    \sum_{j,k} p_{jk}x_{ijk} \leq B_i.\tag{$D_i$} 
    \end{align*}
    Let $\lambda_i\ge 0$ be the Lagrange multiplier for the budget constraint.
    The KKT conditions for $(D_i)$ give $(1+\lambda_i)p_{jk}\ \ge\ v_{ij}q_k^j$ for all $(j,k),$ and $(1+\lambda_i)p_{jk}\ =\ v_{ij}q_k^j$ whenever $x_{ijk}>0$. 
    
    Now consider a \posfppe solution $\alpha$ with induced allocation $x$ and prices $p$. We show that there are instances where $x_i\notin D_i(p)$, hence $(x,p)$ is not a market equilibrium.
    Fix a bidder $i$ and an auction $j$ in which bidder $i$ receives a positive fraction of some slot $k$ with $q_k^j>0$ and is \emph{not tied} with all lower-ranked bidders (i.e., there exists a strictly lower slot $k'>k$ with $q_{k'}^j>0$ allocated to some bidder $i'$ satisfying
    $\alpha_{i'}v_{i'j}<\alpha_iv_{ij}$). Such an instance exists, e.g., whenever there are at least two distinct paced bids in some auction with at least two positive-quality slots.
    
    Since $x_{ijk}>0$, the \posfppe pricing rule gives $p_{jk}=q_k^j\,\alpha_i v_{ij}$, and thus bidder $i$'s bang-per-buck for slot $k$ is
    \(
    \frac{v_{ij}q_k^j}{p_{jk}}
    =\frac{v_{ij}q_k^j}{q_k^j\alpha_i v_{ij}}
    =\frac{1}{\alpha_i}.
    \)
    For the lower slot $k'$, we have $p_{jk'}=q_{k'}^j\,\alpha_{i'}v_{i'j}$ and hence
    \[
    \frac{v_{ij}q_{k'}^j}{p_{jk'}}
    =\frac{v_{ij}q_{k'}^j}{q_{k'}^j\alpha_{i'}v_{i'j}}
    =\frac{v_{ij}}{\alpha_{i'}v_{i'j}}
    >\frac{v_{ij}}{\alpha_iv_{ij}}
    =\frac{1}{\alpha_i}
    =\frac{v_{ij}q_k^j}{p_{jk}},
    \]
    where the strict inequality uses $\alpha_{i'}v_{i'j}<\alpha_iv_{ij}$.
    Therefore, bidder $i$ does \emph{not} maximize $\frac{v_{ij}q_k^j}{p_{jk}}$ across slot-goods in this auction. This contradicts the KKT condition that must hold for any optimal solution $x_i\in D_i(p)$, so $(x,p)$ is not a market equilibrium.
\end{proof}

\subsubsection{Proof of \Cref{prop:no-me-pos}}
\posMEnonEqui*
\begin{proof} 
    Consider a single auction $M=\{1\}$ with three bidders $N=\{1,2,3\}$ and three slots $S=\{1,2,3\}$ with qualities $q^1=(1,0.6,0.2)$.
    Let budgets be sufficiently large so $\alpha_h=1$ for all bidders, and let values be $(v_{11},v_{21},v_{31})=(10,5,1)$.
    
    The paced bids ranking is $v_{11}\succ v_{21}\succ v_{31}$, so bidders $1,2,3$ are allocated slots $1,2,3$ respectively, and the first-price slot prices are
    \(
    p_{11}=10,\, p_{12}=3,\, p_{13}=0.2.
    \)
    Consider bidder $2$ and the position-feasible utility $u_2''(\cdot,p)$. Since the utility is linear in $x$ and the feasible region $\sum_k x_{2,1,k}\le 1$ is a simplex, an optimum is attained at an extreme point.
    Bidder 2's net utility from slot 2 is $v_{21}q_2^1-p_{12}=5\cdot 0.6-3=0,$ whereas her net utility from slot 3 is $v_{21}q_3^1-p_{13}=5\cdot 0.2-0.2=0.8>0.$
    Hence bidder 2 strictly prefers slot 3 at prices $p$, so the allocation induced by the \posfppe is not in her demand set.
    Therefore, $(x,p)$ is not a market equilibrium.
\end{proof}
\subsection{Missing Proofs from Section~\ref{sec:bid-max-pyb}}

\subsubsection{Proof of~\Cref{thm:convex-program}}
\PYBConvexProgram*
\begin{proof}
    The proof follows directly from first-order optimality. Let $\phi(\alpha) = \max_{x\in X}\sum_i\alpha_iv_i(x_i)$. By Danskin's theorem, $\partial\phi(\alpha^*)=\operatorname{conv}\{v(x): x\in \arg\max_{X}\sum_i \alpha_i^* v_i(x_i)\}$.
    Moreover, first-order optimality for the outer problem yields a subgradient $g\in\partial\phi(\alpha^*)$ with $g_i=B_i/\alpha_i^*$ if $0<\alpha_i^*<1$ and $g_i\le B_i/\alpha_i^*$ if $\alpha_i^*=1$. 
    Let $M(\alpha^*)=\arg\max_{x\in X}\sum_i \alpha_i^* v_i(x_i)$ denote the maximizer set at $\alpha^*$.
    Write $g=\sum_t \lambda_t v(x^{(t)})$ where $x^{(t)}\in M(\alpha^*)$ and $\lambda_t\ge 0,\,\sum_t\lambda_t=1$. Since $X$ is convex and the objective $\sum_i \alpha_i^* v_i(x_i)$ is affine in $x$, then the set $M(\alpha^*)$ is convex. Hence, $x^*=\sum_t\lambda_t x^{(t)}\in M(\alpha^*) \subseteq X$, and affinity of $v_i$ yields $v_i(x_i^*)=\sum_t \lambda_t v_i(x_i^{(t)})=g_i$. Hence, setting $p_i^*=\alpha_i^* v_i(x_i^*)=\alpha_i^* g_i$, we have $p_i^* = B_i$ if $0<\alpha_i^*<1$ and $p_i^* \le B_i$ if $\alpha_i^*=1$, which gives budget feasibility and no-unnecessary-pacing, implying that $\alpha^*$ is an \textsf{Abs-PE}.
\end{proof}


\subsubsection{Proof of~\Cref{prop:eg-dual}}\label{app:proof-prop2}
\PYBdualEG*
\begin{proof} 
    Recall the pacing convex program from Theorem~\ref{thm:convex-program}: $\min_{0<\alpha\le 1} \max_{x\in X}  L(\alpha,x)$, where $L\coloneqq \sum_i \bigl(\alpha_i v_i(x_i) - B_i \log \alpha_i\bigr).$
    We first derive a lower bound on the pacing multipliers. Let $\alpha^*$ be an optimal solution of the pacing convex program. By the bounded valuations assumption, there exists $0<\bar V<\infty$ such that $0\le v_i(x_i)\le \bar V$ for all $i$ and all $x\in X$. We first show that for every $i\in N$, $\alpha_i^* \in [\ell_i,\ 1],$ where $\ell_i := \min\Bigl\{1,\frac{B_i}{\bar V}\Bigr\}.$
    If $\alpha_i^*=1$ the bound is trivial. Otherwise $0<\alpha_i^*<1$ and budget feasibility at \abspe gives $B_i = p_i(\alpha^*v_i) = \alpha_i^*\,v_i\bigl(x_i(\alpha^*)\bigr) \le \alpha_i^*\,\bar V$. Hence $\alpha_i^* \ge B_i/\bar  V$, and intersecting with $(0,1]$ yields the lower bound.
    
    Therefore, we can restrict the outer minimization w.l.o.g to the compact set $\mathcal B:= \prod_i[\ell_i,1]$.
    Since $X,\mathcal B$ are convex and compact, $v_i$ is affine, and $-\log$ is convex, then $L$ is convex in $\alpha$ and concave (linear) in $x$, and is continuous in $(\alpha,x)$. By Sion's minimax theorem, $\min_{\alpha\in \mathcal B}\ \max_{x\in X} L(\alpha,x) = \max_{x\in X}\ \min_{\alpha\in\mathcal B} L(\alpha,x).$
    
    Fix $x\in X$. The inner minimization separates:
    \(
    \min_{\alpha\in\mathcal B}L(\alpha,x)
    = \sum_i \min_{\alpha_i\in[\ell_i,1]}\bigl(\alpha_i v_i(x_i) - B_i\log\alpha_i\bigr).
    \)
    Define, for $u=v_i(x_i)$,
    \(
    \psi_i(u):=\min_{0<\alpha_i\le 1}\bigl(\alpha_i u - B_i\log\alpha_i\bigr),
    \)
    For $u\ge 0$, the minimizer over $(0,1]$ is
    $\hat\alpha_i(u)=\min\{1,B_i/u\}$ (with $\hat\alpha_i(0)=1$). 
    Since $u\le \bar V$, we have $\hat\alpha_i(u)\ge \min\{1,B_i/\bar V\}=\ell_i$, so the domain restriction $\alpha_i\in[\ell_i,1]$ doesn't change the optimum and the minimum value equals
    \[
    \psi_i(u) =
    \begin{cases}
    u, & 0\le u\le B_i,\\
    B_i - B_i\log B_i + B_i\log u, & u\ge B_i.
    \end{cases}
    \]
    Therefore the dual problem can be written as
    \(
    \max_{x\in X}\ \sum_i \psi_i\bigl(v_i(x_i)\bigr).
    \)
    We now show that this dual is equivalent (up to an additive constant independent of $x$) to the EG-style program in the statement.
    Fix $x\in X$ and consider the inner maximization over $(u,\delta)$: $\max_{u\ge 0,\delta}\ \sum_i \bigl( B_i \log u_i - \delta_i \bigr)$ subject to $0 \le \delta_i \le B_i$ and $u_i \le v_i(x_i)+\delta_i$. 
    Because $\log$ is increasing, at optimum we must have
    $u_i = v_i(x_i)+\delta_i$, so the problem reduces, for each $i$, to
    maximizing over $\delta_i\in[0,B_i]$ the concave function $f_i(\delta_i)=B_i\log\bigl(v_i(x_i)+\delta_i\bigr) - \delta_i.$
    Writing $u=v_i(x_i)$ and $y=u+\delta_i$, we have $ f_i(\delta_i) \;=\; B_i\log y - (y-u),\ \
      y\in[u,u+B_i].$
    The derivative with respect to $y$ is $f_i'(y) = \frac{B_i}{y} - 1$, which vanishes at $y=B_i$. There are two cases:
    \begin{itemize}
      \item If $u \le B_i$, then $B_i\in[u,u+B_i]$, so the optimum is at
        $y=B_i$, i.e., $\delta_i^* = B_i - u$, and
        \(
          \max_{\delta_i\in[0,B_i]} f_i(\delta_i)
          = B_i\log B_i - (B_i-u)
          = u + \bigl(B_i\log B_i - B_i\bigr).
        \)
      \item If $u \ge B_i$, then the feasible interval $[u,u+B_i]$ lies to the right of $B_i$, so $f_i'(y)\le 0$ on this interval and the optimum is achieved at $y=u$, i.e., $\delta_i^* = 0$, with
        \(
          \max_{\delta_i\in[0,B_i]} f_i(\delta_i)
          = B_i\log u
          = \bigl(B_i\log u + B_i - B_i\log B_i\bigr)
            + \bigl(B_i\log B_i - B_i\bigr).
        \)
    \end{itemize}
    In both cases we have
    \(
      \max_{\delta_i\in[0,B_i]} f_i(\delta_i)
      = \psi_i(u) + C_i,
    \)
    where $C_i = B_i\log B_i - B_i$ is a constant independent of $u$. Summing over $i$, we get
    \(
      \max_{u,\delta} \sum_i \bigl(B_i\log u_i - \delta_i\bigr)
      = \sum_i \psi_i\bigl(v_i(x_i)\bigr) + \sum_i C_i,
    \)
    and the additive constant $\sum_i C_i$ does not affect the maximizer $x$. Thus, maximizing the EG-style program over $(x,u,\delta)$ is equivalent to maximizing $\sum_i \psi_i(v_i(x_i))$ over $x\in X$, so it yields the same optimal allocation $x^*$ as the dual pacing problem.
    Finally, by strong duality and Theorem~\ref{thm:convex-program}, this allocation $x^*$ coincides with the allocation at an \abspe $\alpha^*$ of $\mathcal M$, which proves the claim.
\end{proof}


\subsubsection{Proof of \Cref{lem:abspetoCP}}
\PYBAbspeToCP*
\begin{proof}
    Let $\phi(\beta):=\max_{x\in X}\sum_i \beta_i v_i(x_i)$. For any fixed $\beta$, the objective $\sum_i \beta_i v_i(x_i)$ is continuous (since each $v_i$ is affine) and $X$ is compact, hence the maximum is attained.
    Moreover, Danskin's theorem yields $\partial \phi(\beta)=\operatorname{conv}\Bigl\{v(x): x\in \arg\max_{x\in X}\sum_i \beta_i v_i(x_i)\Bigr\}$. In particular, since $x(\alpha)\in\arg\max_{x\in X}\sum_i \alpha_i v_i(x_i)$, we have $g:=v(x(\alpha))\in\partial \phi(\alpha)$.
    Consider $F(\beta):=\phi(\beta)-\sum_i B_i\log \beta_i$ over $\beta\in(0,1]^n$.
    We verify the KKT conditions for $\min_{\beta\le 1}F(\beta)$ at $\alpha$.
    Introduce multipliers $\lambda\in\mathbb R^n_{\ge 0}$ for the constraints $\beta_i\le 1$. KKT stationarity requires a subgradient $g\in\partial\phi(\alpha)$ such that, for each $i$, $g_i-\frac{B_i}{\alpha_i}+\lambda_i=0,$ together with feasibility $\alpha\le 1$, dual feasibility $\lambda\ge 0$, and complementary slackness $\lambda_i(1-\alpha_i)=0$.
    
    In the bid-maximizing pay-your-bid mechanism, the induced payment satisfies $p_i=\alpha_i v_i(x_i(\alpha))=\alpha_i g_i$. By budget feasibility, $p_i\le B_i$ for all $i$, and by no-unnecessary-pacing, if $\alpha_i<1$ then $p_i=B_i$. This means that $\alpha_i<1 \implies g_i=\frac{B_i}{\alpha_i},$ and $\alpha_i=1 \implies g_i\le B_i=\frac{B_i}{\alpha_i}.$
    Define $\lambda_i:=0$ when $\alpha_i<1$, and $\lambda_i:=B_i-g_i\ge 0$ when $\alpha_i=1$. Then the KKT conditions are satisfied, and $\alpha$ is optimal for (CP1).
\end{proof}

\subsubsection{Proof of \Cref{thm:abspe-lw}}
\abspeLW*
\begin{proof}
Consider a bid-maximizing pay-your-bid mechanism with \abspe $\alpha^*$ and corresponding allocation $x^*$. We prove the following LW guarantees for $x^*$. 
\begin{enumerate}
    \item \textit{Multiplicative guarantee.} For each buyer $i$, let $P_i:= \alpha_i^*v_i(x_i^*)$ denote her payment at the \abspe. Partition the buyers into paced and unpaced buyers: $T:=\{i\in N:P_i=B_i\}$, and $U:=\{i\in N:P_i<B_i\}.$ By the no-unnecessary-pacing condition, $\alpha_i^*=1$ for every $i\in U$.
    Fix any feasible allocation $y\in X$. By the definition of liquid welfare,
    \begin{align}
    \mathrm{LW}(y) &= 
    \sum_{i\in T}\min\{v_i(y_i),B_i\} + \sum_{i\in U}\min\{v_i(y_i),B_i\} \nonumber\\ 
    &\le \sum_{i\in T}B_i + \sum_{i\in U}v_i(y_i). 
    \label{eq:bmpyb-lw-decomposition}
    \end{align}
    Since $\alpha_i^*=1$ for every $i\in U$ and valuations are nonnegative, then 
    \begin{align}
    \sum_{i\in U}v_i(y_i)
    &= \sum_{i\in U}\alpha_i^*v_i(y_i) \nonumber\\
    &\le \sum_{i\in N}\alpha_i^*v_i(y_i) \nonumber\\
    &\le \sum_{i\in N}\alpha_i^*v_i(x_i^*) = \sum_{i\in N}P_i.
    \label{eq:bmpyb-bid-max}
    \end{align}
    The second inequality follows because the mechanism is bid-maximizing, so $x^*$ maximizes $\sum_i \alpha_i^*v_i(x_i)$ over all $x\in X$.

    We next observe that $P_i=\mathrm{LW}_i(x_i^*)$ for every buyer $i$. Indeed, if $P_i<B_i$, then $\alpha_i^*=1$ by no unnecessary pacing, and therefore $P_i=v_i(x_i^*)<B_i,$ $P_i=\min\{v_i(x_i^*),B_i\} =\mathrm{LW}_i(x_i^*).$ 
    If $P_i=B_i$, then $P_i=\alpha_i^*v_i(x_i^*)\le v_i(x_i^*),$ since $\alpha_i^*\le1$, and hence $\mathrm{LW}_i(x_i^*) = \min\{v_i(x_i^*),B_i\} = B_i = P_i.$ Thus, $\sum_{i\in N}P_i = \mathrm{LW}(x^*).$
    
    Moreover, for every paced buyer $i\in T$, $B_i=P_i=\mathrm{LW}_i(x_i^*),$ and therefore $\sum_{i\in T}B_i \le \mathrm{LW}(x^*).$
    
    Combining the last two inequalities with \eqref{eq:bmpyb-lw-decomposition} and \eqref{eq:bmpyb-bid-max}, we obtain $\mathrm{LW}(y) \le 2\,\mathrm{LW}(x^*)$ for every feasible allocation $y\in X$. Taking $y$ to be liquid-welfare optimal gives $\mathrm{OPT}_{\mathrm{LW}} \le 2\,\mathrm{LW}(x^*),$ or equivalently, $\mathrm{LW}(x^*) \ge \frac12\mathrm{OPT}_{\mathrm{LW}}.$

    \item \textit{Additive guarantee.} We first prove the following claim: for any allocation $x\in X$, $\mathrm{LW}(x)\le \phi(x).$
    Fix a buyer $i$ and an allocation $x\in X$, and write $u := v_i(x_i)\in[0,\bar V]$. If $0\le u\le B_i$, then $\psi_i(u)=u=\mathrm{LW}_i(x_i).$ If $u\ge B_i$, then  $\mathrm{LW}_i(u) = B_i$ and $\psi_i(u)=B_i+B_i\log\frac{u}{B_i}\ge B_i=\mathrm{LW}_i(x_i).$ Summing over buyers gives $\mathrm{LW}(x)\le \phi(x).$

    We are now ready to prove the LW additive approximation bound. We evaluate $\phi$ at the equilibrium allocation $x^*$. Fix a buyer $i$. There are two cases:
    \begin{itemize}
        \item If $\alpha_i^*=1$, then, since the mechanism is pay-your-bid, $p_i^*=\alpha_i^*v_i(x_i^*)=v_i(x_i^*).$ Budget feasibility implies $v_i(x_i^*)\le B_i$, and hence $\psi_i\bigl(v_i(x_i^*)\bigr)=v_i(x_i^*)=\mathrm{LW}_i(x_i^*).$
        Moreover, $B_i\log\frac{1}{\alpha_i^*}=0.$

        \item If $\alpha_i^*<1$, then by the no-unnecessary-pacing condition, buyer $i$ exhausts her budget, so $B_i=p_i^*=\alpha_i^*v_i(x_i^*)< v_i(x_i^*).$ Thus $v_i(x_i^*)=\frac{B_i}{\alpha_i^*}\ge B_i,$ and therefore $\psi_i\bigl(v_i(x_i^*)\bigr) = B_i+B_i\log \frac{v_i(x_i^*)}{B_i} = B_i+B_i\log\frac{1}{\alpha_i^*} = \mathrm{LW}_i(x_i^*)+ B_i\log\frac{1}{\alpha_i^*}.$
    \end{itemize}
    Hence, in both cases, $\psi_i\bigl(v_i(x_i^*)\bigr) = \mathrm{LW}_i(x_i^*) + B_i\log\frac{1}{\alpha_i^*}.$ Summing over buyers yields 
    \begin{equation}
    \phi(x^*)=\mathrm{LW}(x^*)+ \sum_i B_i\log\frac{1}{\alpha_i^*}.
    \label{eq:phi-lw-equilibrium}
    \end{equation}

    Let $x^{\mathrm{LW}}\in\arg\max_{x\in X}\mathrm{LW}(x)$ be a liquid-welfare-optimal allocation. By~\Cref{prop:eg-dual}, $x^\ast$ maximizes $\phi(x) = \sum_i \psi_i(v_i(x_i))$ over $x\in X$, so $\phi(x^\ast) \ge \phi(x^{\mathrm{LW}}).$
    Using the lower bound $\phi(x)\ge \mathrm{LW}(x)$ and \eqref{eq:phi-lw-equilibrium}, we obtain $\mathrm{LW}(x^*)+ \sum_i B_i\log\frac{1}{\alpha_i^*} = \phi(x^*) \ge \phi(x^{\mathrm{LW}}) \ge \mathrm{LW}(x^{\mathrm{LW}}) = \mathrm{OPT}_{\mathrm{LW}}.$ Rearranging gives $\mathrm{LW}(x^*) \ge \mathrm{OPT}_{\mathrm{LW}} - \sum_i B_i\log\frac{1}{\alpha_i^*},$ as claimed.    
\end{enumerate}
Combining both guarantees yields 
\(
\mathrm{LW}(x^*) \ge \mathrm{OPT}_{\mathrm{LW}} - \min\left\{ \frac12\mathrm{OPT}_{\mathrm{LW}}, \sum_{i\in N}B_i\log\frac1{\alpha_i^*}\right\}.
\)
\end{proof}

\subsubsection{Proof of \Cref{lem:lw-tight-exp}}
\PYPfppeBound*
\begin{proof}
The construction adapts the two-bidder example from Theorem~5.2 of~\citet{fikioris2023liquid}.
Consider $m$ goods and two buyers with budgets $B_1=B_2=1$.
Valuations are uniform across goods: for all $j\in[m],$ we have $v_{1j}=1, \,v_{2j}=\varepsilon$, where $\varepsilon=1/m$.
Consider the pacing vector $\alpha=(\varepsilon, 1)$, and the allocation $x^*$ which allocate every good to buyer~1 (i.e., $x^*_{1j}=1$ and $x^*_{2j}=0$ for all $j$), breaking ties in favor of buyer~1. We verify that $(\alpha,x^*)$ constitutes a \fppe of the market. 

Define prices by $p_j=\max_i \alpha_i v_{ij}$, so for every good $j$,
\(
p_j = \max\{\alpha_1 v_{1j},\,\alpha_2 v_{2j}\}
     = \max\{\varepsilon\cdot 1,\,1\cdot \varepsilon\}
     = \varepsilon >0.
\)
Thus each good has strictly positive price and must be fully allocated in any \fppe outcome. Moreover, both buyers are tied as highest bidders on every good since $\alpha_1 v_{1j}=p_j=\alpha_2 v_{2j}$.
Hence, $x^*$ is consistent with the \fppe conditions: each good with $p_j>0$ is fully sold, and it is assigned only to (tied) highest bidders.
Buyer~1's total spend is
\(
P_1 = \sum_{j=1}^m p_j x^*_{1j} = \sum_{j=1}^m \varepsilon = m\varepsilon = 1 = B_1,
\)
and buyer~2 spends $P_2=0<B_2$.
By the ``no unnecessary pacing'' requirement in the definition of \textsf{FPPE}, any buyer who is not budget-tight must have multiplier~$1$, and indeed $\alpha_2=1$. Hence $(\alpha,x^*)$ is an \fppe outcome.

We compute the liquid welfare of the \fppe allocation. Since buyer~1 receives all $m$ goods, then
\(
\mathrm{LW}(x^*)
= \min\Bigl\{\sum_{j=1}^m v_{1j}x^*_{1j},\,B_1\Bigr\}
+ \min\Bigl\{\sum_{j=1}^m v_{2j}x^*_{2j},\,B_2\Bigr\}
= \min\{m,1\} + \min\{0,1\} = 1.
\)

Next, to compute optimal liquid welfare, consider the allocation $y$ that gives buyer~1 one good and buyer~2 the remaining $m-1$ goods. Then buyer~1 obtains value $1$ and buyer~2 obtains value $(m-1)\varepsilon$.
Therefore,
\(
\mathrm{LW}(y)
= \min\{1,B_1\} + \min\{(m-1)\varepsilon,B_2\}
= 1 + (m-1)\varepsilon
= 1 + \Bigl(1-\varepsilon\Bigr)
= 2-\varepsilon.
\)
Since $\max_{y\in X}\mathrm{LW}(y)\ge \mathrm{LW}(y)=2-\varepsilon$, we get
\(
\frac{\mathrm{LW}(x^*)}{\max_{y\in X}\mathrm{LW}(y)}
\le \frac{1}{2-\varepsilon}.
\)
On the other hand, in this instance $\max_{y\in X}\mathrm{LW}(y)=2-\varepsilon$ (buyer~1 is capped at $B_1=1$ and buyer~2's total value is at most $m\varepsilon=1$), so the inequality is in fact an equality, yielding the claimed ratio $1/(2-\varepsilon)$.
\end{proof}
\subsection{Missing Proofs from Section~\ref{sec:abstract-mech}}

\subsubsection{Proof of \Cref{prop:cont-alloc-pymt}}
\ContinuousPayments*
\begin{proof}
Consider the measurable space $\mathcal H = [0,1].$ Recall that if $F$ and $F'$ are two probability distributions defined on $\mathcal H$ with density functions $f$ and $f'$ respectively, then the TV distance is defined by
\[
\|F'-F\|_{TV}
= \frac12\int_{0}^{1}\bigl|f'(x)-f(x)\bigr|\,dx
= \sup_{A\subseteq[0,1]}\bigl|F'(A)-F(A)\bigr|
= \sup_{\substack{0\le g\le1\\\|g\|_\infty\le1}}
\bigl|\mathbb E_{F'}[g]-\mathbb E_F[g]\bigr|,
\]
for $g:\mathcal H \to \mathbb R$ a bounded measurable function on $[0,1]$ (Theorem 4.3 and Example 6.1 in~\cite{wu2017lecture}). 
It follows that $\bigl|\mathbb E_{F'}[g] -\mathbb E_F[g]\bigr| \le \|g\|_{\infty}\,\|F'-F\|_{TV}$.

Moreover, at each value of $x$, we have $|f'(x)-f(x)|=\max\{f'(x),f(x)\}-\min\{f'(x),f(x)\}$. Since both functions integrate to 1, we have 
\[ \int \max\{f'(x),f(x)\}dx = \int \{f'(x)+f(x)\}dx-\int\min\{f'(x),f(x)\}dx = 2- \int\min\{f'(x),f(x)\}dx.
\] Putting these together, we get \[
\int |f'(x)-f(x)| dx = \int [\max\{f'(x),f(x)\}-\min\{f'(x),f(x)\}] dx= 2-2\int \min\{f'(x),f(x)\}dx.
\]
Therefore, $\|F'-F\|_{TV} = \frac12 \int_0^1 |f'(x)-f(x)| dx = 1-\int \min\{f'(x),f(x)\} dx.$

Consider a buyer $i$, and fix the bids $b_{-i}$ of the other buyers. Fix any $\epsilon_{-i}\in[1-\delta,1]^{n-1}$ and write $\epsilon=(\epsilon_i,\epsilon_{-i})$. Let $g:\mathcal H\rightarrow \mathbb R$ denote the allocation or payment function of buyer $i$ under $\epsilon_{-i}$, that is 
$g_{\epsilon_{-i}}(y)= x_i(y\,v_i,\,\alpha_{-i}\epsilon_{-i} v_{-i})$ or 
$g_{\epsilon_{-i}}(y)= p_i(y\,v_i,\,\alpha_{-i}\epsilon_{-i} v_{-i})$. Both these functions are bounded: $X\subseteq\mathbb R^d$ is compact and $x(\cdot)\in X$ so $x_i(\cdot)$ is bounded, and by individual rationality (IR) of the mechanism, the payment satisfies $0\le g_{\epsilon_{-i}}(y)\le y\,\bar V \le \bar V$ for all $y\in[0,1]$.
Let $F_{\alpha_i}$ denote the probability distribution of $\alpha_i\epsilon_i$ where $\epsilon_i\sim \mathrm{Unif}[1-\delta,1]$.
We will show that $\tilde x_i(\alpha_i,\alpha_{-i})$ and $\tilde p_i(\alpha_i,\alpha_{-i})$ are continuous in $\alpha_i$ on $(0,1]$,
and that $\tilde p_i(\alpha_i,\alpha_{-i})$ additionally extends continuously to $\alpha_i=0$.

Fix $\alpha_i\in (0,1]$ and let $\alpha_i'=\alpha_i+\mu\in (0,1]$ with $|\mu| <  \frac{\delta}{1-\delta}\alpha_i$.
Then $F_{\alpha_i}=\mathrm{Unif}(I)$ and $F_{\alpha_i'}=\mathrm{Unif}(I')$, where
$I=[(1-\delta)\alpha_i,\alpha_i]$, and $I'=[(1-\delta)\alpha_i',\alpha_i'].$
Let $L=|I|=\delta\alpha_i$ and $L'=|I'|=\delta\alpha_i'$. The corresponding densities are
$f_{\alpha_i}(x)=\frac{1}{L}\mathbf 1_I(x)$ and $f_{\alpha_i'}(x)=\frac{1}{L'}\mathbf 1_{I'}(x)$.
We have
\(
\|F_{\alpha_i'} - F_{\alpha_i}\|_{TV}
=1 - \int_{[0,1]}\min\{f_{\alpha_i'}(x),f_{\alpha_i}(x)\}\,dx
=1 - \int_{I\cap I'}\min\bigl\{f_{\alpha_i}(x),f_{\alpha_i'}(x)\bigr\}\,dx, 
\)
where the first equality follows from the definition of TV distance involving the minimum, and the second is due to the fact that the minimum is zero outside $I\cap I'$. 

Next, note that $I\cap I' = 
\bigl[\max\{(1-\delta)\alpha_i,\,(1-\delta)(\alpha_i+\mu)\},
      \min\{\alpha_i,\;\alpha_i+\mu\}\bigr]
= \bigl[(1-\delta)(\alpha_i+\mu),\;\alpha_i\bigr],$
and on this interval $f(x)=1/L$, $f'(x)=1/L'$ with $L'>L$. Hence, we have
$\min\{f_{\alpha_i}(x),f_{\alpha_i'}(x)\} = \frac1{L'}$ for $x\in I\cap I'$, and $|I\cap I'| = \alpha_i - (1-\delta)(\alpha_i+\mu)
= \delta\,\alpha_i - (1-\delta)\mu.$
Therefore $\|F_{\alpha_i'}-F_{\alpha_i}\|_{TV}
= 1 - \frac{|I\cap I'|}{L'}
= 1 - \frac{\delta\,\alpha_i - (1-\delta)\mu}{\delta\,(\alpha_i+\mu)}
= \frac{\mu}{\delta\,(\alpha_i+\mu)},$ which tends to zero as $\mu\to0$ for all $\epsilon_{-i}$. Combining with the established total‐variation bound, it follows that for any sequence $\alpha^k\to\alpha$ we have
$\mathbb E_{\epsilon_i}\bigl[g_{\epsilon_{-i}}(\alpha^k\epsilon_i)\bigr]
\to \mathbb E_{\epsilon_i}\bigl[g_{\epsilon_{-i}}(\alpha\epsilon_i)\bigr]$
for every $\epsilon_{-i}$.
Finally, taking expectation over $\epsilon_{-i}$ and using dominated convergence gives continuity of 
$E_{\epsilon_i,\epsilon_{-i}}\!\bigl[x_i(\alpha\epsilon_i v_i,\alpha_{-i}\epsilon_{-i}v_{-i})\bigr]$ and 
$E_{\epsilon_i,\epsilon_{-i}}\!\bigl[p_i(\alpha\epsilon_i v_i,\alpha_{-i}\epsilon_{-i}v_{-i})\bigr]$ on (0,1]. 

It remains to handle $\alpha_i=0$ for payments. For any realization of $\epsilon$, let $\tilde b=(\alpha_i\epsilon_i v_i,\alpha_{-i}\epsilon_{-i}v_{-i})$, then IR gives
\(
0\le p_i(\tilde b)\le \tilde b_i(x_i(\tilde b))=\alpha_i\epsilon_i\,v_i(x_i(\tilde b))\le \alpha_i\epsilon_i \bar V\le \alpha_i\bar V.
\)
Taking expectations yields $0\le \tilde p_i(\alpha_i,\alpha_{-i})\le \alpha_i\bar V$, and hence $\tilde p_i(\alpha_i,\alpha_{-i})\to 0$ as $\alpha_i\downarrow 0$. Moreover, at $\alpha_i=0$, we have $0\le p_i(0,\alpha_{-i}\tilde v_{-i})\le 0$, so $\tilde p_i(0,\alpha_{-i})=0$. Therefore $\tilde p_i$ is continuous at $\alpha_i=0$ as well.
\end{proof}

\subsubsection{Proof of \Cref{prop:log-lip}}
\logLip*
\begin{proof}
    Fix a buyer $i$, and assume w.l.o.g. that $\alpha_i' > \alpha_i$. 
    We first prove a one-coordinate bound. Fix multipliers $\alpha_{-i}$ of the opponents. 
    From the proof of~\Cref{prop:cont-alloc-pymt}, we use the total‐variation bound in applied to the payment function, and the equality $\|F_{\alpha_i'}-F_{\alpha_i}\|_{TV} = \frac{\mu}{\delta\,(\alpha_i+\mu)},$ with the fact that it holds for every $\epsilon_{-i}$, to get $\bigl|\mathbb E_{\epsilon}[p_i(\alpha_i'\epsilon_i v_i,\alpha_{-i}\epsilon_{-i} v_{-i})] -\mathbb E_{\epsilon}[p_i(\alpha_i\epsilon_i v_i,\alpha_{-i}\epsilon_{-i} v_{-i})]\bigr| 
    \le \|p_i\|_\infty\;\frac{|\alpha_i'-\alpha_i|}{\delta\,(\alpha_i')} 
    \le \frac{\bar V}{\delta}\Bigl|1-\frac{\alpha_i}{\alpha_i'}\Bigr| 
    \le \frac{\bar V}{\delta}\bigl|\log \alpha_i' - \log \alpha_i\bigr|,$
    where the second inequality follows from individual rationality of the mechanism, and the third inequality uses $1-e^{-y}\le y$ with $y=\log(\alpha_i'/\alpha_i)$. 
    To obtain the joint bound, interpolate between $\alpha$ and $\alpha'$ by changing one coordinate at a time:
    let $\alpha^{(0)}=\alpha$ and for $k=1,\dots,n$ define
    \(
    \alpha^{(k)} := (\alpha_1',\dots,\alpha_k',\alpha_{k+1},\dots,\alpha_n).
    \)
    There, we obtain
    \(
    |\tilde p_i(\alpha')-\tilde p_i(\alpha)|
    \le \sum_{k=1}^n |\tilde p_i(\alpha^{(k)})-\tilde p_i(\alpha^{(k-1)})|
    \le \frac{\bar V}{\delta}\sum_{k=1}^n |\log \alpha_k'-\log \alpha_k|
    \le \frac{n\bar V}{\delta}\,\|\log \alpha'-\log \alpha\|_\infty.
    \)
\end{proof}


\subsubsection{Proof of \Cref{lem:compact-space}}
\compactness*
\begin{proof}
Since $\mathcal A_\delta\subseteq[0,1]^n$, it is bounded. Define $\Phi:[0,1]^n\to\mathbb R^n$ by $\Phi(\alpha)=(\tilde p_i(\alpha))_{i\in N}$. We claim that $\Phi$ is continuous on $[0,1]^n$. 
On the interior $(0,1]^n$, continuity follows from~\Cref{prop:log-lip}.
To extend continuity to boundary points, fix $\bar\alpha\in[0,1]^n$ and let $\alpha^t\to\bar\alpha$. Let $Z:=\{i\in N:\ \bar\alpha_i=0\}$.
Write $b^t=\alpha^t\epsilon v$ and $\bar b=\bar\alpha\,\epsilon v$ for the corresponding (random) bid profiles under $\epsilon\sim \mathrm{Unif}[1-\delta,1]^n$. For each fixed realization of $\epsilon$, we have $b^t\to \bar b$ coordinate-wise, and all profiles lie in the bounded bid region.
Fix a buyer~$i\in N$.
\begin{itemize}
     \item If $i\in Z$, then for any realization of $\epsilon$, individual rationality gives \(
     0\le p_i(\alpha^t\epsilon v)\le \alpha_i^t\epsilon_i v_i(x_i(\alpha^t\epsilon v))\le \alpha_i^t\bar V.
     \)
     Since $\alpha_i^t\to 0$, we have $p_i(\alpha^t\epsilon v)\to 0$ for every $\epsilon$. Moreover, at $\bar\alpha_i=0$, individual rationality gives $p_i(\bar\alpha\,\epsilon v)=0$. Since $0\le p_i(\alpha^t\epsilon v)\le \bar V$ for all $t$ and $\epsilon$, dominated convergence yields \(
     \tilde p_i(\alpha^t v)=\mathbb E_\epsilon[p_i(\alpha^t\epsilon v)]\to 0 = \tilde p_i(\bar\alpha v).
     \)
     \item If $i\notin Z$, consider the bid vectors $(b_{-Z}^t,b_Z^t)$ and $\bar b=(\bar b_{-Z},0^Z)$. For each fixed $\epsilon$, the vanishing-bids continuity assumption gives $p_i(b_{-Z}^t,b_Z^t)\to p_i(\bar b_{-Z},0^Z)$. Since $0\le p_i(b^t)\le \alpha_i^t\epsilon_i\bar V\le \bar V$ for all $t$ and $\epsilon$, dominated convergence gives
    \(
    \tilde p_i(\alpha^t v)=\mathbb E_\epsilon[p_i(b^t)]
    \to \mathbb E_\epsilon[p_i(\bar b)] =\tilde p_i(\bar\alpha v).
    \)
\end{itemize}
Thus each $\tilde p_i$ is continuous on $[0,1]^n$, hence $\Phi$ is continuous. Therefore, $\mathcal A_\delta=\Phi^{-1}\!\bigl(\prod_{i\in N}[0,B_i]\bigr)$ is closed. As a closed subset of the compact set $[0,1]^n$, $\mathcal A_\delta$ is compact.
\end{proof}

\subsubsection{Proof of \Cref{prop:pareto-dominant-BFM}}
\ParetoDominantBFM*
\begin{proof}
The set $\mathcal A_\delta$ is closed under component-wise maximum (\Cref{lem:join-closure} applies verbatim to $\tilde p$). For each $i$, define the coordinate-wise supremum $(\alpha^{\sup})_i := \sup\{\alpha_i:\ \alpha\in \mathcal A_\delta\}$.
Fix $\rho>0$. For each $i$ there exists $\alpha^{(i)}\in\mathcal A_\delta$ such that $\alpha^{(i)}_i>(\alpha^{\sup})_i-\rho$. Let $\alpha^\rho:=\bigvee_{i=1}^n \alpha^{(i)}$. By join-closure, $\alpha^\rho\in\mathcal A_\delta$, and by construction $\alpha_i^\rho>(\alpha^{\sup})_i-\rho$ for all $i$.
Now take $\rho_k=1/k$ and consider $\alpha^{(k)}:=\alpha^{\rho_k}\in\mathcal A_\delta$. By Lemma~\ref{lem:compact-space}, $\mathcal A_\delta$ is compact, so $\{\alpha^{(k)}\}$ has a convergent subsequence. Relabel this subsequence as $\{\alpha^{(k)}\}$ and denote its limit by $\alpha^*\in\mathcal A_\delta$. For each $i$, we have $\alpha_i^{(k)}>(\alpha^{\sup})_i-\rho_k$, hence letting $k\to\infty$, gives $\rho_k\to 0$, which implies $\alpha_i^*\ge (\alpha^{\sup})_i$.
Since $\alpha^*\in\mathcal A_\delta$, also $\alpha_i^*\le (\alpha^{\sup})_i$ by definition of the supremum. Thus $\alpha_i^*=(\alpha^{\sup})_i$ for all $i$, and consequently $\alpha^*$ is coordinate-wise greatest in~$\mathcal A_\delta$.
\end{proof}

\subsubsection{Proof of \Cref{prop:existence-delta}}
\ExistenceDelta*
\begin{proof}
Consider the smoothed mechanism $\tilde{\mathcal M}_\delta$. By~\Cref{prop:cont-alloc-pymt}, for every buyer $i$ the function $\tilde p_i(\alpha_i,\alpha_{-i})$ is continuous in $\alpha_i$ for fixed $\alpha_{-i}$.
Suppose for contradiction that buyer $i$ is unnecessarily paced under $\alpha^*$, i.e., $\alpha_i^*<1$ and $\tilde p_i(\alpha^*)<B_i$. Let $\omega:= B_i-\tilde p_i(\alpha^*)>0$ be the leftover budget. By continuity of $\tilde p$ in $\alpha_i$, there exists $0<\Delta<1-\alpha_i^*$ such that for all $\hat\alpha_i\in(\alpha_i^*,\,\alpha_i^*+\Delta)$, $\tilde p_i(\hat\alpha_i,\alpha_{-i}^*)
< \tilde p_i(\alpha_i^*,\alpha_{-i}^*)+\omega= B_i.$
Pick any $\Delta'\in(0,\Delta)$ and define $\hat\alpha$ by $\hat\alpha_i=\alpha_i^*+\Delta'$ and
$\hat\alpha_{-i}=\alpha_{-i}^*$. Then by continuity, buyer $i$ remains budget-feasible, so $\tilde p_i(\hat \alpha_i,\hat \alpha_{-i})<B_i$.
Moreover, since $\hat\alpha_i>\alpha_i^*$, buyer $i$'s bid increases while all other bids are unchanged. Moreover, increasing $\alpha_i$ only raises buyer $i$'s bid while leaving all other bids fixed. Since Assumption~\ref{ass:pymt-opp-inverse-mnt} is preserved under smoothing (and hence holds for $\tilde p$), for every $j\ne i$ we have $\tilde p_j(\hat\alpha)\le \tilde p_j(\alpha^*)\le B_j$. Therefore $\hat\alpha\in\mathcal A_\delta$. This contradicts the maximality of $\alpha^*$, since $\hat\alpha\ge \alpha^*$ and $\hat\alpha\ne \alpha^*$. Hence, no buyer is unnecessarily paced under $\alpha^*$.
\end{proof}

\subsubsection{Proof of \Cref{thm:existence}}
\existence*
\begin{proof}
Fix a perturbation parameter $0<\delta<1$, and perturb mechanism $\mathcal M$ as described in~\Cref{section:mechanism-perturb}.
By~\Cref{prop:existence-delta}, the smoothed mechanism $\tilde{\mathcal M_{\delta}}$ admits an \textsf{Abs-PE} $\alpha^{(\delta)}$ with corresponding outcome rule $(x^{(\delta)},\,p^{(\delta)})$.
Since $\alpha^{(\delta)}\in[0,1]^n$ and $(x^{(\delta)},p^{(\delta)})$ lie in the bounded allocation/payment region, we can extract a sequence $\delta_k\to0$ along which
$\bigl(\alpha^{(\delta_k)},x^{(\delta_k)},p^{(\delta_k)}\bigr)\rightarrow
  \bigl(\alpha^*,x^*,p^*\bigr).$
By continuity of the payment rule in the smoothed mechanism, $\alpha^*$ remains budget‐feasible, and the no-unnecessary-pacing condition carries over in the limit:
since each $\alpha^{(\delta_k)}$ is budget-feasible, we have $p^{(\delta_k)}_i \le B_i$ for all $i$, so in the limit $p^*_i \le B_i$. Moreover,
if $\alpha^*_i < 1$, then for all $k$, $\alpha^{(\delta_k)}_i < 1$ and $p^{(\delta_k)}_i = B_i$. So by continuity, $p^*_i = B_i$. 

Finally, since each pair $(x^{(\delta_k)},p^{(\delta_k)})$ is itself the expectation
$\mathbb{E}_{\epsilon\sim U[1-\delta_k,1]^n}\bigl[x(\alpha^{(\delta_k)}\cdot\epsilon\,v), p(\alpha^{(\delta_k)} \cdot \epsilon\,v)\bigr],$
then it lies in the convex hull of deterministic outcomes of $\mathcal M$ at bid profiles arbitrarily close to $\alpha^*\cdot v$. 
As $\delta_k \to 0$, the bids converge to $\alpha^* \cdot v = b^*$, so the limit point $(x^*, p^*)$ lies in the convex closure of the outcomes of $\mathcal{M}$ at $b^*$. More precisely, we have
\(
  (x^*,p^*) \in \hat X(b^*)
  =\mathrm{conv}\Bigl\{\lim_{b^n\to \alpha^*v}(x(b^n),\,p(b^n))\Bigr\}.
\)
Therefore, $(x^*,p^*)\in \hat X(b^*)$, is a valid tie-breaking outcome for $\mathcal M$ on input $b^*=\alpha^*\cdot v$. Therefore, by~\Cref{def:tie-break-corresp} there exists a tie-breaking rule for $\mathcal{M}$ such that the outcome at $\alpha^* \cdot v$ is exactly $(x^*, p^*)$.
Finally, since $\alpha^*$ is a budget-feasible vector with no unnecessarily paced buyers under the tie-breaking rule that selects $(x^*, p^*)$, then it is an \abspe of the original mechanism $\mathcal{M}$.
\end{proof}

\subsubsection{Proof of \Cref{prop:uniqueness}}
\uniquness*
\begin{proof}
    Let $\alpha^1$ and $\alpha^2$ be two \abspe pacing vectors for $\mathcal M'$. Define $\bar\alpha:=\max(\alpha^1, \alpha^2)$. By Assumption~\ref{ass:pymt-opp-inverse-mnt} applied to $p'$, \Cref{lem:join-closure} implies that $\bar \alpha$ is budget-feasible.
    Next, we show that \abspe vectors are maximal budget-feasible vectors. Consider $\alpha^1$: there is no budget-feasible $\beta$ with $\beta\ge \alpha^1$ (coordinate-wise) and $\beta\neq \alpha^1$. Indeed, if such a $\beta$ existed, let $G=\{i:\beta_i>\alpha^1_i\}$ (since $\beta\ne \alpha^1$ then $G\ne \emptyset$). Then $\beta_{-G}=\alpha^1_{-G}$ and $\alpha^1_i<1$ for all $i\in G$ (since $\beta_i\le 1$). Hence, by no-unnecessary-pacing, we have $p'_i(\alpha^1 v)=B_i$ for all $i\in G$ and hence $\sum_{i\in G}p'_i(\alpha^1 v)>0$. Assumption~\ref{ass:collec-pymt-mnt} applied to $p'$ on $(\beta,\alpha^1)$ and group $G$ yields $\sum_{i\in G}p'_i(\beta v)>\sum_{i\in G}p'_i(\alpha^1 v)=\sum_{i\in G}B_i$, so it must be that some $i\in G$ has $p'_i(\beta v)>B_i$, which contradicts feasibility of $\beta$. 
    The same argument shows $\alpha^2$ is maximal among budget-feasible vectors.
    Finally, since $\bar\alpha$ is budget-feasible and $\bar\alpha\ge \alpha^1$, maximality forces $\bar\alpha=\alpha^1$. Similarly, $\bar \alpha=\alpha^2$. Therefore, $\alpha^1=\alpha^2$, proving uniqueness of the equilibrium pacing vector.
\end{proof}

\subsubsection{Proof of~\Cref{lem:rev-max-tiebreak}}
\revMax*
\begin{proof}
Let $\bar\alpha$ be the unique \abspe of $\mathcal M'$ and let $\mathcal A'=\{\alpha\in[0,1]^n:\ p'_i(\alpha v)\le B_i\ \forall i\}$.
By~\Cref{prop:uniqueness}, $\bar\alpha$ is maximal in $\mathcal A'$, and by join-closure $\mathcal A'$ is closed under coordinate-wise maximum.
Fix any $\alpha\in\mathcal A'$ and set $\beta=\bar\alpha\vee\alpha$. Then $\beta\in\mathcal A'$ and $\beta\ge \bar\alpha$, so maximality forces $\beta=\bar\alpha$, i.e., $\bar\alpha\ge \alpha$. Thus $\bar\alpha$ is the coordinate-wise greatest element of $\mathcal A'$.
Applying Assumption~\ref{ass:monotone-rev} to the bid profiles $\bar\alpha v\ge \alpha v$ yields $\sum_i p'_i(\bar\alpha v)\ge \sum_i p'_i(\alpha v)$ for all $\alpha\in\mathcal A'$.
\end{proof}

\subsubsection{Proof of \Cref{prop:shill-proof}}
To establish shill-proofness of the revenue-maximizing \abspe, we first prove the following lemma. 
\begin{lemma}\label{lem:alpha-existence-by-continuity}
    Consider a $\delta$-smoothed mechanism $\tilde {\mathcal M_{\delta}}=(\tilde x, \tilde p)$, and a set of buyers with perturbed valuations and budgets profiles $\tilde v$ and $B$ respectively. Then, for any budget infeasible pacing vector $\alpha$, there exists a budget-feasible pacing vector $\hat\alpha$ such that $\sum_{i\in N}\tilde p_i(\hat\alpha v)\ge \sum_i \min\{\tilde p_i(\alpha v),\,B_i\}$.
\end{lemma}

\begin{proof}
    Run the following algorithm initialized at $\alpha$: while there exists some buyer $i$ that violates the budget constraint under $\alpha$, i.e. $\tilde p_i(\alpha v) > B_i$, continuously decrease their pacing multiplier $\alpha_i$ until $\tilde p_i(\alpha_i v_i, \alpha_{-i}v_{-i})=B_i$ (existence follows from continuity and the intermediate value theorem, knowing that $\tilde p_i(0,\alpha_{-i}v_{-i})=0$ by IR). For any $j\ne i$, decreasing $\alpha_i$ weakly increases their payment (by Assumption~\ref{ass:pymt-opp-inverse-mnt}). Hence, $\Phi(\alpha)=\sum_k\min\{\tilde p_k(\alpha v),B_k\}$ is non-decreasing at every step. 
    The generated sequence of pacing multipliers is coordinate-wise non-increasing and bounded below in $[0,1]^n$, so it converges to some $\hat\alpha$. By continuity of the payment function, we have $\tilde p_i(\hat\alpha v)\le B_i$ for all $i$, and $\sum_i \tilde p_i(\hat\alpha v)=\Phi(\hat\alpha) \ge \Phi(\alpha)=\sum_i \min\{\tilde p_i(\alpha v),B_i\}.$
\end{proof}

We are now ready to prove \Cref{prop:shill-proof}. Since different sets of buyers are considered, we write $(x^{N},p^{N})$ to denote the allocation and payment rules induced by running $\mathcal M$ with participants $N$.
\shillProof*
\begin{proof}
    We will construct a tie-breaking rule $\tau$ under which an \abspe exists and the revenue-maximizing \abspe is shill-proof. Consider the mechanism $\mathcal M'=(x',p')$ that is equivalent to $\mathcal M$ under
    $\tau$. Instantiate $\mathcal M'$ on the set of real buyers $N$ with valuations $v^N$ and budgets $B^N$, and denote this base instance by $I=(N,v^N,B^N)$. Let $F=\{n+1,\dots,n+f\}$ be any set of fake buyers, with valuations $v^F$ and budgets $B^F$, and consider the shilled instance
    $\bar I=(N\cup F,(v^N,v^F),(B^N,B^F))$.
    
    Fix any tie-breaking selection on market $N\cup F$ under which an \abspe exists (guaranteed by~\Cref{thm:existence}), and let $\bar\alpha$ denote a corresponding \abspe of the shilled instance $\bar I$. We will define the tie-breaking rule $\tau_N$ on the base market $N$ in Step~3 and then extend these choices arbitrarily to obtain a global tie-breaking rule $\tau$ under which the revenue-maximizing \abspe is shill-proof.
    Let $\bar b=\bar\alpha\cdot (v^N,v^F)$ be the corresponding bid profile, and write $\bar\alpha_N\in[0,1]^n$ for the restriction of $\bar\alpha$ to the real buyers $N$.
    
    \emph{Step 1 (upper bound by capped revenue on $N$).}
    Consider the modified profile on $N\cup F$ in which fake buyers bid zero, i.e., bids are $(\bar\alpha_N\cdot v^N,\mathbf 0^F)$.
    Since $\mathbf 0^F \le \bar\alpha_F\cdot v^F$ coordinate-wise and own bids of $i\in N$ are unchanged, opponent inverse monotonicity gives
    \(
      \sum_{i\in N} p_i'^{\,N\cup F}(\bar b)
      \le \sum_{i\in N} p_i'^{\,N\cup F}(\bar\alpha_N\!\cdot\! v^N,\mathbf 0^F).
    \)
    By zero-bid invariance, the right-hand side equals the revenue obtained when running $\mathcal M'$ on the reduced market $N$:
    \(
      \sum_{i\in N} p_i'^{\,N\cup F}(\bar\alpha_N\!\cdot\! v^N,\mathbf 0^F) = 
      \sum_{i\in N} p_i'^{\,N}(\bar\alpha_N\!\cdot\! v^N).
    \)
    Moreover, budget feasibility of $\bar\alpha$ in $\bar I$ implies $p_i'^{\,N\cup F}(\bar b)\le B_i^N$ for all $i\in N$, hence $\sum_{i\in N} p_i'^{\,N\cup F}(\bar b) \le \sum_{i\in N} \min\{p_i'^{\,N}(\bar\alpha_N\!\cdot\! v^N),\,B_i^N\}.$
    
    \emph{Step 2 (projection in the smoothed mechanism on $N$).}
    Fix $0<\delta<1$ and consider the $\delta$-smoothed mechanism on the base market $N$, denoted $\tilde{\mathcal M}_\delta$ with payments $\tilde p^{\,N,(\delta)}$ and perturbed valuations $\tilde v^N$. By \Cref{lem:alpha-existence-by-continuity} applied to the smoothed instance $\tilde I=(N,\tilde v^N,B^N)$ and initialized at $\bar\alpha_N$, there exists a budget-feasible vector $\hat\alpha^{(\delta)}\in[0,1]^n$ such that 
    \(\sum_{i\in N} \tilde p_i^{\,N,(\delta)}(\hat\alpha^{(\delta)}\!\cdot\! v^N) \ge \sum_{i\in N} \min\Bigl\{\tilde p_i^{\,N,(\delta)}(\bar\alpha_N\!\cdot\! v^N),\,B_i^N\Bigr\}.\)
    
    \emph{Step 3 (letting $\delta\downarrow 0$ and returning to $\mathcal M'$ on $N$).}
    By compactness of $[0,1]^n$, extract a subsequence $\delta_k\downarrow0$ such that $\hat\alpha^{(\delta_k)}\to\dot\alpha$.
    Passing to a further subsequence if necessary (and relabeling), we may assume that all limits displayed below hold along the same sequence $\delta_k$.
    Along this subsequence, the smoothed outcomes on $N$ converge to some limit point $(x^{N*},p^{N*})$ at bids $\dot b=\dot\alpha\cdot v^N$, i.e.,
    \(
    \bigl(x^{N,(\delta_k)}(\hat\alpha^{(\delta_k)}\!\cdot\! v^N),\ \tilde p^{\,N,(\delta_k)}(\hat\alpha^{(\delta_k)}\!\cdot\! v^N)\bigr)
    \ \longrightarrow\ (x^{N*},p^{N*}).
    \)
    By the definition of the tie-breaking correspondence, we have $(x^{N*},p^{N*})\in \hat X^N(\dot b)$. Hence, we may define the tie-breaking rule $\tau$ on market $N$ so that it selects this limit outcome at $\dot b$. 
    More precisely, we may choose $\tau$ on market $N$ so that at the bid profile $\dot b$ it selects this particular limit outcome, i.e.,
    \(
    (x'^{\,N}(\dot b),p'^{\,N}(\dot b)) := \tau_N(\dot b) = (x^{N*},p^{N*}).
    \)
    With this choice, the induced payments satisfy $p'^{\,N}(\dot b)=p^{N*}$.
    Moreover, along the same subsequence, we may assume that the capped-payment vector at the profile $\bar\alpha_N$ converges as well: $ \tilde p^{\,N,(\delta_k)}(\bar\alpha_N\!\cdot\! v^N) \to q$ for some $q\in\mathbb R_+^n.$
    Given $\tilde p^{\,N,(\delta_k)}(\bar\alpha_N\!\cdot\! v^N)=\mathbb E_{\epsilon}[p(\bar\alpha_N\!\cdot\! \epsilon \, v^N)]$, since $\bar\alpha_N\!\cdot\!\epsilon\, v^N\to \bar\alpha_N\!\cdot\!v^N$  and each smoothed payment vector $\tilde p^{\,N,(\delta_k)}(\bar\alpha_N\!\cdot\! v^N)$ is a convex combination of payment vectors $p^N(b)$ at bids $b$ arbitrarily close to $\bar\alpha_N\!\cdot\!v^N$, it follows from the definition of the tie-breaking correspondence that there exists some allocation $x^q$ such that $(x^q,q)\in \hat X^N(\bar\alpha_N\!\cdot\!v^N)$. We therefore define $\tau_N(\bar\alpha_N\!\cdot\!v^N):=(x^q,q)$.
    Finally, for all other bid profiles $b$ on market $N$, define $\tau_N(b)$ arbitrarily by selecting any element of $\hat X^N(b)$.
    
    With these choices, continuity of $t\mapsto \min\{t,B_i^N\}$ implies
    \(
    \sum_{i\in N}\min\Bigl\{\tilde p_i^{\,N,(\delta_k)}(\bar\alpha_N\!\cdot\! v^N),\,B_i^N\Bigr\}
    \longrightarrow
    \sum_{i\in N}\min\Bigl\{p_i'^{\,N}(\bar\alpha_N\!\cdot\!v^N),\,B_i^N\Bigr\}.
    \)
    Taking $\limsup$ in the inequality obtained in Step~2 along $\delta_k$ and using convergence yields
    \(
    \sum_{i\in N} \min\Bigl\{p_i'^{\,N}(\bar\alpha_N\!\cdot\!v^N),\,B_i^N\Bigr\}
    \le\; \limsup_{k\to\infty}\ \sum_{i\in N} \tilde p_i^{\,N,(\delta_k)}(\hat\alpha^{(\delta_k)}\!\cdot\! v^N) = 
    \sum_{i\in N} p_i^{N*} = 
    \sum_{i\in N} p_i'^{\,N}(\dot b).
    \)
    Moreover, feasibility of $\hat\alpha^{(\delta_k)}$ implies $\tilde p_i^{\,N,(\delta_k)}(\hat\alpha^{(\delta_k)}\!\cdot\! v^N)\le B_i^N$ for all $i$ and $k$, and passing to the limit gives $p_i^{N*}\le B_i^N$ for all $i\in N$. Hence $\dot\alpha\in\mathcal A'(N)$.
    
    Let $\alpha^{*_\tau}$ be a revenue-maximizing vector in $\mathcal A'(N)$ under $\tau$. Combining with the upper bound on $\sum_{i\in N} p_i'^{\,N\cup F}(\bar b)$ obtained in Step~1 gives
    \(
    \sum_{i\in N} p_i'^{\,N\cup F}(\bar b) \le \sum_{i\in N} p_i'^{\,N}(\alpha^{*_\tau}\!\cdot\!v^N),
    \)
    which is the desired shill-proofness inequality.
\end{proof}


\subsection{Convergence Proofs of the Adaptive Budget Dynamics}

\subsubsection{Proof of~\Cref{thm:borgs-conv}: Convergence of \Cref{alg:borgs}}
\borgs*
 We prove~\Cref{thm:borgs-conv} through a sequence of lemmas. The proof technique is similar in spirit to that in~\cite{borgs2007dynamics}, but adapted to our smoothed abstract-mechanism setting. 
    Fix $\delta\in(0,1)$ and consider the smoothed mechanism $\tilde{\mathcal M}_\delta$. Our goal is to show that the budget-adjustment dynamics converge to a $\gamma$-approximate \abspe of $\tilde{\mathcal M}_\delta$. 

    We first verify that common-scale responsiveness condition is preserved under smoothing. For any $\lambda\ge 1$, $\tilde p_i(\lambda\alpha v)
    \coloneqq \mathbb E_{\epsilon}\!\left[p_i(\lambda\alpha\tilde v)\right]
    \ge \lambda^{\kappa_i} \mathbb E_{\epsilon}\!\left[p_i(\alpha\tilde v)\right]
    = \lambda^{\kappa_i}\tilde p_i(\alpha v),$ where the inequality follows from Assumption~\ref{ass:homog}. Equivalently, for every $\theta\in(0,1]$, $\tilde p_i(\theta\alpha v) \le \theta^{\kappa_i}\tilde p_i(\alpha v).$
    
    We are now ready to begin the main argument. We first show that common-scale responsiveness and opponent inverse monotonicity of payments ensure that the algorithm’s updates move in the correct direction. More precisely, if a buyer’s payment exceeds (resp. falls short of) her budget at some iteration, then it must decrease (resp. increase) at the next iteration.

    \begin{lemma}\label{lem:borgs-pymt-update}
        If $\tilde p_i(\alpha^{t-1}v)>B_i$, then $\tilde p_i(\alpha^tv) \le \tilde p_i(\alpha^{t-1}v)e^{-\kappa_i\eta}.$
        Otherwise, if $\tilde p_i(\alpha^{t-1}v)\le B_i$ and the upward update of buyer $i$ is not clipped at $1$, then $\tilde p_i(\alpha^tv) \ge \tilde p_i(\alpha^{t-1}v)e^{\kappa_i\eta}.$
    \end{lemma}
    \begin{proof}
        Fix a buyer $i$. Suppose first that $\tilde p_i(\alpha^{t-1}v)>B_i$. By the update step, $\alpha_i^t=\alpha_i^{t-1}e^{-\eta}$, while for every other buyer $\alpha_{-i}^t\ge \alpha_{-i}^{t-1}e^{-\eta}$. Hence, \[\tilde p_i(\alpha^tv)= \tilde p_i\!\left(\alpha_i^{t-1}e^{-\eta}v_i, \alpha_{-i}^tv_{-i}\right) \le \tilde p_i\!\left(\alpha_i^{t-1}e^{-\eta}v_i, \alpha_{-i}^{t-1}e^{-\eta}v_{-i}\right)=\tilde p_i(e^{-\eta}\alpha^{t-1}v) \le e^{-\kappa_i\eta}\tilde p_i(\alpha^{t-1}v),\] where the first inequality follows from
        Assumption~\ref{ass:pymt-opp-inverse-mnt} and the second from Assumption~\ref{ass:homog}.
        
        Now suppose that $\tilde p_i(\alpha^{t-1}v)\le B_i$ and buyer $i$'s upward update is not clipped. Then $\alpha_i^t=\alpha_i^{t-1}e^\eta$, while $\alpha_{-i}^t\le\alpha_{-i}^{t-1}e^\eta$. Therefore, $\tilde p_i(\alpha^tv) = \tilde p_i\!\left(\alpha_i^{t-1}e^\eta v_i,\alpha_{-i}^tv_{-i}\right) \ge \tilde p_i\!\left( \alpha_i^{t-1}e^\eta v_i,\alpha_{-i}^{t-1}e^\eta v_{-i}\right) = \tilde p_i(e^\eta\alpha^{t-1}v) \ge e^{\kappa_i\eta}\tilde p_i(\alpha^{t-1}v)$, where the first inequality follows again from Assumption~\ref{ass:pymt-opp-inverse-mnt} and the second from Assumption~\ref{ass:homog}.
        
\end{proof}
    
    Next, we upper bound the difference between payments at consecutive iterations.
    \begin{lemma}\label{lem:borgs-pymt-bound}
        For all $t$ and $i$, we have $\bigl|\tilde p_i(\alpha^tv)-\tilde p_i(\alpha^{t-1}v)\bigr| \le \eta \, C B_i$, where $C:=\frac{n\bar V}{\delta\,\min_i B_i}$. 
    \end{lemma}
    \begin{proof}
        For every buyer $j$, the update rule implies $\bigl|\log\alpha_j^{t}-\log\alpha_j^{t-1}\bigr|\le \eta$, hence $\|\log\alpha^{t}-\log\alpha^{t-1}\|_\infty\le \eta$. Since $\tilde p_i$ is Lipschitz continuous on $(0,1]^n$ in the log domain (see the proof of~\Cref{lem:compact-space})   yields
        \(
        \bigl|\tilde p_i(\alpha^tv)-\tilde p_i(\alpha^{t-1}v)\bigr|
        \le \frac{n\bar V}{\delta}\,\|\log\alpha^{t}-\log\alpha^{t-1}\|_\infty
        \le \frac{n\bar V}{\delta}\,\eta
        \le \eta\,C B_i.
        \)
    \end{proof}

    \begin{corollary}\label{cor:borgs-pymt-upper-bound}
        If $\tilde p_i(\alpha^{t-1}v) \le B_i$ then $\tilde p_i(\alpha^tv)\leq B_i(1+\gamma)$ provided $\eta$ is chosen such that $\eta\,C\le \gamma$.
    \end{corollary}
    \begin{proof}
        Since $\tilde p_i(\alpha^{t-1}v) \le B_i$, it directly follows from~\Cref{lem:borgs-pymt-bound} that $\tilde p_i(\alpha^tv)\leq \tilde p_i(\alpha^{t-1}v)+ \eta\, C B_i \le \tilde p_i(\alpha^{t-1}v)+\gamma B_i \le B_i(1+\gamma)$.
    \end{proof}
    
    Next, we show that after some finite number of iterations, no buyer exceeds her budget by more than a factor of $(1+\gamma)$. 
    \begin{lemma}
        For every $0\le \gamma \le 1$, $\eta \le \frac{\gamma}{C}$, and $t\geq \frac{1}{\underline\kappa\eta}\ln(\frac{\bar V/\min_iB_{i}}{1+\gamma})\coloneqq T_{\gamma}^{\text{BF}}$, we have $\tilde p_i(\alpha^tv)\le B_i(1+\gamma)$. 
    \end{lemma}
    \begin{proof}
        From~\Cref{lem:borgs-pymt-update}, we have $\tilde p_i(\alpha^t v) \le \tilde p_i(\alpha^{t-1}v)e^{-\kappa_i\eta} \le \tilde p_i(\alpha^{t-1}v)e^{-\underline\kappa\eta}$ if $\tilde p_i(\alpha^{t-1}v)>B_i$. 
        Moreover,~\Cref{cor:borgs-pymt-upper-bound}, gives $\tilde p_i(\alpha^tv)\le B_i(1+\gamma)$ if $\tilde p_i(\alpha^{t-1}v) \le B_i$. 
        These two inequalities imply that $\tilde p_i(\alpha^tv)\le \max\{\tilde p_i(\alpha^{t-1}v)e^{-\underline\kappa\eta}, \; B_i(1+\gamma)\}$. 
        Moreover, by individual rationality, we have $\tilde p_i(\alpha^0v)\le  \bar V \le \frac{\bar V}{\min_i B_i}B_i$.
        It follows that for $t \ge \frac{1}{\underline \kappa\eta} \ln (\frac{\bar V/\min_i B_i}{1+\gamma})$, we have $\tilde p_i(\alpha^tv)\le (1+\gamma)B_i$.
    \end{proof}
    
    Finally, we show that after some finite iterations, no buyer is ``too much unnecessarily paced''. 
    \begin{lemma}
         For all $t \geq 2\,T_{\gamma}^{\text{BF}}-\frac1\eta\ln(\min_i\alpha_i^0)$ and all $i$, one of the following holds:
         \begin{align}
             \tilde p_i(\alpha^tv) &\ge (1-\gamma) B_i \label{eq:borgs-unnec-pac-1} \\
             \alpha_i^t &\ge e^{-\eta}\label{eq:borgs-unnec-pac-2}
         \end{align}
    \end{lemma}
    
    \begin{proof}
        We prove the lemma by backwards induction. First, suppose neither \eqref{eq:borgs-unnec-pac-1} nor \eqref{eq:borgs-unnec-pac-2} holds on iteration $t$ and $t - 1 \ge T_{\gamma}^{\text{BF}}$. We will show that neither inequality holds on iteration $t-1$. 
        
        Since \eqref{eq:borgs-unnec-pac-1} doesn't hold  at iteration $t$, then~\Cref{lem:borgs-pymt-bound} gives $\tilde p_i(\alpha^{t-1}v) \le \tilde p_i(\alpha^tv) + \eta\,C B_i \le (1-\gamma)B_i+\gamma B_i = B_i$ given that $\eta\,C \le \gamma$. 
        Moreover, since \eqref{eq:borgs-unnec-pac-2} doesn't hold at iteration $t$, we have $\alpha_i^t<e^{-\eta}<1$. Hence the upward update from $t-1$ to $t$ remains below 1. Since $\tilde p_i(\alpha^{t-1}v)\le B_i$, \Cref{lem:borgs-pymt-update} gives $\tilde p_i(\alpha^tv) \ge e^{\kappa_i\eta}\tilde p_i(\alpha^{t-1}v) \ge \tilde p_i(\alpha^{t-1}v).$ Therefore, since $\tilde p_i(\alpha^tv)<(1-\gamma)B_i$, we also have $\tilde p_i(\alpha^{t-1}v)<(1-\gamma)B_i$, so \eqref{eq:borgs-unnec-pac-1} does not hold at iteration $t-1$.
        Also since $\alpha_i^t=\alpha_i^{t-1}e^\eta$, then $\alpha_i^{t-1}=\alpha_i^{t}e^{-\eta} < e^{-2\eta}<e^{-\eta}$ so \eqref{eq:borgs-unnec-pac-2} also fails at iteration \(t-1\). 
    
        For the base case, notice that as long as neither \eqref{eq:borgs-unnec-pac-1} nor \eqref{eq:borgs-unnec-pac-2} holds, then $\alpha_i^t=\alpha_i^{t-1}e^{\eta}$. Hence, for $t\ge 2\,T_{\gamma}^{\text{BF}}-\frac1\eta\ln(\min_i\alpha_i^0)$, inequality \eqref{eq:borgs-unnec-pac-2} will hold.
    \end{proof}

\subsubsection{Proof of~\Cref{thm:MW-avg-conv}: Convergence proof of \Cref{alg:MW}}
\MWconv*
We prove~\Cref{thm:MW-avg-conv} by establishing two main results. First, we show that MW achieves approximate budget feasibility in time average. Second, we prove that the fraction of iterations in which a buyer is approximately unnecessarily paced vanishes as the horizon grows. Together, these imply that the defining \abspe conditions hold on average, and hence MW converges on average to a $\gamma$-approximate \abspe of the smoothed mechanism.

    First, we verify that the own-bid responsiveness condition on payments is preserved under smoothing. Fix $\alpha_{-i}$ and let $\alpha_i>\alpha_i'$. For each realization $\epsilon\sim (\text{Unif}[1-\delta,1])^n$, we have $\tilde v=\epsilon v\in V$, so Assumption~\ref{ass:mult-OBM} implies
    \(
    p_i(\alpha_i \tilde v_i,\alpha_{-i}\tilde v_{-i}) \ge
    p_i(\alpha_i' \tilde v_i,\alpha_{-i}\tilde v_{-i}) \Bigl(\tfrac{\alpha_i}{\alpha_i'}\Bigr)^{\mu_i}.
    \)
    Taking expectations over $\epsilon$ and using linearity yields 
    \(
      \tilde p_i(\alpha_i v_i,\alpha_{-i} v_{-i}) \ge 
      \tilde p_i(\alpha_i' v_i,\alpha_{-i} v_{-i})\cdot\Bigl(\tfrac{\alpha_i}{\alpha_i'}\Bigr)^{\mu_i}.
    \)
    
    The next Lemma establishes a uniform lower bound on MW iterates, which will be used to prove average budget feasibility. 
    
    \begin{lemma}\label{lem:mw-alpha-lower} 
        Consider the update step of~\Cref{alg:MW}. For every buyer $i$, there exists a constant $\bar \alpha_i$ (depending only on $i,\eta,B_i,\tilde p_i,\alpha_i^0$) such that $\alpha_i^t \ge \bar\alpha_i$ for all $t\ge 0$.
    \end{lemma}

    \begin{proof}
        Define $G_i^+:=\sup_{\alpha\in[0,1]^n}\,(\tilde p_i(\alpha)-B_i)^+$ (this is finite since $G_i^+ \le G \,\forall i$). 
        We will show that $\bar\alpha_i = \min\Bigl\{\alpha_i^0,\;\hat \alpha_i\,e^{-\eta\,G_i^+}\Bigr\},$ where $\hat \alpha_i>0$ is any value satisfying $sup_{\alpha_{-i}\in[0,1]^{n-1}} \tilde p_i\bigl(\hat \alpha_iv_i,\alpha_{-i}v_{-i}\bigr) \le \frac{B_i}{2}.$
        Define
        \(
        F_i(a) = \sup_{\alpha_{-i}\in[0,1]^{n-1}} \tilde p_i(av_i,\alpha_{-i}v_{-i}), \ a\in[0,1].
        \)
        By continuity of $\tilde p_i$ and compactness of $[0,1]^{n-1}$, $F_i$ is continuous on $[0,1]$. Moreover, $\tilde p_i(0,\alpha_{-i}v_{-i})=0$ for all $\alpha_{-i}$ implies $F_i(0)=0$. Since $F_i(0)=0$ and $B_i>0$, then by continuity there exists $\hat \alpha_i>0$ such that the condition holds.
    
        Fix $t\ge 0$. We consider the two cases.
        \begin{itemize}
            \item \emph{Case 1: $\alpha_i^t>\hat \alpha_i$.}
            If $g_{t,i}\le 0$, then $\alpha_i^{t+1}=\min\{\alpha_i^t e^{-\eta g_{t,i}},1\}\ge\alpha_i^t > \hat \alpha_i$.
            If $g_{t,i}>0$, then $\alpha_i^{t+1} = \alpha_i^t e^{-\eta g_{t,i}} \ge \hat \alpha_i\, e^{-\eta G_i^+}$,
            since $g_{t,i}\le (\tilde p_i-B_i)^+ \le G_i^+$. Hence whenever $\alpha_i^t>\hat \alpha_i$,
            \(
            \alpha_i^{t+1} \ge \min\{\hat \alpha_i,\,\hat \alpha_ie^{-\eta G_i^+}\}
            \ge \hat \alpha_ie^{-\eta G_i^+}.
            \)
    
            \item \emph{Case 2: $\alpha_i^t\le \hat \alpha_i$.} 
            By the own-bid responsiveness condition on payments, we have, for $\alpha_i^t\le \hat\alpha_i$,
            \(
              \tilde p_i(\alpha_i^t v_i,\alpha_{-i}^t v_{-i})
              \le \tilde p_i(\hat\alpha_i v_{i},\alpha_{-i}^t v_{-i})\Bigl(\tfrac{\alpha_i^t}{\hat\alpha_i}\Bigr)^{\mu_i}
              \le \tilde p_i(\hat\alpha_i v_{i},\alpha_{-i}^t v_{-i})
              \le F_i(\hat\alpha_i)
              \le \frac{B_i}{2}.
            \)
            Hence $g_{t,i}=\tilde p_i(\alpha^t v)-B_i\le -B_i/2$.
            Let $a := \alpha_i^t e^{-\eta g_{t,i}}$. Then $g_{t,i}<0$ implies $a = \alpha_i^t e^{-\eta g_{t,i}} \ge \alpha_i^t.$ By the update step of the algorithm, $\alpha_i^{t+1} = \min\{a,1\}.$ If $a\le 1$, then $\alpha_i^{t+1}=a\ge\alpha_i^t$. If $a>1$, then $\alpha_i^{t+1}=1\ge\alpha_i^t$ (since $\alpha_i^t\le 1$). In either case, $\alpha_i^{t+1} \ge \alpha_i^t.$ Thus whenever $\alpha_i^t\le\hat\alpha_i$, the sequence $(\alpha_i^t)_t$ is nondecreasing at $t$.
        \end{itemize}
        Now combine these observations to bound the global minimum over all $t$. 
        If $\alpha_i^t>\hat \alpha_i$ for all $t$, then from Case 1 applied at $t=0$, we have
        \(
        \alpha_i^t \ge \min\{\alpha_i^0,\,\hat \alpha_ie^{-\eta G_i^+}\}= \bar \alpha_i \quad \forall t\ge 0.
        \)
        Otherwise, let $t_0:=\min\{t\ge 0,\, \alpha_i^t \le \hat \alpha_i\}$ be the first time at which $\alpha_i^t$ enters $[0,\hat\alpha_i]$.
    
        \begin{itemize}
          \item If $t_0=0$, then $\alpha_i^0\ge\bar\alpha_i$ by definition of $\bar\alpha_i$, and by Case 2 we have $\alpha_i^{t+1}\ge\alpha_i^t$ whenever $\alpha_i^t\le\hat\alpha_i$. Hence
          \(
          \alpha_i^t \ge \alpha_i^0 \ge \bar\alpha_i \ \forall\,t\ge 0 \text{ with }\alpha_i^t\le\hat\alpha_i.
          \)
          \item If $t_0\ge 1$, then $\alpha_i^{t_0-1}>\hat\alpha_i$ and by Case 1,
          \(
            \alpha_i^{t_0}
            = \alpha_i^{t_0-1} e^{-\eta g_{t_0-1,i}}
            \ge \hat\alpha_i e^{-\eta G_i^+}.
          \)
          For all $t\ge t_0$ with $\alpha_i^t\le\hat\alpha_i$, and Case 2 yields
          \(
            \alpha_i^t \ge \alpha_i^{t_0} \ge \hat\alpha_i e^{-\eta G_i^+}.
          \)
        \end{itemize}
        In both subcases, whenever $\alpha_i^t\le\hat\alpha_i$, we have
        \(
        \alpha_i^t \ge \min\{\alpha_i^0,\,\hat\alpha_i e^{-\eta G_i^+}\} = \bar\alpha_i.
        \)
        If later $\alpha_i^t$ goes back above $\hat \alpha_i$, then trivially $\alpha_i^t\ge \hat \alpha_i\ge \hat \alpha_ie^{-\eta G_i^+}$. In all cases, for any sequence of opponents' bids $(\alpha_{-i}^t)_t$, we have 
        \(
        \alpha_i^t \ge \bar\alpha_i = \min\{\alpha_i^0,\;\hat \alpha_ie^{-\eta G_i^+}\} > 0
        \quad \forall t\ge 0,
        \)
        which completes the proof.
    \end{proof}

    We first prove the average-budget-feasibility claim in part~(1) of \Cref{thm:MW-avg-conv} using the bound in~\Cref{lem:mw-alpha-lower}.

    \begin{proof}{\emph{Proof of part~(1) of~\Cref{thm:MW-avg-conv}. }}
        Let $u_i^t:=\ln \alpha_i^t$. By~\Cref{lem:mw-alpha-lower}, $\alpha_i^t\in[\bar\alpha_i,1]$ so $u_i^t\in[\ln\bar\alpha_i,0]$ for all $t$. From the update step, we have $u_i^{t+1} \le u_i^t - \eta\, g_{t,i}.$
        Summing over $t=0,\dots,T-1$ gives
        \(
        \sum_{t=0}^{T-1} g_{t,i} \le \frac{u_i^0-u_i^T}{\eta}
        \le \frac{\ln(1/\bar\alpha_i)}{\eta},
        \)
        since $u_0-u_T\le 0-\ln\bar\alpha_i=\ln(1/\bar\alpha_i)$. Dividing by $T$ yields
        \(
        \frac{1}{T}\sum_{t=0}^{T-1} g_{t,i} \le \frac{\ln(1/\bar\alpha_i)}{\eta\, T}.
        \)
        Using $g_{t,i}=\tilde p_i(\alpha^tv)-B_i$ gives the stated bound on the average payment. If
        $T \ge \ln(1/\bar\alpha_i)/(\eta\,\gamma\,B_i)$, then $\frac{\ln(1/\bar\alpha_i)}{\eta T}\le \gamma B_i$, so
        \(
        \frac{1}{T}\sum_{t=0}^{T-1} \tilde p_i(\alpha^tv)
        \le B_i+\gamma B_i=(1+\gamma)B_i,
        \)
        which gives approximate budget feasibility on average.
    \end{proof}

    Next, we establish approximate no-unnecessary-pacing on average. The following bound controls the movement of MW over consecutive iterates in the log domain and will be used to bound the change in expected payments between consecutive iterates.
    
    \begin{lemma}\label{lem:mw-log-step}
        Let $G>0$ be such that $|g_{t,i}|\le G$ for all $t,i$.
        Under the MW update step with initialization $\alpha^0\in(0,1]^n$, we have 
        \(
          \bigl\|\log \alpha^{t}-\log \alpha^{t-1}\bigr\|_\infty \le \eta G,
        \)
        where $\log$ is applied coordinate-wise.
        Moreover, for every buyer $i$, $\left(\log\alpha_i^t-\log\alpha_i^{t-1}\right)_+\le \eta B_i.$
    \end{lemma}

    \begin{proof}
        Fix a round $t\ge 1$ and a coordinate $i$. Since $\alpha^0\in(0,1]^n$ and the update is
        multiplicative with an upper cap at $1$, we have $\alpha^{t-1}_i\in(0,1]$ and thus
        $\log\alpha^{t-1}_i$ is well-defined.
        There are two cases.
        \begin{enumerate}
            \item If $\alpha_i^{t-1}e^{-\eta g}\le 1$, then $\alpha_i^{t}=\alpha_i^{t-1}e^{-\eta g_{t-1,i}}$ and hence
            \(
              \log\alpha_i^{t}-\log\alpha_i^{t-1}
              = \log(\alpha_i^{t-1}e^{-\eta g_{t-1,i}})-\log\alpha_i^{t-1}
              = -\eta g_{t-1,i},
            \)
            so $|\log\alpha_i^{t}-\log\alpha_i^{t-1}|=\eta|g_{t-1,i}|\le \eta G$.
            
            \item If $\alpha_i^{t-1}e^{-\eta g_{t-1,i}}>1$, then $\alpha_i^{t}=1$ and $\log\alpha_i^{t}=0$. The clipping condition implies
            \(
              0 < \log(\alpha_i^{t-1}e^{-\eta g_{t-1,i}})
              = \log\alpha_i^{t-1}-\eta g_{t-1,i},
            \)
            and since clipping can only occur when $g_{t-1,i}<0$, we have $-\eta g_{t-1,i}=\eta|g_{t-1,i}|$. Rearranging yields $-\log\alpha_i^{t-1} < \eta|g_{t-1,i}|$.
            Therefore,
            \(
              \bigl|\log\alpha_i^{t}-\log\alpha_i^{t-1}\bigr|
              = \bigl|0-\log\alpha_i^{t-1}\bigr|
              = -\log\alpha_i^{t-1}
              < \eta|g_{t-1,i}|
              \le \eta G.
            \)
        \end{enumerate}
        In both cases, $|\log\alpha_i^{t}-\log\alpha_i^{t-1}|\le \eta G$ for every $i$. Taking the maximum over coordinates gives
        \(
          \|\log\alpha^{t}-\log\alpha^{t-1}\|_\infty
          = \max_i |\log\alpha_i^{t}-\log\alpha_i^{t-1}|
          \le \eta G,
        \)
        as claimed.
        
        For the second claim, an increase in $\log\alpha_i$ can occur only when $g_{t-1,j}<0$. In that case, including the possibility of clipping at $1$, we have $\left(\log\alpha_i^t-\log\alpha_i^{t-1}\right)_+ \le \eta(-g_{t-1,i}) = \eta\bigl(B_i-\tilde p_i(\alpha^{t-1}v)\bigr) \le\eta B_i,$ where the final inequality follows from nonnegativity of payments.
        
    \end{proof}  

    It remains to prove the vanishing-unnecessary-pacing claim in part~(2) of~\Cref{thm:MW-avg-conv}. For this, we use the log-step bound in~\Cref{lem:mw-log-step}.

    \begin{proof}{\emph{Proof of part~(2) of~\Cref{thm:MW-avg-conv}. }}
        We first record a directional refinement of \Cref{prop:log-lip}. Fix buyer $i$ and fix her multiplier at $\alpha_i\in(0,1]$. For two opponent multiplier profiles $\alpha_{-i},\alpha_{-i}'\in(0,1]^{n-1}$ with $\alpha_{-i}\le \alpha'_{-i}$ coordinate-wise, $\tilde p_i(\alpha_i v_i,\alpha_{-i}v_{-i}) -\tilde p_i(\alpha_i v_i,\alpha_{-i}'v_{-i}) \le \frac{\alpha_i\bar V}{\delta}\sum_{j\neq i}\left(\log\alpha_j'-\log\alpha_j \right)_+ .$
        Indeed, the one-coordinate argument in \Cref{prop:log-lip}, together with individual rationality, gives a log-Lipschitz constant $\alpha_i\bar V/\delta$ when buyer $i$'s own multiplier is fixed at $\alpha_i$, and Assumption~\ref{ass:pymt-opp-inverse-mnt} implies that decreases in opponents' multipliers cannot decrease buyer $i$'s payment. 

        Next, fixing a buyer $i$ again, we first show that the iterations at which $i$ is unnecessarily paced propagate backwards in time. Specifically, if $t\ge 1$ and $t\in\mathcal N_i(\gamma)$, then $t-1\in\mathcal N_i(\gamma)$.
        Assume $t\in\mathcal N_i(\gamma)$. By definition,
        \(
          \tilde p_i(\alpha^tv)< (1-\gamma)B_i, \ \alpha_i^t < 1-\gamma.
        \)

        \emph{Step 1: Backward propagation of the gradient condition.}
        Let $\hat\alpha :=(\alpha_i^t,\alpha_{-i}^{t-1}).$ Since $t\in\mathcal N_i(\gamma)$, we have $\alpha_i^t<1-\gamma$. By the directional refinement above and \Cref{lem:mw-log-step}, $\tilde p_i(\hat\alpha v)-\tilde p_i(\alpha^t v) \le \frac{\alpha_i^t\bar V}{\delta}\sum_{j\neq i}\left(\log\alpha_j^t-\log\alpha_j^{t-1}\right)_+ \le \eta\frac{(1-\gamma)\bar V}{\delta}\sum_{j\neq i}B_j\le \zeta B_i,$ where the last inequality follows from the upper bound on $\eta$ and the definition of $R_B$.
        
        We claim that $g_{t-1,i}\le0$. Suppose instead that $g_{t-1,i}>0$. Then buyer $i$'s multiplier decreases, with $\log\alpha_i^{t-1}-\log\alpha_i^t= \eta g_{t-1,i}.$ The one-coordinate bound from~\Cref{prop:log-lip} gives $\tilde p_i(\alpha^{t-1}v)-\tilde p_i(\hat\alpha v) \le\frac{\bar V}{\delta}\eta g_{t-1,i}.$ Combining this inequality with the upper bound on $\tilde p_i(\hat\alpha v)-\tilde p_i(\alpha^t v) $ yields $\tilde p_i(\alpha^{t-1}v)-\tilde p_i(\alpha^t v) \le \frac{\eta\bar V}{\delta}g_{t-1,i} +\zeta B_i.$ 
        On the other hand, $\tilde p_i(\alpha^{t-1}v)-\tilde p_i(\alpha^t v) > \gamma B_i+g_{t-1,i},$ since $\tilde p_i(\alpha^{t-1}v)=B_i+g_{t-1,i}$ and $\tilde p_i(\alpha^tv)<(1-\gamma)B_i$. Given $\eta\le\delta/\bar V$, these two inequalities imply $\gamma B_i+g_{t-1,i} < g_{t-1,i}+\zeta B_i,$ contradicting $\zeta<\gamma$. Hence $g_{t-1,i}\le0$.

        Therefore, buyer $i$ weakly increases her multiplier between $t-1$ and $t$. By Assumption~\ref{ass:mult-OBM}, $\tilde p_i(\alpha^{t-1}v) \le \tilde p_i(\hat\alpha v).$ Combining this with the upper bound on $\tilde p_i(\hat\alpha v)-\tilde p_i(\alpha^t v) $ gives $\tilde p_i(\alpha^{t-1}v) < (1-\gamma+\zeta)B_i,$ and therefore $g_{t-1,i} < -(\gamma-\zeta)B_i,$ meaning that $|g_{t-1,i}| > (\gamma-\zeta)B_i.$
        
        \emph{Step 2: Backward propagation of the pacing condition.}
        From $g_{t-1,i}\le 0$ and the update step,
        \(
          \alpha_i^t=\min\{\alpha_i^{t-1}e^{-\eta g_{t-1,i}},1\}\ge \alpha_i^{t-1}.
        \)
        Since $\alpha_i^t<1-\gamma$, we obtain $\alpha_i^{t-1}\le \alpha_i^t<1-\gamma$, so the pacing condition in the definition of $\mathcal N_i(\gamma)$ holds at $t-1$.

        \emph{Step 3: Backward propagation of the underpayment condition.}
        Let $\bar\alpha:= (\alpha_i^{t-1},\alpha_{-i}^t).$ By Step~2, $\alpha_i^{t-1}<1-\gamma.$
        Applying the inequality in Step~1 with buyer $i$'s own multiplier fixed at $\alpha_i^{t-1}$ gives $\tilde p_i(\alpha^{t-1}v)-\tilde p_i(\bar\alpha v) \le \frac{\alpha_i^{t-1}\bar V}{\delta} \sum_{j\neq i}\left(\log\alpha_j^t-\log\alpha_j^{t-1} \right)_+ \le \eta\frac{(1-\gamma)\bar V}{\delta}\sum_{j\neq i}B_j \le \zeta B_i.$
        
        From Step~1, $g_{t-1,i}<0.$ Since $\alpha_i^t<1-\gamma<1$, the update does not clip, and hence $\frac{\alpha_i^t}{\alpha_i^{t-1}} = e^{\eta|g_{t-1,i}|}.$
        Applying Assumption~\ref{ass:mult-OBM} with opponents fixed at $\alpha_{-i}^t$ gives $\tilde p_i(\bar\alpha v) \le \tilde p_i(\alpha^t v) e^{-\eta\mu_i|g_{t-1,i}|}.$ Therefore $\tilde p_i(\alpha^{t-1}v) \le \tilde p_i(\alpha^t v) e^{-\eta\mu_i|g_{t-1,i}|} +\zeta B_i.$ 
        Using $\tilde p_i(\alpha^t v)<(1-\gamma)B_i$ and $|g_{t-1,i}| > (\gamma-\zeta)B_i,$ we obtain $\tilde p_i(\alpha^{t-1}v) < (1-\gamma)B_i e^{-\eta\mu_i(\gamma-\zeta)B_i} +\zeta B_i.$ The lower bound on $\eta$ implies $e^{-\eta\mu_i(\gamma-\zeta)B_i} \le \frac{1-\gamma-\zeta}{1-\gamma},$ and hence $\tilde p_i(\alpha^{t-1}v) < (1-\gamma-\zeta)B_i+\zeta B_i = (1-\gamma)B_i.$

        \emph{Step 4: Unnecessary-pacing iterations are successive and finite.} If $\mathcal N_i(\gamma)=\emptyset$, there is nothing to prove. Otherwise, fix any $t\in\mathcal N_i(\gamma)$. By backward propagation, $\{0,1,\ldots,t\}\subseteq\mathcal N_i(\gamma).$ Hence, for every $k=0,\ldots,t-1$, $\alpha_i^{k+1}<1-\gamma<1.$ Thus the MW update does not clip on any of these transitions.
        Moreover, $g_{k,i} = \tilde p_i(\alpha^kv)-B_i < -\gamma B_i.$ Therefore $\alpha_i^{k+1} = \alpha_i^k e^{-\eta g_{k,i}} > \alpha_i^k e^{\eta\gamma B_i},\ k=0,\ldots,t-1.$ Iterating gives $\alpha_i^t > \alpha_i^0 e^{\eta\gamma B_i t}.$ 
        Since $t\in\mathcal N_i(\gamma)$ also implies $\alpha_i^t<1-\gamma$, we obtain $t < \frac{1}{\eta\gamma B_i}\ln\left(\frac{1-\gamma}{\alpha_i^0} \right):= C_i(\gamma).$
        Thus every element of $\mathcal N_i(\gamma)$ is bounded above by the same finite constant. Consequently $\mathcal N_i(\gamma)$ is finite, and for any horizon $T\ge 1$,
        \(
          \#\bigl\{t\in\{0,\dots,T-1\}:\ t\in\mathcal N_i(\gamma)\bigr\}
          \le C_i(\gamma)+1,
        \)
        and hence
        \(
          \frac{1}{T}\,\#\bigl\{t\le T-1:\ t\in\mathcal N_i(\gamma)\bigr\}
          \le \frac{C_i(\gamma)+1}{T}\xrightarrow[T\to\infty]{}0.
        \)
        \end{proof}

\section{Supplemental Experiments}
\subsection{Additional Results for Position-Auction Demand Gaps}\label{app:sec-exps-position}

\paragraph{Market thickness and the number of slots.}

Figure~\ref{fig:sec3-demand-violation} in Section~ref{sec:exps-position} shows that demand violations become less frequent and smaller as the buyer-to-slot ratio $n/s$ increases, holding the number of slots fixed at $s=5$. We ask whether the same pattern holds when the absolute number of slots also varies.
\Cref{fig:app-thickness} repeats the experiment for $s\in\{5,10,20\}$. Slot qualities are given by $q_k=0.9^{k-1}$, and we vary $n$ to obtain different buyer-to-slot ratios. As in the main experiment, buyer $i$ has a demand violation when $r_i>0.05$. 

\begin{figure}[h]
    \centering
    \includegraphics[width=\textwidth]{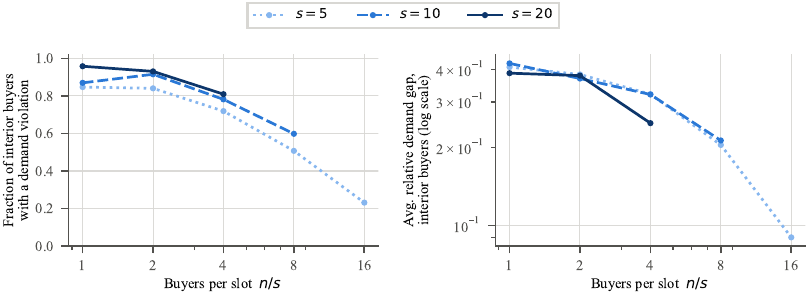}
    \caption{Demand violations across buyer-to-slot ratios and numbers of slots. The left panel reports the fraction of buyers with a demand violation. The right panel reports the average relative demand~gap.}
    \label{fig:app-thickness}
\end{figure}

\Cref{fig:app-thickness}  shows that both the fraction of buyers with demand violations and the average relative demand gap decline as the buyer-to-slot ratio \(n/s\) increases, for each value of \(s\).  Thus, demand violations become less prevalent and smaller as competition per slot increases. Conditional on the buyer-to-slot ratio, the absolute number of slots has a weaker and less systematic effect.

\paragraph{Demand gaps as a function of pacing multipliers and budget tightness.}

The previous experiment varies market structure while keeping the pacing regime approximately fixed. Instead, we now study how relative demand gaps vary with the equilibrium pacing vector, and choose the experimental parameters to obtain enough variation in the pacing multipliers. We set \(n=20\), \(m=10\), \(s=n=20\), \(q_k=0.8^{k-1}\), and \(B_i~=~\lambda_i\frac{1}{n}\sum_j v_{ij}\), where \(\lambda_i\sim U[0,14]\). Conditional on receiving slot \(k\) in auction~\(j\), buyer \(i\)'s utility at \posfppe prices is \(q_kv_{ij}(1-\alpha_i^*)\), and therefore decreases as \(\alpha_i^*\) increases.

\begin{figure}[h]
    \centering
    \includegraphics[width=\textwidth]{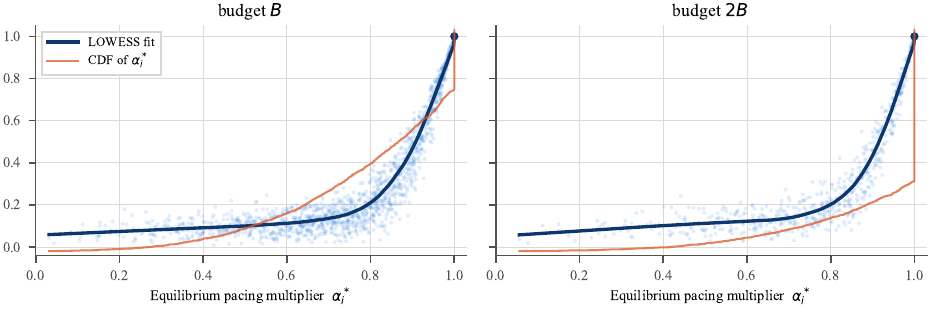}
    \caption{Demand gaps and equilibrium pacing multipliers. Each point represents a buyer-instance pair. The blue curve is a LOWESS fit of the relative demand gap against the equilibrium pacing multiplier, while the orange curve is the empirical CDF of the equilibrium pacing multiplier. Both are shown on a common \([0,1]\) vertical scale.}
    \label{fig:sec3-alpha-gain}
\end{figure}

The blue LOWESS curve in the left panel of Figure~\ref{fig:sec3-alpha-gain} shows that relative demand gaps are small for heavily paced buyers and increase sharply as \(\alpha_i^*\) approaches one. The point \((1,1)\) follows directly from pay-your-bid pricing: an unpaced buyer obtains zero utility from her assigned \posfppe allocation. Hence, if reoptimizing at the induced prices would yield positive utility, she forgoes all attainable utility relative to that benchmark, and her relative demand gap equals one.

We also study the effect of budget tightness. Holding valuations and market structure fixed, we recompute the \posfppe after doubling every buyer's budget. The LOWESS curves are very similar across the two budget regimes, indicating that the relationship between \(\alpha_i^*\) and the relative loss in attainable utility changes little as budgets are relaxed. In contrast, the orange CDFs show that doubling budgets shifts significantly more buyers toward \(\alpha_i^*=1\), with the unpaced share increasing from \(27\%\) to \(68\%\). Thus, looser budgets primarily affect the distribution of pacing multipliers rather than the conditional relationship between pacing and relative demand gaps. The next experiment extends the comparison to \(B/2\) and \(4B\).

\paragraph{Budget tightness and pacing multipliers.}\label{sec:app-LW-tightness}
We extend the comparison in~\Cref{fig:sec3-alpha-gain} by reconsidering the same 100 market instances under budget scales \(B/2\), \(B\), \(2B\), and \(4B\). 
The left panel of \Cref{fig:app-alpha-gain-budget} reports the mean relative demand gap within bins of the equilibrium pacing multiplier~$\alpha_i^*$. The right panel reports the fraction of buyers in each bin. 
We use binned means rather than LOWESS fits because the support of $\alpha_i^*$ differs significantly across budget regimes.
As in~\Cref{fig:sec3-alpha-gain}, the conditional relationship between \(\alpha_i^*\) and the relative demand gap is very similar across budget regimes: gaps are small for heavily paced buyers and rise sharply as \(\alpha_i^*\) approaches one. At the same time, as shown in the right panel of \Cref{fig:app-alpha-gain-budget}, relaxing budgets shifts significantly more buyers toward \(\alpha_i^*=1\), where demand gaps are largest. Thus, budget tightness primarily affects demand gaps through the distribution of pacing multipliers rather than their conditional relationship with~\(\alpha_i^*\).

\begin{figure}[h]
    \centering
    \includegraphics[width=\textwidth]{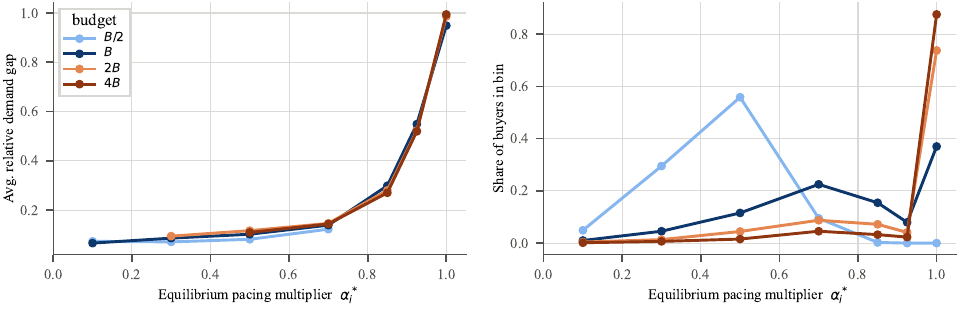}
    \caption{Demand gaps across budget regimes. The left panel reports the mean relative demand gap within bins of the equilibrium pacing multiplier $\alpha_i^*$. The right panel reports the share of buyers in each bin.} 
    \label{fig:app-alpha-gain-budget}
\end{figure}

\subsection{Tightness of the liquid welfare guarantees}\label{app:lw-bound}
We compare the realized liquid welfare loss with the multiplicative and additive guarantees in \Cref{thm:abspe-lw}. For each market, we report the realized relative loss $1-\frac{LW(x^{PE})}{OPT_{LW}},$ together with the normalized additive bound from \Cref{thm:abspe-lw}(2), $\frac{\sum_i B_i\log(1/\alpha_i^*)}{OPT_{LW}}.$ The multiplicative guarantee in \Cref{thm:abspe-lw}(1) corresponds to the uniform relative-loss bound of $1/2$. Since in these simulations, the additive bound is tighter than the multiplicative bound in both environments, we also measure the fraction of the additive bound that is realized.

\begin{figure}[t]
    \centering
    \includegraphics[width=\textwidth]{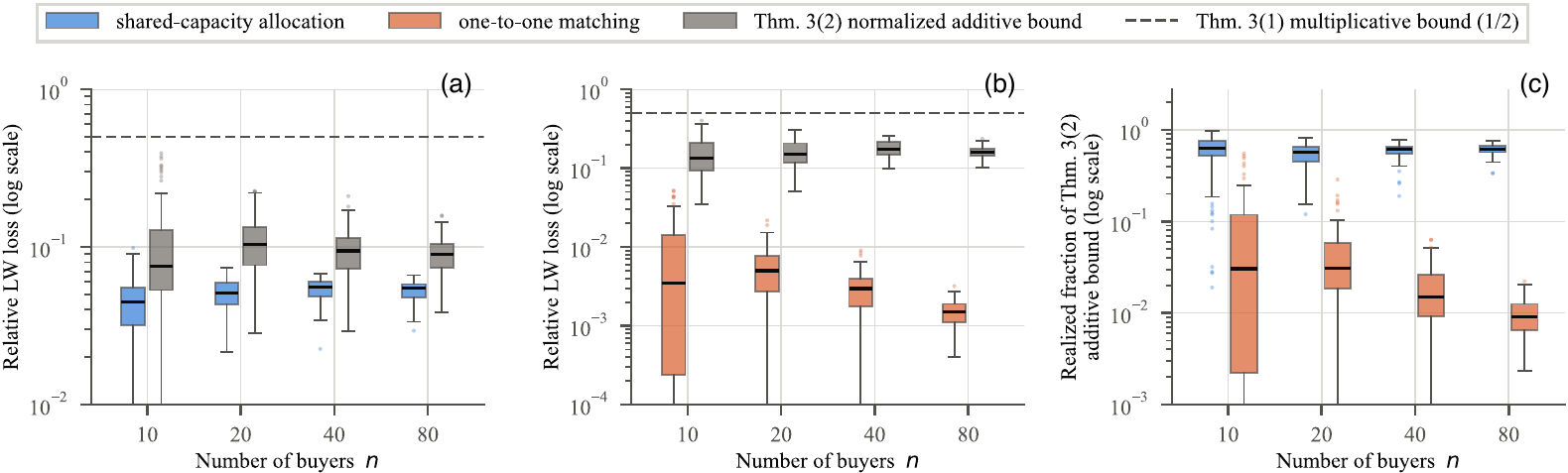}
    \caption{Tightness of the liquid welfare guarantees. Panels (a) and (b) compare realized relative LW losses with the normalized additive bound from \Cref{thm:abspe-lw}(2) for the shared-capacity and one-to-one matching mechanisms, respectively. The dashed line shows the multiplicative bound of 1/2. Panel (c) reports the fraction of the additive bound realized in each mechanism. Values below the displayed range are omitted.}
    \label{fig:lw-bound}
\end{figure}

The results are reported in~\Cref{fig:lw-bound}. Across the $2\times4\times10=800$ simulated markets, the normalized additive bound from Theorem~\ref{thm:abspe-lw}(2) never exceeds \(0.402\). Its median ranges from approximately \(0.075\) to \(0.10\) under shared-capacity allocation and from \(0.13\) to \(0.17\) under one-to-one matching. Realized losses remain below the additive bound in both environments, but are much closer to the bound under shared-capacity allocation than under matching: the fraction of the additive bound realized ranges from approximately \(0.58\) to \(0.63\) under the shared-capacity allocation, compared with only \(0.009\) to \(0.031\) under matching, where it also declines with market size. 

\subsection{PACE in Bid-Maximizing Pay-Your-Bid Environments}\label{app:exps-pace}

Section~\ref{sec:bid-max-pyb} extends the PACE framework to bid-maximizing pay-your-bid mechanisms. The resulting updates require only a bid-maximization oracle over the feasible allocation set. We evaluate their convergence by comparing the resulting pacing multipliers with \abspe multipliers computed independently from the convex-program characterization.

We consider three bid-maximizing pay-your-bid environments: single-item allocation, position auctions, and shared-capacity allocation. Each market has $n=50$ buyers and $m=100$ allocation opportunities. 
In the position-auction environment, there are $s=15$ slots with qualities $q_k=~0.6^{k-1}$. The shared-capacity environment uses the same allocation constraints as in the liquid welfare experiment in Section~\ref{app:lw-bound}. 
Budgets are set to $B_i=\frac{1}{n}\sum_j v_{ij},$ and the initial multipliers are drawn independently from $U(0,1]$. We take $\gamma=0.01$. We generate 100 markets for each environment and run the updates for 300 iterations. We compute $\alpha^*$ by solving the convex program in Proposition~\ref{prop:eg-dual} using CVXPY with CLARABEL and then recovering $\alpha_i^*=\min\{1,B_i/v_i(x_i^*)\}$ from the resulting allocation.

\begin{figure}[h]
    \centering
    \includegraphics[width=\textwidth]{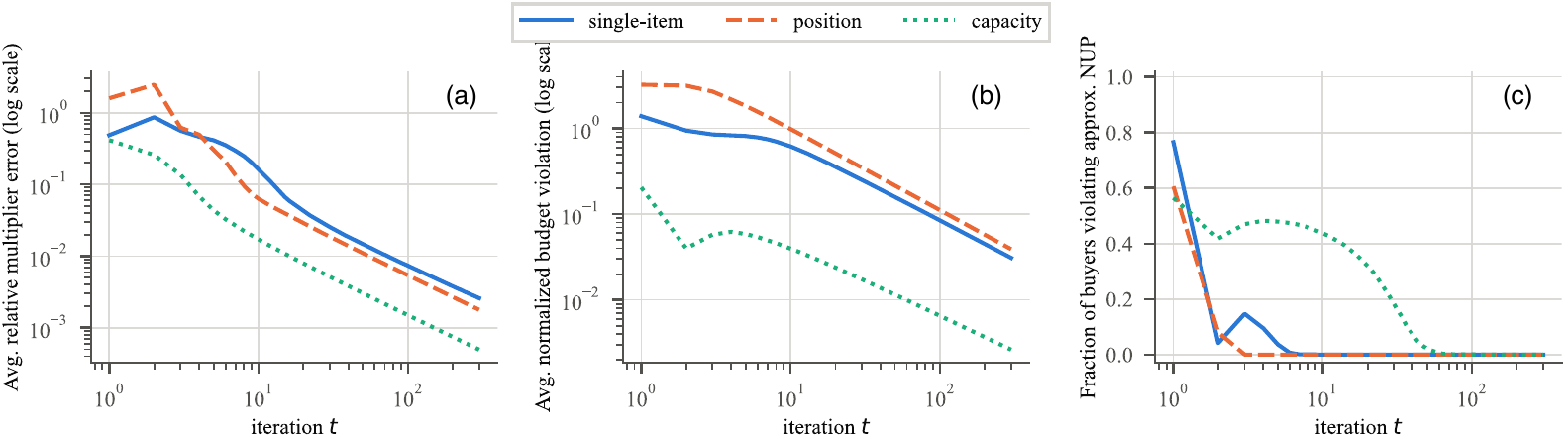}
    \caption{PACE convergence in three bid-maximizing pay-your-bid environments. Panels (a), (b) and (c) report the average relative pacing-multiplier error, normalized budget violation, and fraction of buyers violating approximate NUP, respectively, in pay-your-bid single-item, position-auction, and shared-capacity mechanisms.}
    \label{fig:app-pace}
\end{figure}

\Cref{fig:app-pace} reports three measures of convergence. Panel (a) shows the average relative error in the pacing multipliers, $\frac{1}{n}\sum_i \frac{|\alpha_i^t-\alpha_i^*|}{\alpha_i^*}.$
Panels (b) and (c) report violations of budget feasibility and no-unnecessary-pacing, respectively. 
Since the bid-maximization oracle may choose among multiple optimal allocations, realized spending can fluctuate across iterations. We therefore evaluate budget feasibility and no-unnecessary-pacing using time-averaged spending, $\bar p_i^t=\frac{1}{t}\sum_{\tau=1}^t p_i^\tau.$ Panel (b) reports the average normalized budget violation, $\frac{1}{n}\sum_i\left[\frac{\bar p_i^t-B_i}{B_i}\right]_+,$ while Panel (c) reports the fraction of buyers violating approximate NUP, i.e., buyers with $\bar p_i^t < (1-\gamma)B_i$ and $\alpha_i^t < 1-\gamma.$

The updates approach the convex-program benchmark in all three environments. At $T=300$, the average relative multiplier error is $2.6\times10^{-3}$ in the single-item environment, $1.8\times10^{-3}$ in the position environment, and $4.9\times10^{-4}$ under shared capacity; the fraction of buyers violating approximate no-unnecessary-pacing is zero in all three cases. Thus, the PACE framework continues to perform well when the allocation rule is generalized from single-item auctions to richer bid-maximizing pay-your-bid environments.

\subsection{Additional Results on Pacing Dynamics} \label{app:exps-worst-buyer}

\paragraph{Worst-buyer performance of the last-iterate dynamics.} 
\Cref{fig:app-alogs-max}(a) complements the average budget-violation results in Figure~\ref{fig:sec5-borgs} by reporting the maximum relative budget violation across buyers. This measure corresponds directly to the buyer-level budget-feasibility condition in Theorem~\ref{thm:borgs-conv}. Across the mechanisms and market sizes considered, the worst-buyer violation falls below the target $\gamma=0.1$ after approximately $20\%-36\%$ of the budget-feasibility horizon $T_\gamma^{BF}$. Thus, the fast convergence observed in Figure~\ref{fig:sec5-borgs} also holds at the worst-buyer level.

\begin{figure}[h]
   \centering
    \includegraphics[width=\textwidth]{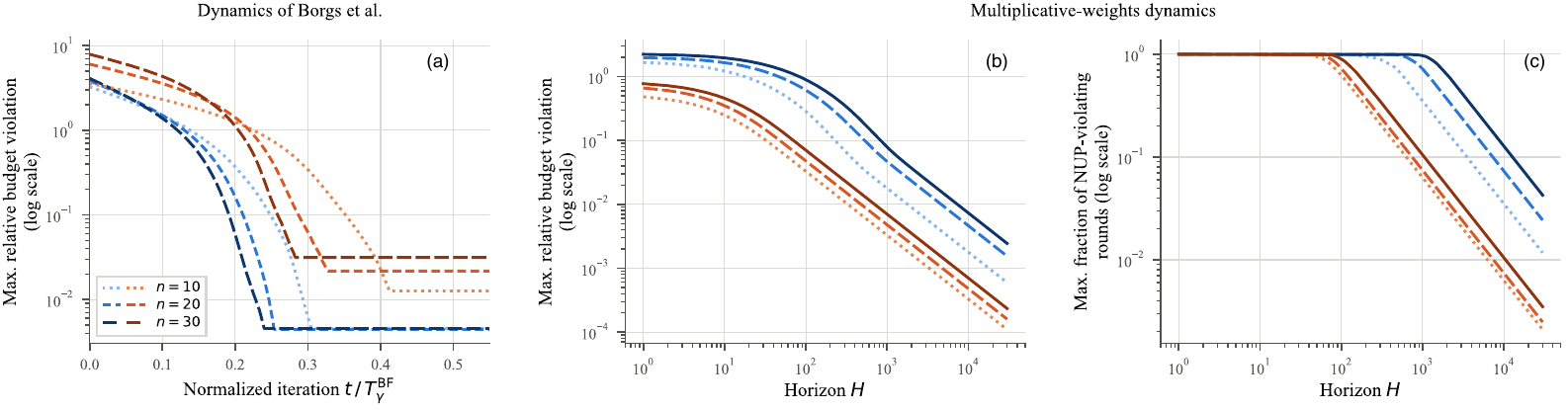}
     \caption{Worst-buyer convergence diagnostics. Panel~(a) reports the largest relative budget violation across buyers for the dynamics of Borgs et al. For the MW dynamics, panel~(b) reports the maximum time-averaged relative budget violation across buyers, and panel~(c) reports the largest fraction of NUP-violating rounds across buyers. Blue curves correspond to proportional sharing and orange curves to first-price position auctions.} 
    \label{fig:app-alogs-max}
\end{figure}

\paragraph{Worst-buyer performance of the multiplicative-weights dynamics.} 
We report the analogous worst-buyer diagnostics for the multiplicative-weights dynamics studied in Figure~\ref{fig:sec5-mw}. Specifically, \Cref{fig:app-alogs-max}(b) shows the maximum relative budget violation across buyers, and \Cref{fig:app-alogs-max}(c) shows the maximum, across buyers, of the fraction of rounds in which that buyer violates approximate no-unnecessary-pacing. Both measures decline with the averaging horizon. At $H=3\times 10^4$, the maximum relative budget violation is below $2.5\times 10^{-3}$ in every specification. Thus, the convergence pattern in Figure~\ref{fig:sec5-mw} also holds at the worst-buyer~level.

\end{document}